\documentclass[runningheads,a4paper]{llncs}
\pdfoutput=1

\usepackage{amssymb}
\usepackage{enumerate}
\usepackage[colorlinks=true]{hyperref}
\usepackage{tikz}
\usepackage{xcolor,latexsym,amsmath,extarrows,alltt}
\usepackage{xspace}
\usepackage{booktabs}
\usepackage{mathtools}
\usepackage{enumitem}
\usepackage{stmaryrd}
\usepackage{microtype}
\usepackage{bussproofs}
\usepackage{multirow}

\newcommand{\N}{\mathbb{N}}
\newcommand{\A}{\mathcal{A}}
\newcommand{\B}{\mathcal{B}}
\newcommand{\F}{\mathcal{F}}
\newcommand{\V}{\mathcal{V}}
\newcommand{\X}{\mathcal{X}}
\newcommand{\Sorts}{\mathcal{S}}
\newcommand{\Constructors}{\mathcal{C}}
\newcommand{\Defineds}{\mathcal{D}}
\newcommand{\Var}{\mathit{Var}}
\newcommand{\Data}{\texttt{Data}}
\newcommand{\VValue}{\mathtt{Value}_{d_1,\dots,d_M}^\prog}
\newcommand{\domain}{\mathtt{dom}}
\newcommand{\Card}{\mathsf{Card}}
\newcommand{\transition}[5]{#1~\displaystyle{\mathop{=\!\!=\!\!\!
\Longrightarrow}^{#2/#3\ #4}}~#5}
\newcommand{\numrep}[1]{[#1]}
\newcommand{\timecomp}[1]{\ensuremath{\textrm{TIME}\left(#1\right)}}
\newcommand{\exptime}[1]{\mathsf{EXP}^{#1}\mathsf{TIME}}
\newcommand{\nexptime}[1]{\mathsf{NEXP}^{#1}\mathsf{TIME}}
\newcommand{\expspace}[1]{\ensuremath{\textrm{EXP}^{#1}\textrm{SPACE}}}
\newcommand{\logtime}{\mathsf{LOGTIME}}
\newcommand{\logspace}{\mathsf{LOGSPACE}}
\newcommand{\asort}{\iota}
\newcommand{\asortorpair}{\kappa}
\newcommand{\atype}{\sigma}
\newcommand{\btype}{\tau}
\newcommand{\ctype}{\pi}
\newcommand{\aindex}{l}
\newcommand{\bindex}{p}
\newcommand{\cindex}{q}
\newcommand{\prog}{\mathsf{p}}
\newcommand{\eprog}{\mathsf{p}'}

\newcommand{\arity}{\mathtt{arity}_\prog}
\newcommand{\vdashcall}{\vdash^{\mathtt{call}}}
\newcommand{\vdashif}{\vdash^{\mathtt{if}}}
\newcommand{\vvdash}{\Vdash}
\newcommand{\vvdashcall}{\Vdash^{\mathtt{call}}}

\newcommand{\treeroot}{\mathit{root}}
\newcommand{\app}[2]{#1\ #2}
\newcommand{\apps}[3]{#1\ #2 \cdots #3}

\newcommand{\symb}[1]{\mathtt{#1}}
\newcommand{\interpret}[1]{\llbracket #1 \rrbracket_\B}
\newcommand{\pinterpret}[1]{\langle\!| #1 |\!\rangle_\B}
\newcommand{\numinterpret}[1]{\langle #1 \rangle}
\newcommand{\down}[2]{#1\!\Downarrow\!#2}
\newcommand{\arrtype}{\Rightarrow}
\newcommand{\arrr}{\to}
\newcommand{\too}{\Rightarrow}
\newcommand{\subtermeq}{\unlhd}
\newcommand{\supterm}{\rhd}
\newcommand{\suptermeq}{\unrhd}
\newcommand{\ssupseteq}{\sqsupseteq'}
\newcommand{\consistent}[2]{#1\wr#2}
\newcommand{\blank}{\textbf{\textvisiblespace}}
\newcommand{\nul}{\symb{0}}

\newcommand{\nil}{\symb{[]}}
\newcommand{\cons}{\symb{::}}
\newcommand{\strue}{\symb{true}}
\newcommand{\sfalse}{\symb{false}}
\newcommand{\suc}{\symb{s}}
\newcommand{\bits}{\symb{list}}
\newcommand{\bool}{\symb{bool}}
\newcommand{\nat}{\symb{nat}}
\newcommand{\msymbol}{\symb{symbol}}
\newcommand{\mstate}{\symb{state}}
\newcommand{\mdirec}{\symb{direc}}
\newcommand{\zero}[1][\pi]{\symb{zero}_{#1}}
\newcommand{\pred}[1][\pi]{\symb{pred}_{#1}}
\newcommand{\seed}[1][\pi]{\symb{seed}_{#1}}
\newcommand{\numtype}[1][\pi]{\alpha_{#1}}
\newcommand{\exx}{\symb{x}}
\newcommand{\linear}{\langle 1,1\rangle}
\newcommand{\pol}{{\langle a,b\rangle}}
\newcommand{\epi}{{\symb{e}[\pi]}}
\newcommand{\er}[1][\pi]{\psi[{#1}]}
\newcommand{\eelin}{{\psi[\psi[\linear]]}}

\newcommand{\ifte}[3]{\mathtt{if} \, #1 \, \mathtt{then} \, #2 \,
                    \mathtt{else} \, #3}
\newcommand{\choice}{\mathtt{choose}}
\newcommand{\identifier}[1]{\mathtt{#1}}
\newcommand{\progeval}[2]{\llbracket #1 \rrbracket(#2)}
\newcommand{\progevaluation}{\llbracket \prog \rrbracket(d_1,\dots,d_M)}
\newcommand{\progresult}{\llbracket \prog \rrbracket(d_1,\dots,d_M)
\mapsto b}
\newcommand{\pclass}{\mathsf{P}}
\newcommand{\npclass}{\mathsf{NP}}
\newcommand{\expclass}{\mathsf{EXP}}
\newcommand{\logspaceclass}{\mathsf{L}}
\newcommand{\nlogspaceclass}{\mathsf{NL}}
\newcommand{\pspaceclass}{\mathsf{PSPACE}}
\newcommand{\elementary}{\mathsf{ELEMENTARY}}
\newcommand{\complexityfun}{h}
\newcommand{\OO}{\mathcal{O}}
\newcommand{\typeorder}[1]{\ensuremath{\mathit{ord}\!\left(#1\right)}}

\newtheorem{algorithm}[lemma]{Algorithm}

\begin{document}

\thispagestyle{empty}
\definecolor{comment}{rgb}{0.92, 0.92, 0.92}
\colorbox{comment}{\parbox{0.9\textwidth}{
  This paper is a pre-print of the paper
  \emph{The Power of Non-Determinism in Higher-Order Implicit
  Complexity} which has been accepted for publication at the
  \emph{European Symposium for Programming} (ESOP 2017). \\
  \ \\
  The text is (almost) identical to the published version, but
  the present work includes an appendix containing full proofs of all
  the results in the paper.
}}
\setcounter{page}{0}

\newpage

\mainmatter

\title{The Power of Non-Determinism in Higher-Order Implicit Complexity
\thanks{The authors are supported by the Marie Sk{\l}odowska-Curie
action ``HORIP'', program H2020-MSCA-IF-2014, 658162 and by the Danish
Council for Independent Research Sapere Aude grant ``Complexity via
Logic and Algebra'' (COLA).}}
\subtitle{Characterising Complexity Classes using Non-deterministic
Cons-free Programming}

\author{Cynthia Kop
\and Jakob Grue Simonsen
}
\authorrunning{C. Kop and J. Simonsen}
\institute{
Department of Computer Science, University of Copenhagen (DIKU) \\
\email{kop@di.ku.dk}
\quad\quad\quad
\email{simonsen@di.ku.dk}
}

\maketitle

\begin{abstract}
We investigate the power of non-determinism in purely functional
programming languages with higher-order types.
Specifically, we consider \emph{cons-free} programs of varying data
orders, equipped with explicit non-deterministic choice.  Cons-freeness
roughly means that data constructors cannot occur in
function bodies and all manipulation of storage space thus has to
happen indirectly using the call stack.

\hspace{2ex}
While cons-free programs have previously been used by several authors
to characterise complexity classes, the work on
\emph{non-deterministic} programs has almost exclusively considered
programs of data order $0$.
Previous work has shown that adding explicit non-determinism to
cons-free programs taking data of order 0 does not increase
expressivity
; we prove that this---dramatically---is not the case for higher data orders: adding
non-determinism to programs with data order at least $1$
allows for a characterisation of the entire class of elementary-time
decidable sets.

\hspace{2ex}
Finally we show how, even 
with
non-deterministic choice, the original hierarchy of characterisations
is restored by imposing different restrictions.

\end{abstract}

\section{Introduction}

\emph{Implicit complexity} is, roughly, the study of how to create
bespoke programming languages that allow the programmer to write
programs which are guaranteed to (a) \emph{only} solve problems
within a certain complexity class (e.g., the class of polynomial-time
decidable sets of binary strings), and (b) to be able to solve
\emph{all} problems in this class.
When equipped with an efficient execution engine, the programs of
such a language may themselves be guaranteed to run within the
complexity bounds of the class (e.g., run in polynomial time), and
the plethora of means available for analysing programs devised by the
programming language community means that methods from outside
traditional complexity theory can conceivably be brought to
bear on open problems in computational complexity.

One successful approach to implicit complexity is to syntactically
constrain the programmer's ability to create new data structures.  In
the seminal
paper~\cite{jon:01}, Jones introduces \emph{cons-free programming}.
Working with a small functional programming language, cons-free
programs are 
\emph{read-only}: recursive data cannot
be created or altered (beyond taking sub-expressions), only read from
input.  By imposing further restrictions on \emph{data order}
(i.e., order 0 = integers, strings; order 1 = functions on data of
order 0; etc.) and recursion scheme (e.g., full/tail/primitive
recursion), classes of cons-free programs turn out to characterise
various deterministic classes in the time and space hierarchies of
computational complexity.

However, Jones' language is deterministic and, perhaps as a
result, his characterisations concern only deterministic complexity
classes. It is tantalising to consider the method in a
non-deterministic setting: could adding non-deterministic
choice to Jones' language increase its
expressivity; for example, from $\pclass$ to $\npclass$?

The immediate answer is \emph{no}: following Bonfante~\cite{bon:06},
adding a non-deterministic choice operator to cons-free programs with
data order $0$ makes no difference in expressivity---deterministic or
not, 
they
characterise $\pclass$.  However, the 
details are subtle and depend 
on other features of the
language;
when only primitive recursion is allowed,
non-determinism \emph{does} increase expressivity from
$\logspaceclass$ to $\nlogspaceclass$ \cite{bon:06}.

While many authors consider the expressivity of 
higher types, the interplay of higher types and non-determinism is
not fully understood. 
Jones obtains
several hierarchies of deterministic complexity classes by increasing
data orders~\cite{jon:01},
but these hierarchies have
at most an exponential increase 
between levels.  Given
the expressivity added by non-determinism, it is \emph{a priori} not
evident that similarly ``tame'' hierarchies would arise in the
non-deterministic setting. 

The purpose of the present paper is to investigate the power of
\emph{higher-order} (cons-free) programming 
to characterise
complexity classes. The main surprise is that while non-determinism
does not add expressivity for first-order programs, the combination of
second-order (or higher) programs and non-determinism characterises
the full class of elementary-time decidable sets---and increasing the
order beyond second-order programs does not further increase
expressivity.
However, we will also show that there are simple changes to
the restrictions that allow us to obtain a hierarchy of
characterisations as in the deterministic setting.

Proofs for the results in this paper are all available in the appendix.

\subsection{Overview and contributions}

\begin{figure}[!ht]
\begin{tabular}{c|c|c|c|c}
& \textbf{data order 0} &
\textbf{data order 1} &
\textbf{data order 2} &
\textbf{data order 3} 
\\
\cline{1-5}
\vphantom{$2^{x^y}$}
\textbf{cons-free} &
$\pclass =$ &
$\expclass =$ &
\multirow{2}{*}{$\exptime{2}$} &
\multirow{2}{*}{$\exptime{3}$} 
\\
\textbf{deterministic} &
$\exptime{0}$ &
$\exptime{1}$ &
& \\
\cline{1-5}
\vphantom{$2^{x^y}$}
\textbf{cons-free} &
$\logspaceclass$ &
$\pspaceclass$ & & \\
\textbf{tail-recursive} &
$=$ &
$=$ &
$\expspace{1}$ &
$\expspace{2}$ 
\\
\textbf{deterministic} &
$\expspace{-1}$ &
$\expspace{0}$ &
& \\
\cline{1-5}
\vphantom{$2^{x^y}$}
\textbf{cons-free} &
$\logspaceclass$ &
$\pclass$ & $\pspaceclass$ & $\expclass$ \\
\textbf{primitive recursive} &
$=$ & $=$ & $=$ & $=$
\\
\textbf{deterministic} &
$\expspace{-1}$ &
$\exptime{0}$ & $\expspace{0}$
& $\exptime{1}$ \\
\cline{1-5}

\multicolumn{5}{c}{} \\
\multicolumn{5}{c}{
\parbox[c]{\textwidth}{
\emph{The characterisations obtained in~\cite{jon:01}, transposed to the
more permissive language used here.  This list (and the one below)
should be imagined as extending infinitely to the right.
The ``limit'' for all rows (i.e., all finite data orders allowed)
characterises $\elementary$, the class of elementary-time decidable
sets.
}}}\\
\multicolumn{5}{c}{\vspace{0.1cm}} \\

& \textbf{data order 0} &
\textbf{data order 1} &
\textbf{data order 2} &
\textbf{data order 3} 
\\
\cline{1-5}
\vphantom{$2^{x^y}$}
\textbf{cons-free} &
$\pclass$ &
$\elementary$ &
$\elementary$ &
$\elementary$ 
\\
\cline{1-5}
\textbf{cons-free} &
$\pclass =$ &
$\expclass =$ &
\multirow{2}{*}{$\exptime{2}$} &
\multirow{2}{*}{$\exptime{3}$} 
\\
\textbf{unitary variables} &
$\exptime{0}$ &
$\exptime{1}$ &
\\
\cline{1-5}
\multicolumn{5}{c}{} \\
\multicolumn{5}{c}{
\parbox[c]{\textwidth}{
\emph{The characterisations obtained by allowing non-deterministic choice.
As above, the ``limit'' where all data orders are allowed
characterises $\elementary$ (for both rows).}
\\
}} \\
& \scriptsize\textbf{arrow depth 0} &
\scriptsize\textbf{arrow depth 1} &
\scriptsize\textbf{arrow depth 2} &
\scriptsize\textbf{arrow depth 3} 
\\
\cline{1-5}
\textbf{cons-free} &
$\pclass$ &
$\elementary$ &
$\elementary$ &
$\elementary$ 
\\
\cline{1-5}
\multicolumn{5}{c}{} \\
\multicolumn{5}{c}{
\parbox[c]{\textwidth}{
\emph{The characterisations obtained by allowing non-deterministic choice and
considering \emph{arrow depth} as the variable factor rather than
data order}}
} \\
\end{tabular}
\vspace{-6pt}
\caption{Overview of the results discussed or obtained in this paper.}
\vspace{-12pt}
\label{fig:overview}
\end{figure}

We define a purely functional programming language with
non-deterministic choice and, following Jones~\cite{jon:01},
consider the restriction to \emph{cons-free} programs.

Our results are summarised in Figure~\ref{fig:overview}.  For
completeness, we have also included the results from~\cite{jon:01};
although the language used there is slightly more syntactically
restrictive than ours, the results easily generalise provided we limit
interest to \emph{deterministic} programs, where the
$\choice$ operator is not used.  As the technical machinations
involved to procure the results for a language with full recursion are
already intricate and lengthy, we have not yet considered the
restriction to tail- or primitive recursion in the non-deterministic
setting.

Essentially,
our paper has two major contributions:
(a) we show that previous observations about the increase in
expressiveness when adding non-determinism change dramatically at
higher types, and (b) we provide
two characterisations of the $\exptime{}$ hierarchy using
a non-deterministic language
---which
may provide a basis
for future characterisation of
common non-deterministic classes as well.

Note that (a) is highly surprising:
As evidenced by early work of Cook \cite{coo:71} merely adding full
non-determinism to a restricted (i.e., non-Turing complete)
computation model may result in it still characterising a
\emph{deterministic} class of problems.  This 
also holds true
for cons-free programs with non-determinism,
as shown in different settings by Bonfante \cite{bon:06}, by de Carvalho and
Simonsen \cite{car:sim:14},
and by Kop and Simonsen \cite{DBLP:conf/rta/KopS16}, all resulting only in
characterisations of deterministic classes such as $\pclass$.
With the exception of \cite{DBLP:conf/rta/KopS16}, all of the above
attempts at adding non-determinism consider data order at most
$0$, and one would expect few changes when passing to higher data
orders.  This turns out to be patently false as simply increasing to
data order $1$ already results in
an explosion of expressive power.

\subsection{Overview of the ideas in the paper}

Cons-free programs (Definition~\ref{def:consfree}) are, roughly,
functional programs where function bodies are allowed to contain
constant data and substructures of the function arguments, but \emph{no data constructors}---e.g.,
clauses $\symb{tl}\ (x\cons xs) = xs$ and $\symb{tl}\ \nil = \nil$
are both allowed, but $\symb{append}\ (x\cons xs)\ ys = x\cons
(\symb{append}\ xs\ ys)$ is not.\footnote{The formal definition is slightly
more liberal to support easier implementations using pattern-matching,
but the ideas remain the same.}
This restriction severely limits expressivity, as it means no new data
can be created.

A key idea in Jones' original work on cons-free programming is \emph{counting}: expressions which
represent numbers and functions to calculate with them.
It is not in general possible to represent numbers in the
usual unary way as $\nul$, $\suc\ \nul$, $\suc\ (\suc\ \nul)$,
etc., or as lists of bits---since in a cons-free program these
expressions cannot be built unless they already occur in the
input---but counting up to limited bounds \emph{can} be achieved by other tricks.
By repeatedly simulating a single step of a Turing Machine up to such bounds, Jones shows that any decision problem in $\exptime{K}$ can be decided using a cons-free
program (\cite{jon:01} and Lemma~\ref{lem:deterministic:simulate}).

The core insight in the present paper is that in the presence of
non-determinism, an expression of type $\atype \arrtype \btype$
represents a \emph{relation} between expressions of type $\atype$ and
expressions of type $\btype$ rather than a \emph{function}.
While the number of functions for a given type is exponential in the
order of that type, the number of relations is exponential in the
depth of arrows occurring in it.  We exploit this (in
Lemma~\ref{lem:nondetmodule}) by counting up to arbitrarily high
numbers using only first-order data.
This observation also suggest that by limiting the \emph{arrow depth}
rather than the \emph{order of types}, the increase in expressive power
disappears (Theorem~\ref{thm:arrowdepth}). 

Conversely, we also provide an
algorithm to compute the output of cons-free programs 
potentially much faster than the program's own running time, by using a
tableaux to store results.  Although similar to Jones' ideas, our
proof style deviates to easily support both non-deterministic
and deterministic programs.

\subsection{Related work}

The creation of programming languages that characterise complexity
classes has been a research area since Cobham's work in the 1960ies,
but saw rapid development only after similar advances in the related
area of \emph{descriptive complexity} (see, e.g.,
\cite{Immerman99descriptivecomplexity}) in the 1980ies
and Bellantoni and Cook's work on characterisations of $\pclass$
\cite{DBLP:journals/cc/BellantoniC92} using constraints on recursion
in a purely functional language with programs reminiscent of classic
recursion theoretic functions. Following Bellantoni and Cook,
a number of authors obtained programming languages by 
constraints on
recursion, and under a plethora of names (e.g., \emph{safe},
\emph{tiered} or \emph{ramified} recursion, see
\cite{Clote:handbook,DalLago2012} for overviews), and this area
continues to be active. The main difference with our work is that we
consider full recursion in all variables, but place syntactic
constraints on the function bodies (both cons-freeness and
unitary variables).
Also,
as in traditional complexity theory we 
consider decision problems (i.e., what \emph{sets} can be decided by
programs), whereas much research in implicit complexity considers
functional complexity (i.e., what 
\emph{functions} can be computed).

Cons-free programs, combined with various limitations on recursion,
were introduced by Jones \cite{jon:01}, building on ground-breaking
work by Goerdt
\cite{DBLP:journals/tcs/Goerdt92a,DBLP:journals/iandc/Goerdt92}, 
and have been studied by a number of authors (see, e.g., \cite{DBLP:conf/icalp/Ben-AmramP98,bon:06,DBLP:journals/njc/KristiansenV05,KRISTIANSEN2004139}).
The main difference with our work is that we consider full recursion
with full non-determinism, but impose constraints not present in the
previous literature.

Characterisation of non-deterministic complexity classes via
programming languages remains a largely unexplored area.
Bellantoni
obtained a characterisation of $\npclass$ in his dissertation
\cite{Bellantoni:thesis} using similar approaches as
\cite{DBLP:journals/cc/BellantoniC92}, but at the cost of having a
minimisation operator (as in recursion theory), a restriction later
removed by Oitavem \cite{DBLP:journals/apal/Oitavem11}. A general
framework for implicitly characterising
a larger hierarchy of non-deterministic classes remains an open problem.

\section{A purely functional, non-deterministic,
call-by-value programming language}

We define a simple call-by-value programming language with explicit
non-deterministic choice.  This generalises Jones' toy language
in~\cite{jon:01} by supporting different types and pattern-matching
as well as non-determinism.
The more permissive language actually \emph{simplifies} proofs
and examples, since we do not need to encode all data as boolean
lists, and have fewer special cases.

\subsection{Syntax}

We consider programs defined by the syntax in Figure~\ref{fig:syntax}

\begin{figure}[!ht]
\vspace{-12pt}
\begin{center}
\fbox{
\begin{tabular}{rcl}
$\prog \in \texttt{Program}$ & ::= & $\rho_1$ $\rho_2$ $\dots$
$\rho_N$ \\
$\rho \in \texttt{Clause}$ & ::= &
$\apps{\identifier{f}}{\ell_1}{\ell_k} = s$ \\
$\ell \in \texttt{Pattern}$ & ::= &
$x \mid \apps{\identifier{c}}{\ell_1}{\ell_m}$ \\
$s,t \in \texttt{Expr}$ & ::= & $x \mid \identifier{c} \mid
\identifier{f} \mid \ifte{s_1}{s_2}{s_3} \mid
\apps{\choice}{s_1}{s_n} \mid (s,t) \mid s\ t$ \\
$x,y \in \V$ & ::= & identifier \\
$\identifier{c} \in \Constructors $ & ::= &
identifier disjoint from $\V$
\quad (we assume $\{\strue,\sfalse\} \subseteq \Constructors$) \\
$\identifier{f},\identifier{g} \in \Defineds$ & ::= &
identifier disjoint from $\V$ and $\Constructors$ \\
\end{tabular}
}
\end{center}
\vspace{-12pt}
\caption{Syntax}\label{fig:syntax}
\vspace{-12pt}
\end{figure}

We call elements of $\V$ \emph{variables}, elements of $\Constructors$
\emph{data constructors} and elements of $\Defineds$ \emph{defined
symbols}.  The \emph{root} of a clause $\apps{\identifier{f}}{
\ell_1}{\ell_k} = s$ is the defined symbol $\identifier{f}$. The
\emph{main function} $\identifier{f}_1$ of the program is the root of
$\rho_1$.
We denote $\Var(s)$ for the set of variables occurring in an
expression $s$.  An expression $s$ is \emph{ground} if $\Var(s) =
\emptyset$.
Application is left-associative, i.e., $s\ t\ u$ should be read
$(s\ t)\ u$.

\begin{definition}
For expressions $s,t$, we say that \emph{$t$ is a sub-expression
of $s$}, notation $s \suptermeq t$, if this can be derived using the
clauses:
\[
\begin{array}{rclclrclcl}
s & \suptermeq & t & \text{if} & s = t\ \text{or}\ s \supterm t \\
(s_1,s_2) & \supterm & t & \text{if} & s_1 \suptermeq t\ \text{or}\ 
s_2 \suptermeq t\quad &
\ifte{s_1}{s_2}{s_3} & \supterm & t & \text{if} & s_i \suptermeq t\ 
\text{for some}\ i \\
s_1\ s_2 & \supterm & t & \text{if} & s_1 \supterm t\ \text{or}\ 
s_2 \suptermeq t &
\apps{\choice}{s_1}{s_n}& \supterm & t & \text{if} & s_i \suptermeq t
\ \text{for some}\ i \\
\end{array}
\]
Note: the head $s$ of an application $s\ t$ is \emph{not}
\pagebreak
considered a sub-expression of $s\ t$.
\end{definition}

Note that the programs we consider have no pre-defined data
structures like integers: these may be encoded using inductive data
structures in the usual way.


\begin{example}\label{ex:successor}
Integers can be encoded as bitstrings of unbounded length:
$\Constructors \supseteq \{ \sfalse,\strue,\cons,\nil \}$.  Here,
$\cons$ is
considered infix and right-associative, and $\nil$ denotes the end of
the string.  Using little endian, $6$ is
encoded by
$\sfalse\cons\strue\cons\strue\cons\nil$ as well as
$\sfalse\cons\strue\cons\strue\cons\sfalse\cons\sfalse\cons\nil$.
We for instance have $\strue\cons (\symb{succ}\ xs)
\linebreak
\suptermeq xs$ (for $xs \in \V$).
The program below imposes $\Defineds = \{\symb{succ}\}$:
\[
\begin{array}{lcl}
\symb{succ}\ \nil = \strue\cons\nil & \quad\quad &
\symb{succ}\ (\sfalse\cons xs) = \strue\cons xs \\
& & \symb{succ}\ (\strue\cons xs) = \sfalse\cons (\symb{succ}\ xs) \\
\end{array}
\]
\end{example}

\subsection{Typing}

Programs have explicit simple types without polymorphism, with the
usual definition of type order $\typeorder{\atype}$; this is formally
given in Figure~\ref{fig:order}.

\begin{figure}[!ht]
\vspace{-12pt}
\begin{center}
\fbox{
\begin{minipage}{0.4\textwidth}
\begin{tabular}{rcl}
$\asort \in \Sorts$ & ::= & sort identifier \\
$\atype,\btype \in \texttt{Type}$ & ::= & $\asort \mid \atype \times
\btype \mid \atype \arrtype \btype$ \\
\\
\end{tabular}
\end{minipage}
\begin{minipage}{0.5\textwidth}\vspace{-12pt}
$$
\begin{array}{rcl}
\typeorder{\asort} & = & 0\ \ \text{for}\ \asort \in \Sorts \\
\typeorder{\atype \times \btype} & = & \max(\typeorder{\atype},
\typeorder{\btype}) \\
\typeorder{\atype \arrtype \btype} & = & \max(\typeorder{\atype}+1,
\typeorder{\btype}) \\
\end{array}
$$
\end{minipage}
}
\end{center}
\vspace{-12pt}
\caption{Types and type orders}
\vspace{-12pt}
\label{fig:order}
\end{figure}

The (finite) set $\Sorts$ of sorts is used to type 
atomic data such as bits; we assume $\bool \in \Sorts$.
The function arrow $\arrtype$ is considered right-associative.
Writing $\asortorpair$ for a sort or a pair type $\atype \times
\btype$, any type can be uniquely presented in the form
$\atype_1 \arrtype \dots \arrtype \atype_m \arrtype \asortorpair$.
We will limit interest to \emph{well-typed, well-formed} programs:

\begin{definition}\label{def:well-typed}
A program $\prog$ is \emph{well-typed} if there is an assignment $\F$
from $\Constructors \cup \Defineds$ to the set of simple types such
that:
\begin{itemize}
\item the main function $\identifier{f}_1$ is assigned a type
$\asortorpair_1 \arrtype \dots \arrtype \asortorpair_M \arrtype
\asortorpair$, with $\typeorder{\asortorpair_i} = 0$ for $1 \leq
i \leq M$ and also $\typeorder{\asortorpair} = 0$
\item data constructors $\identifier{c} \in \Constructors$ are
assigned a type
$\asortorpair_1 \arrtype \dots \arrtype \asortorpair_m \arrtype
\asort$ with $\asort \in \Sorts$ and $\typeorder{\asortorpair_i} =
0$ for $1 \leq i \leq m$
\item for all clauses $\apps{\identifier{f}}{\ell_1}{\ell_k} = s \in
\prog$, the following hold:
\begin{itemize}
\item $\Var(s) \subseteq \Var(\apps{\identifier{f}}{\ell_1}{\ell_k})$
  and
  each variable occurs only once in
  $\apps{\identifier{f}}{\ell_1}{\ell_k}$;
\item there exist a type environment $\Gamma$ mapping
  $\Var(\apps{\identifier{f}}{\ell_1}{\ell_k})$ to simple types, and a
  simple type $\atype$, such that both
  $\apps{\identifier{f}}{\ell_1}{\ell_k} : \atype$ and $s : \atype$
  using the rules in Figure~\ref{fig:typing}; we call $\atype$ the
  type of the clause.
\end{itemize}
\end{itemize}
\end{definition}

\begin{figure}[!htpb]
\vspace{-12pt}
\begin{center}
\fbox{
\begin{minipage}{0.95\linewidth}
\begin{minipage}{0.3\linewidth}
\begin{prooftree}
\AxiomC{}
\RightLabel{if $a : \atype \in \Gamma \cup \F$}
\UnaryInfC{$a : \atype$}
\end{prooftree}
\end{minipage}
\begin{minipage}{0.3\linewidth}
\begin{prooftree}
\AxiomC{$s : \atype$}
\AxiomC{$t : \btype$}
\BinaryInfC{$(s,t) : \atype \times \btype$}
\end{prooftree}
\end{minipage}
\begin{minipage}{0.3\linewidth}
\begin{prooftree}
\AxiomC{$s : \atype \arrtype \btype$}
\AxiomC{$t : \atype$}
\BinaryInfC{$s\ t : \btype$}
\end{prooftree}
\end{minipage}

\begin{minipage}{0.45\linewidth}
\begin{prooftree}
\AxiomC{$s_1 : \bool$}
\AxiomC{$s_2 : \atype$}
\AxiomC{$s_3 : \atype$}
\TrinaryInfC{$\ifte{s_1}{s_2}{s_3} : \atype$}
\end{prooftree}
\end{minipage}
\begin{minipage}{0.45\linewidth}
\begin{prooftree}
\AxiomC{$s_1 : \atype$}
\AxiomC{$\dots$}
\AxiomC{$s_n : \atype$}
\TrinaryInfC{$\apps{\choice}{s_1}{s_n} : \atype$}
\end{prooftree}
\end{minipage}

\end{minipage}
}
\end{center}
\vspace{-8pt}
\caption{Typing (for fixed $\F$ and $\Gamma$, see Definition \ref{def:well-typed})}
\label{fig:typing}
\vspace{-12pt}
\end{figure}


Note that this definition does not allow for polymorphism:
there is a single type assignment $\F$ for the full program.  The
assignment $\F$ also forces a unique choice for the type environment
$\Gamma$ of variables in each clause.  Thus, we may speak of
\emph{the} type of an expression in a clause without risk of
confusion.

\begin{example}\label{ex:successortype}
The program of Example~\ref{ex:successor} is typed using $\F = \{
\sfalse : \bool, \strue : \bool, \nil : \bits, \cons : \bool \arrtype \bits
\arrtype \bits, \symb{succ} : \bits \arrtype \bits \}$.
As all argument and output types have order $0$, the variable
restrictions are satisfied and all clauses can be typed using $\Gamma
= \{ xs : \bits \}$, the program is well-typed.
\end{example}

\begin{definition}
A program $\prog$ is \emph{well-formed} if it is well-typed, and
moreover:
\begin{itemize}
\item data constructors are always fully applied: for all
$\identifier{c} \in \Constructors$ with
$\identifier{c} : \asortorpair_1 \arrtype \dots
\arrtype \asortorpair_m \arrtype \asort \in \F$: if a sub-expression
$\apps{\identifier{c}}{t_1}{t_n}$ occurs in any clause, then $n = m$;
\item the number of arguments to a given defined symbol is fixed:
if $\apps{\identifier{f}}{\ell_1}{\ell_k} = s$ and
$\apps{\identifier{f}}{\ell_1'}{\ell_n'} = t$ are both in $\prog$,
then $k = n$; we let $\arity(\identifier{f})$ denote $k$.
\end{itemize}
\end{definition}

\begin{example}\label{ex:successorform}
The program of Example~\ref{ex:successor} is well-formed, and
$\arity(\symb{succ}) = 1$.

However, the program would not be well-formed if the clauses below
were added,
as here the defined symbol $\symb{or}$ does not have a consistent
arity.
\[
\begin{array}{rclcrclcrcl}
\symb{id}\ x & = & x & \quad\quad &
\symb{or}\ \strue\ x & = & \strue & \quad\quad &
\symb{or}\ \sfalse & = & \symb{id} \\
\end{array}
\]
\end{example}

\begin{remark}\label{remark:constructor}
Data constructors must (a) have a sort as output type (\emph{not} a
pair), and (b) occur only fully applied.
This is consistent with typical functional programming languages,
where sorts and constructors are declared with a grammar such as: \\
\begin{tabular}{rcl}
$\mathit{sdec} \in \mathtt{SortDec}$ & ::= &
$\symb{data}\ \asort = \mathit{cdec}_1 \mid \dots \mid
\mathit{cdec}_n$ \\
$\mathit{cdec} \in \mathtt{ConstructorDec}$ & ::= &
$\identifier{c}\ \atype_1\ \cdots\ \atype_m$
\end{tabular} \\
In addition, we require that the arguments to data constructors have
type order $0$.  This is not standard in functional programming, but
is the case in~\cite{jon:01}.  We limit interest to such constructors
because, practically, these are the only ones which can be used in a
\emph{cons-free} program (as we will discuss in
Section~\ref{sec:consfree}).
\end{remark}

\begin{definition}
A program has \emph{data order $K$} if all clauses can be typed using
type environments $\Gamma$ such that, for all $x : \atype \in \Gamma$:
$\typeorder{\atype} \leq K$.
\end{definition}



\begin{example}\label{ex:hocount}
We consider a higher-order program, operating on the same data
constructors as Example~\ref{ex:successor}; however, now we encode
numbers using \emph{functions}:
\[
\begin{array}{l}
\symb{fsucc}\ F\ \nil = \ifte{\:F\ \nil\:}{\:\symb{set}\ F\ []\ \sfalse\:}{
  \:\symb{set}\ F\ \nil\ \strue\:} \\
\symb{fsucc}\ F\ xs = \ifte{\:F\ xs\:}{
  \:\symb{fsucc}\ (\symb{set}\ F\ xs\ \sfalse)\ (\symb{tl}\ xs)\\
  \phantom{\symb{fsucc}\ F\ xs = }
  }{\:\symb{set}\ F\ xs\ \strue\:} \\
\symb{set}\ F\ \mathit{val}\ xs\ ys =
  \ifte{\:\symb{eqlen}\ xs\ ys\:}{\:val\:}{\:F\ ys} \\
\symb{tl}\ (x\cons xs) = xs
\phantom{ABEfg}
\symb{eqlen}\ (x\cons xs)\ (y\cons ys) = \symb{eqlen}\ xs\ ys \\
\symb{eqlen}\ \nil\ \nil = \strue
\phantom{ABE}
\symb{eqlen}\ xs\ ys = \sfalse \\
\end{array}
\]

Only one typing is possible, with $\symb{fsucc} : (\bits
\arrtype \bool) \arrtype \bits \arrtype \bits \arrtype \bool$;
therefore, $F$ is always typed $ \bits \arrtype \bool$---which has
type order $1$---and all other variables with a type of order $0$.
Thus, this program has data order $1$.

To explain the program: we use boolean lists as \emph{unary} numbers
of a limited size; assuming that (a) $F$ represents a
bitstring of length $N+1$, and (b) $\mathit{lst}$ has length $N$, the
successor of $F$ (modulo wrapping) is obtained by $\symb{fsucc}\ F\ 
\mathit{lst}$.
\end{example}

\subsection{Semantics}\label{sec:semantics}

Like Jones, our language has a closure-based call-by-value semantics.
We let data expressions, values and environments be defined by
the grammar in Figure~\ref{fig:values}.

\begin{figure}[!htpb]
\vspace{-12pt}
\begin{center}
\fbox{
\begin{minipage}{0.5\textwidth}
\begin{tabular}{rcl}
$d,b \in \texttt{Data}$ & ::= & $\apps{c}{d_1}{d_m} \mid (d,b)$ \\
$v,w \in \texttt{Value}$ & ::= & $d \mid (v,w) \mid
\apps{f}{v_1}{v_n}$ \\
& & ($n < \arity(f)$) \\
$\gamma,\delta \in \texttt{Env}$ & ::= & $\V \rightarrow
\texttt{Value}$ \\
\end{tabular}
\end{minipage}
\begin{minipage}{0.4\textwidth}
\begin{tabular}{rcl}
\multicolumn{2}{c}{\textbf{Instantiation:}} \\
$x\gamma$ & := & $\gamma(x)$ \\
$(\apps{\identifier{c}}{\ell_1}{\ell_n})\gamma$ & := &
$\apps{\identifier{c}}{(\ell_1\gamma)}{(\ell_n\gamma)}$ \\
\\
\end{tabular}
\end{minipage}
}
\end{center}
\vspace{-12pt}
\caption{Data expressions, values and environments}
\vspace{-12pt}
\label{fig:values}
\end{figure}

\smallskip\noindent
Let $\domain(\gamma)$ denote the domain of an environment (partial
function) $\gamma$.  Note that values are ground expressions, and we
only use well-typed values with fully\linebreak
applied data constructors.
To every pattern $\ell$ and environment $\gamma$ with $\domain(\gamma)
\supseteq \Var(\ell)$, we associate a value $\ell\gamma$ by
instantiation in the obvious way, see Figure~\ref{fig:values}.

Note that, for every value $v$ and pattern $\ell$, there is at most one
environment $\gamma$ with $\ell\gamma = v$.
We say that an expression $\apps{\identifier{f}}{s_1}{s_n}$
\emph{instantiates} the left-hand side of a clause
$\apps{\identifier{f}}{\ell_1}{\ell_k}$ if $n = k$ and there is an
environment $\gamma$ with each $s_i = \ell_i\gamma$.

Both input and output to the program are data expressions.  If
$\identifier{f}_1$ has type $\asortorpair_1 \arrtype \dots \arrtype
\asortorpair_M \arrtype \asortorpair$, we can think of the program as
calculating a function $\progevaluation$ from $M$ input data arguments
to an output data expression.

Expression and program evaluation are given by the rules in
Figure~\ref{fig:evaluation}.
Since, in [Call], there is at most one suitable
$\gamma$, the only source of non-determinism is the
$\choice$ operator.  Programs without this operator are called
\emph{deterministic}.
By contrast, we may refer to a 
\emph{non-deterministic} program as one which is not explicitly
required to be deterministic, so which may or may not contain
$\choice$.

\begin{figure}[!ht]
\fbox{
\begin{minipage}{0.98\linewidth}
\begin{center}
\textbf{Expression evaluation:}
\vspace{-8pt}
\end{center}
\begin{minipage}{0.5\linewidth}
\begin{prooftree}
\AxiomC{}
\LeftLabel{[Instance]:\quad}
\UnaryInfC{$\prog,\gamma \vdash x \arrr \gamma(x)$}
\end{prooftree}
\end{minipage}
\begin{minipage}{0.5\linewidth}
\begin{prooftree}
\AxiomC{$\prog \vdashcall \identifier{f} \arrr w$}
\RightLabel{for $\identifier{f} \in \Defineds$}
\LeftLabel{[Function]:\quad}
\UnaryInfC{$\prog,\gamma \vdash \identifier{f} \arrr w$}
\end{prooftree}
\end{minipage}

\begin{prooftree}
\AxiomC{$\prog,\gamma \vdash s_1 \arrr b_1$}
\AxiomC{$\cdots$}
\AxiomC{$\prog,\gamma \vdash s_m \arrr b_m$}
\LeftLabel{[Constructor]:\quad}
\TrinaryInfC{$\prog,\gamma \vdash \apps{\identifier{c}}{s_1}{s_m}
\arrr \apps{\identifier{c}}{b_1}{b_m}$}
\end{prooftree}

\begin{prooftree}
\AxiomC{$\prog,\gamma \vdash s \arrr v$}
\AxiomC{$\prog,\gamma \vdash t \arrr w$}
\LeftLabel{[Pair]:\quad}
\BinaryInfC{$\prog,\gamma \vdash (s,t) \arrr (v,w)$}
\end{prooftree}

\begin{prooftree}
\AxiomC{$\prog,\gamma \vdash s_i \arrr w$}
\LeftLabel{[Choice]:\quad}
\RightLabel{for $1 \leq i \leq n$}
\UnaryInfC{$\prog,\gamma \vdash \apps{\choice}{s_1}{s_n} \arrr w$}
\end{prooftree}

\begin{prooftree}
\AxiomC{$\prog,\gamma \vdash s_1 \arrr d$}
\AxiomC{$\prog,\gamma \vdashif d,s_2,s_3 \arrr w$}
\LeftLabel{[Conditional]:}
\BinaryInfC{$\prog,\gamma \vdash \ifte{s_1}{s_2}{s_3} \arrr w$}
\end{prooftree}

\begin{minipage}{0.5\linewidth}
\begin{prooftree}
\AxiomC{$\prog,\gamma \vdash s_2 \arrr w$}
\LeftLabel{[If-True]:}
\UnaryInfC{$\prog,\gamma \vdashif \strue,s_2,s_3 \arrr w$}
\end{prooftree}
\end{minipage}
\begin{minipage}{0.5\linewidth}
\begin{prooftree}
\AxiomC{$\prog,\gamma \vdash s_3 \arrr w$}
\LeftLabel{[If-False]:}
\UnaryInfC{$\prog,\gamma \vdashif \sfalse,s_2,s_3\ \arrr w$}
\end{prooftree}
\end{minipage}

\begin{prooftree}
\AxiomC{$\prog,\gamma \vdash s \arrr \apps{\identifier{f}}{v_1}{v_n}$}
\AxiomC{$\prog,\gamma \vdash t \arrr v_{n+1}$}
\AxiomC{$\prog \vdashcall \apps{\identifier{f}}{v_1}{v_{n+1}} \arrr w$}
\LeftLabel{[Appl]:\quad}
\TrinaryInfC{$\prog,\gamma \vdash s\ t \arrr w$}
\end{prooftree}

\begin{prooftree}
\AxiomC{}
\RightLabel{if $n < \arity(\identifier{f})$}
\LeftLabel{[Closure]:\quad}
\UnaryInfC{$\prog \vdashcall \apps{\identifier{f}}{v_1}{v_n} \arrr
\apps{\identifier{f}}{v_1}{v_n}$}
\end{prooftree}
\begin{prooftree}
\AxiomC{$\prog,\gamma \vdash s \arrr w$}
\RightLabel{\begin{tabular}{l}
if $\apps{\identifier{f}}{\ell_1}{\ell_k} = s$ is the
first clause in $\prog$ such \\
that $\apps{\identifier{f}}{v_1}{v_k}$ instantiates
$\apps{\identifier{f}}{\ell_1}{\ell_k}$, and \\
$\domain(\gamma) = \Var(\apps{\identifier{f}}{\ell_1}{\ell_k})$
and each $v_i = \ell_i\gamma$
\end{tabular}}
\LeftLabel{[Call]:\quad}
\UnaryInfC{$\prog\vdashcall \apps{\identifier{f}}{v_1}{v_k} \arrr w$}
\end{prooftree}

\begin{center}
\vspace{2pt}
\textbf{Program execution:}
\end{center}
\vspace{-14pt}

\begin{prooftree}
\AxiomC{$\prog,[x_1:=d_1,\dots,x_M:=d_M] \vdash
\apps{\identifier{f}_1}{x_1}{x_M} \arrr b$}
\UnaryInfC{$\progresult$}
\end{prooftree}
\end{minipage}
}
\caption{Call-by-value semantics}
\vspace{-8pt}
\label{fig:evaluation}
\end{figure}

\begin{example}
For the program from Example~\ref{ex:successor},
$\progeval{\prog}{\strue\cons\sfalse\cons\strue\cons\nil} \mapsto
\sfalse\cons\strue\cons\strue\cons\nil$, giving $5+1 = 6$.
In the program $\identifier{f}_1\ x\ y = \choice\ 
x\ y$, we can both derive $\progeval{\prog}{\strue,\sfalse} \mapsto
\strue$ and $\progeval{\prog}{\strue,\sfalse} \mapsto \sfalse$.
\end{example}

The language is easily seen to be Turing-complete unless further
restrictions are imposed.  In order to assuage any fears
on
whether the complexity-theoretic characterisations
we obtain are due to brittle design choices, we
add some remarks.

\begin{remark}
We have omitted some constructs common to even some toy pure
functional languages, but these are in general simple syntactic sugar
that can be readily expressed by the existing constructs in the
language, even in the presence of non-determinism. For instance, a
let-binding $\texttt{let} \, x = s_1 \, \texttt{in} \, s_2$ can be
straightforwardly encoded by a function call in a pure call-by-value
setting (replacing $\texttt{let} \, x = s_1 \, \texttt{in} \, s_2$ by
$\symb{helper}\ s_1$ and adding a clause $\symb{helper}\ x = s_2$).
%
\end{remark}

\begin{remark}
We do not require the clauses of a function definition
to exhaust all possible patterns. For instance, it is possible to have
a clause $\symb{f} \, \strue = \cdots$ without a 
clause
for $\symb{f} \, \sfalse$.
Thus,
a program 
has zero or more values.
\end{remark}

\paragraph{Data order versus program order.}
We have followed Jones in considering \emph{data order} as the
variable for increasing complexity.  However, an alternative choice
---which turns out to streamline our proofs---is
\emph{program order}, which considers the type order of the function
symbols.  Fortunately, these notions are closely related; barring
unused symbols, $\langle$program order$\rangle$ = $\langle$data
order$\rangle$ + 1.

More specifically, we have the
following result:
\edef\propernesslem{\number\value{lemma}}
\begin{lemma}\label{lem:proper}
For every well-formed program $\prog$ with data order $K$, there is a
well-formed program $\eprog$ such that
$\progresult$ iff $\progeval{\eprog}{d_1,\dots,d_M} \mapsto b$
for any $b_1,\dots,b_M,d$ and:
(a) all defined symbols in $\eprog$ have a type $\atype_1 \arrtype
  \dots \arrtype \atype_m \arrtype \asortorpair$ such that both
  $\typeorder{\atype_i} \leq K$ for all $i$ and
  $\typeorder{\asortorpair} \leq K$, and
(b) in all clauses, all sub-expressions of the right-hand side have a
  type of order $\leq K$ as well.
\end{lemma}

\begin{proof}[Sketch]
$\eprog$ is obtained from $\prog$ through the following
successive changes:
\begin{enumerate}
\item\label{lem:proper:arity}
  Replace any clause $\apps{\identifier{f}}{\ell_1}{\ell_k} = s$ where
  $s : \atype \arrtype \btype$ with $\typeorder{\atype \arrtype
  \btype} = K+1$, by $\apps{\identifier{f}}{\ell_1}{\ell_k}\ x = s\ x$
  for a fresh $x$.
  Repeat until no such clauses remain.
\item\label{lem:proper:ifchoice}
  In any clause $\apps{\identifier{f}}{\ell_1}{\ell_k} = s$, replace
  all sub-expressions
  $\apps{(\apps{\choice}{s_1}{s_m})}{t_1\linebreak
  }{t_n}$ or
  $\apps{(\ifte{s_1}{s_2}{s_3})}{t_1}{t_n}$ of $s$ with
  $n > 0$ by
  $\apps{\choice}{(\apps{s_1}{t_1}{t_n})\linebreak
  }{(\apps{s_m}{t_1}{t_n})}$ or
  $\ifte{s_1}{(\apps{s_2}{t_1}{t_n})}{(\apps{s_3}{t_1}{t_n})}$
  respectively.
\item\label{lem:proper:occurrence}
  In any clause $\apps{\identifier{f}}{\ell_1}{\ell_k} = s$, if $s$
  has a sub-expression $t = \apps{\identifier{g}}{s_1}{s_n}$ with
  $\identifier{g} : \atype_1 \arrtype \dots \arrtype \atype_n \arrtype
  \btype$ such that $\typeorder{\btype} \leq K$ but
  $\typeorder{\atype_i} > K$ for some $i$, then replace
  $t$ by a fresh symbol $\bot_\btype$.  Repeat
  until no such sub-expressions remain, then add clauses
  $\bot_\btype = \bot_\btype$ for the new symbols.
\item\label{lem:proper:removal}
  If there exists $\identifier{f} : \atype_1 \arrtype \dots \arrtype
  \atype_m \arrtype \asortorpair \in \F$ with
  $\typeorder{\asortorpair} > K$ or $\typeorder{\atype_i} > K$ for
  some $i$, then remove the symbol $\identifier{f}$ and all clauses
  with root $\identifier{f}$.
\end{enumerate}
The key observation is that if the derivation for $\progresult$ uses
some $\apps{\identifier{f}}{s_1}{s_n} : \atype$ with $\typeorder{\atype}
\leq K$ but $s_i : \btype$ with $\typeorder{\btype} > K$, then there is
a variable with type order $> K$.  Thus,
if a clause introduces such an expression, either the clause is never
used, or the expression occurs beneath an $\symb{if}$ or $\choice$ and
is never selected; it may be replaced with a symbol
whose only rule is unusable.  
This also justifies step~\ref{lem:proper:arity}; for
step~\ref{lem:proper:removal}, only unusable clauses are removed.

(See Appendix~\ref{app:properness} for the complete proof.)
\qed
\end{proof}

\begin{example}\label{ex:proper}
The following program has data order $0$, but clauses of functional type;
$\symb{fst}$ and $\symb{snd}$ have output type $\nat \arrtype
\nat$ of order $1$.  The program is changed by
replacing the last two clauses
by
$\symb{fst}\ x\ y = \symb{const}\ x\ y$ and $\symb{snd}\ x\ y =
\symb{id}\ y$.
\[
\begin{array}{lcl}
\multicolumn{3}{l}{
  \symb{start}\ xs\ ys = \choice\ (\symb{fst}\ xs\ ys)\ 
    (\symb{snd}\ xs\ ys)} \\
\symb{const}\ x\ y = x & \quad\quad &
\symb{fst}\ x = \symb{const}\ x \\
\symb{id}\ x = x & \quad\quad &
\symb{snd}\ x = \symb{id} \\
\end{array}
\]
\end{example}

\section{Cons-free programs}\label{sec:consfree}

Jones defines a cons-free program as one where the list constructor
$\cons$ does not occur in 
any clause.
In our setting (where
more constructors are in principle admitted), this translates
to disallowing 
non-constant data constructors
from being introduced
in the right-hand side of a clause.  We define:

\begin{definition}\label{def:consfree}
A program $\prog$ is cons-free if all clauses in $\prog$ are
cons-free.  A clause $\apps{\identifier{f}}{\ell_1}{\ell_k} = s$ is
cons-free if
for all $s \suptermeq t$: if $t = \apps{\identifier{c}}{s_1}{s_m}$
with $\identifier{c} \in \Constructors$, then $t$ is a data expression
or $\ell_i \suptermeq t$ for some $i$.
\end{definition}

\begin{example}
Example~\ref{ex:successor} is not cons-free, due to the second and
third clause (the first clause \emph{is} cons-free).
Examples~\ref{ex:hocount} and~\ref{ex:proper} are
both cons-free.
\end{example}


The key property of cons-free programming is that no \emph{new} data
structures can be created during program execution.  Formally,
in a derivation tree
with
root
$\progresult$,
all data values (including $b$) are in the set $\B_{d_1,\dots,
d_M}^\prog$:

\begin{definition}
Let $\B^\prog_{d_1,\dots,d_M} := \{ d \in \Data \mid \exists
i 
[d_i \suptermeq d] \vee \exists (\identifier{f}\ \vec{\ell} = s) \in
\prog [s \suptermeq d] \}$.
\end{definition}

$\B_{d_1,\dots,d_M}^\prog$ is a set of data expressions closed under
$\supterm$, with a linear number of elements in the size of $d_1,
\dots,d_M$ (for fixed $\prog$).  The property that no new data is
created during execution is formally expressed by the following lemma.

\edef\safetypreservelem{\number\value{lemma}}
\begin{lemma}\label{lem:safetysimple}
Let $\prog$ be a cons-free program, and suppose that $\progresult$ is
obtained by a derivation tree $T$.  Then for all statements $\prog,
\gamma \vdash s \arrr w$ or $\prog,\gamma \vdashif b',s_1,s_2 \arrr w$
or $\prog \vdashcall \apps{\identifier{f}}{v_1}{v_n} \arrr w$ in T,
and all expressions $t$ such that (a) $w \suptermeq t$, (b) $b'
\suptermeq t$, (c) $\gamma(x) \suptermeq t$ for some $x$ or (d) $v_i
\suptermeq t$ for some $i$:
if $t$ has the form $\apps{\identifier{c}}{b_1}{b_m}$ with
$\identifier{c} \in \Constructors$, then $t \in
\B_{d_1,\dots,d_M}^\prog$.
\end{lemma}
That is, any data expression in the derivation tree of
$\progresult$ (including occurrences as a sub-expression of
other values) is also in $\B_{d_1,\dots,d_M}^\prog$.

\begin{proof}[Sketch]
Induction on the form of $T$, assuming that for a statement under
consideration, (1) the requirements on $\gamma$ and the $v_i$ are
satisfied,
and (2) $\gamma$ maps expressions $t \subtermeq s,s_1,s_2$ to
elements of $\B_{d_1,\dots,d_M}^\prog$ if $t =
\apps{\identifier{c}}{t_1}{t_m}$ with $\identifier{c} \in
\Constructors$.

(See Appendix~\ref{app:consfree} for the complete proof.)
\qed
\end{proof}

Note that Lemma~\ref{lem:safetysimple} implies that the program result
$b$ is in $\B_{d_1,\dots,d_M}^\prog$.
Recall also Remark~\ref{remark:constructor}: if we had
admitted constructors with higher-order argument types, then
Lemma~\ref{lem:safetysimple} shows that they are never used, since
any constructor appearing in a derivation for $\progresult$
must already occur in the (data!) input.


\section{Turing Machines, decision problems and complexity}

We assume familiarity with the standard notions of Turing Machines and
complexity classes (see, e.g., \cite{Papadimitriou:complexity,Jones:CompComp,Sipser:comp});
in this section, we fix the notation we use.

\subsection{(Deterministic) Turing Machines}

Turing Machines (TMs) are triples $(A,S,T)$ where $A$ is a finite set
of tape symbols such that $A \supseteq \{0,1,\blank\}$,\ $S
\supseteq \{ \symb{start}, \symb{accept}, \symb{reject} \}$
is a finite set of states, and $T$ is a finite set of transitions
$(i,r,w,d,j)$ with $i \in S \setminus \{\symb{accept},\symb{reject}\}$
(the \emph{original state}), $r \in A$ (the \emph{read symbol}), $w
\in A$ (the \emph{written symbol}), $d \in \{ \symb{L},\symb{R} \}$
(the \emph{direction}), and $j \in S$ (the \emph{result state}).
We sometimes denote this transition as $\transition{i}{r}{w}{d}{j}$.

A \emph{deterministic} Turing Machine is a TM such that every pair
$(i,r)$ with $i \in
S \setminus \{\symb{accept},\symb{reject}\}$ and $r \in A$ is
associated with exactly one transition $(i,r,w,d,j)$.
Every TM in this paper has a single, right-infinite tape.

A \emph{valid tape} is 
an element $t$ of $A^\N$ with $t(p) \neq \blank$ for only finitely
many $p$.
A \emph{configuration} 
is a triple $(t,p,s)$ with $t$ a valid tape, $p \in \N$ and $s \in S$.
The transitions $T$ induce a 
relation $\Rightarrow$ between configurations in the obvious way.

\subsection{Decision problems}

A \emph{decision problem} is a set $X \subseteq \{0,1\}^+$.
A deterministic TM \emph{decides} $X$ if for any
$x \in \{0,1\}^+$: $x \in X$ if{f} $\blank x_1\dots x_n
\blank\blank\dots,0,\symb{start}) \Rightarrow^* (t,i,\symb{accept})$
for some $t,i$,
and 
$(\blank x_1\dots x_n\blank\blank\dots,0,\symb{start})
\Rightarrow^* (t,i,\symb{reject})$ if{f} $x \notin X$.
Thus, 
the TM
halts on all
inputs, ending in $\symb{accept}$ or $\symb{reject}$
depending on whether $x \in X$
.

If $\complexityfun : \N \longrightarrow \N$ is a function, a
deterministic
TM
\emph{runs in time} $\lambda n.\complexityfun(n)$ if for all $n \in
\N \setminus \{0\}$ and $x \in \{0,1\}^n$:
any evaluation starting in $(\blank x_1\dots x_n\blank\blank\dots,0,
\symb{start})$ ends in the $\symb{accept}$ or $\symb{reject}$ state
in at most $h(n)$ transitions.

\subsection{Complexity and the $\exptime{}$ hierarchy}

We define classes of decision problem based on the \emph{time} needed
to accept them.

\begin{definition}
Let $\complexityfun : \N \rightarrow \N$ be a function.
Then, $\timecomp{\complexityfun(n)}$ is the set of all $X \subseteq
\{0,1\}^+$ such that there exist $a > 0$ and a deterministic TM
running in time $\lambda n.a \cdot \complexityfun(n)$ that decides
$X$.
\end{definition}
By design, $\timecomp{\complexityfun(n)}$) is closed
under $\OO$: $\timecomp{\complexityfun(n)} =
\timecomp{\OO(\complexityfun(n))}$.

\begin{definition}
For $K,n \geq 0$, let $\exp_2^0(n) = n$ and $\exp_2^{K+1}(n) =
\exp_2^K(2^n) = 2^{\exp_2^K(n)}$.
For $K \geq 0$, define
$
\exptime{K} \triangleq \bigcup_{a,b \in \N}
\timecomp{\textrm{exp}_2^{K}(an^b)}
$.
\end{definition}

Since for every polynomial $\complexityfun$, there are $a,b
\in \N$ such that $\complexityfun(n) \leq a \cdot n^b$ for all $n >
0$, we have $\exptime{0} = \pclass$ and $\exptime{1} = \expclass$
(where $\expclass$ is the usual complexity class of this name, see
e.g., \cite[Ch.\ 20]{Papadimitriou:complexity}).
In the literature, $\expclass$ is sometimes called
$\mathsf{EXPTIME}$ or $\mathsf{DEXPTIME}$ (e.g., in the celebrated
proof that \textsc{ML} typability is complete for $\mathsf{DEXPTIME}$
\cite{DBLP:journals/jacm/KfouryTU94}).
Using the Time Hierarchy Theorem \cite{Sipser:comp}, it is easy to
see that $\pclass = \exptime{0} \subsetneq \exptime{1} \subsetneq
\exptime{2} \subsetneq \cdots$.

\begin{definition}
The set $\elementary$ of elementary-time computable languages is
$\bigcup_{K \in \N} \exptime{K}$.
\end{definition}

\subsection{Decision problems and programs}\label{subsec:decision}

To solve decision problems by (cons-free) programs, we will
consider programs with constructors $\strue,\sfalse$ of type $\bool$,
$\nil$ of type $\bits$ and $\cons$ of type $\bool \arrtype
\bits \arrtype \bits$, and whose main function
$\identifier{f}_1$ has type $\bits \arrtype \bool$.

\begin{definition}
We define:
\begin{itemize}
\item A program $\prog$ \emph{accepts} $\symb{a}_1\symb{a}_2\dots
\symb{a}_n \in \{0,1\}^*$ if $\progeval{\prog}{\overline{\symb{a}_1
}\cons\dots\cons\overline{\symb{a}_n}} \mapsto \strue$, where
$\overline{\symb{a}_i} = \strue$ if $\symb{a}_i = 1$ and
$\overline{\symb{a}_i} = \sfalse$ otherwise.
\item The \emph{set accepted by} program $\prog$ is
$\{ \symb{a} \in \{0,1\}^* \mid \prog$ accepts $\symb{a} \}$.
\end{itemize}
\end{definition}

Although we focus on programs of this form, our proofs will
allow for arbitrary input and output---with the limitation (as
guaranteed by the rule for program execution) that both are data.
This makes it possible to for instance consider decision problems on
a larger input alphabet without needing encodings.

\begin{example}
The two-line program with clauses $\symb{even}\ \nil = \strue$ and
$\symb{even}\ (x\cons xs) = \ifte{\:x\:}{\:\sfalse\:}{\:\strue\:}$
accepts the problem $\{ x \in \{0,1\}^* \mid x$ is a bitstring
representing an even number (following Example~\ref{ex:successor})$\}$.
\end{example}

We will sometimes speak of the input size, defined by:

\begin{definition}
The \emph{size} of a list of data expressions $d_1,\dots,d_M$ is
$\sum_{i = 1}^M \mathit{size}(d_i)$, where
$\mathit{size}(\apps{\identifier{c}}{b_1}{b_m})$ is defined as $1 +
\sum_{i = 1}^m \mathit{size}(b_i)$.
\end{definition}

\section{Deterministic characterisations}\label{sec:deterministic}

As a basis, we transfer Jones' basic result on \emph{time} classes
to our more general language.  That is,
we obtain the first line of the first table in
Figure~\ref{fig:overview}.

\smallskip\noindent
\begin{tabular}{c|c|c|c|c|c}
& \textbf{data order 0} &
\textbf{data order 1} &
\textbf{data order 2} &
\textbf{data order 3} & \dots
\\
\cline{1-6}
\vphantom{$2^{x^y}$}
\textbf{cons-free} &
$\pclass =$ &
$\expclass =$ &
\multirow{2}{*}{$\exptime{2}$} &
\multirow{2}{*}{$\exptime{3}$} & \multirow{2}{*}{\dots}
\\
\textbf{deterministic} &
$\exptime{0}$ &
$\exptime{1}$ &
& \\
\cline{1-6}
\end{tabular}

\smallskip
To show that deterministic cons-free programs of data order $K$
characterise $\exptime{K}$ it is necessary to prove two things:
\begin{enumerate}
\item\label{enum:simulate}
if $\complexityfun(n) \leq \exp_2^K(a \cdot n^b)$ for all $n$,
then for every deterministic Turing Machine $M$ running in
$\timecomp{\complexityfun(n)}$, there is a deterministic, cons-free
program with data order at most $K$, which accepts $x \in
\{0,1\}^+$ if and only if $M$ does;
\item\label{enum:algorithm}
for every deterministic cons-free program $\prog$ with data
order $K$, there is a deterministic algorithm operating in
$\timecomp{\exp_2^K(a \cdot n^b)}$ for some $a,b$ which, given input
expressions $d_1,\dots,d_M$, determines $b$ such that
$\progresult$ (if such $b$ exists).
Like Jones \cite{jon:01}, we assume our algorithms are implemented
on a sufficiently expressive Turing-equivalent machine like the RAM.
\end{enumerate}

We will show part (\ref{enum:simulate}) in
Section~\ref{subsec:simulcore}, and part
(\ref{enum:algorithm}) in Section~\ref{subsec:detalgorithm}.

\subsection{Simulating TMs using deterministic
cons-free programs}
\label{subsec:simulcore}\label{subsec:counting}

Let $M := (A,S,T)$ be a deterministic Turing Machine running in time
$\lambda n.\complexityfun(n)$.
Like Jones, we start by assuming that we have a way to represent the
numbers $0,\dots,\complexityfun(n)$ as expressions, along with
successor and predecessor operators and checks for equality.
Our simulation uses the following data constructors
\begin{itemize}
\item
  $\strue : \bool,\ \sfalse : \bool,\ \nil : \bits$ and $\cons :
  \bool \arrtype \bits \arrtype \bits$ as discussed in
  Section~\ref{subsec:decision};
\item
  $\symb{a} : \msymbol$ for $a \in A$ (writing $\symb{B}$ for the
  blank symbol), $\symb{L},\symb{R} : \mdirec$ and
  $\symb{s} : \mstate$ for $s \in S$;
\item
  $\symb{action} : \msymbol \arrtype \mdirec \arrtype \mstate
  \arrtype \symb{trans}$;
and
\item
  $\symb{end} : \symb{state} \arrtype \symb{trans}$.
\end{itemize}

The rules to simulate the machine are given in
\pagebreak
Figure~\ref{fig:machine}.

\begin{figure}[!ht]
$\symb{run}\ cs = \symb{test}\ (\symb{state}\ cs\ \numrep{\complexityfun(|cs|)})$

$\symb{test}\ \symb{accept} = \strue$
\phantom{ABCDEF}
$\symb{transition}\ \symb{i}\ \symb{r} = \symb{action}\ \symb{w}\ 
\symb{d}\ \symb{j}$
\hfill for all $\transition{i}{r}{w}{d}{j} \in T$ \\
$\symb{test}\ \symb{reject} = \sfalse$
\phantom{abCDEF}
$\symb{transition}\ \symb{i}\ x = \symb{end}\ \symb{i}$
\hfill for $i \in \{ \symb{accept},\symb{reject} \}$

\medskip
$\symb{state}\ cs\ \numrep{n} = \ifte{\numrep{n = 0}}{\:\symb{start}
\:}{\:\symb{get3}\ (\symb{transat}\ cs\ \numrep{n-1})}$ \\
$\symb{transat}\ cs\ \numrep{n} =
\symb{transition}\ (\symb{state}\ cs\ \numrep{n})\ 
(\symb{tapesymb}\ cs\ \numrep{n})$

\medskip
$\symb{get1}\ (\symb{action}\ x\ y\ z) = x$
$\phantom{ABC}\symb{get1}\ (\symb{end}\ x) = \symb{B}$ \\
$\symb{get2}\ (\symb{action}\ x\ y\ z) = y$
$\phantom{ABC}\symb{get2}\ (\symb{end}\ x) = \symb{R}$ \\
$\symb{get3}\ (\symb{action}\ x\ y\ z) = z$
$\phantom{ABC}\symb{get3}\ (\symb{end}\ x) = x$

\medskip
$\symb{tapesymb}\ cs\ \numrep{n} = \symb{tape}\ cs\ \numrep{n}\ 
(\symb{pos}\ cs\ \numrep{n})$

\medskip
$\symb{tape}\ cs\ \numrep{n}\ \numrep{p} = \ifte{\numrep{n = 0}}{\:
\symb{inputtape}\ cs\ \numrep{p}
\\\phantom{\symb{tape}\ cs\ \numrep{n}\ \numrep{p} =\,}
}{\:
\symb{tapehelp}\ cs\ \numrep{n}\ \numrep{p}\ 
(\symb{pos}\ cs\ \numrep{n-1})}$ \\
$\symb{tapehelp}\ cs\ \numrep{n}\ \numrep{p}\ \numrep{i} =
\ifte{\numrep{p = i}}{\:
\symb{get1}\ (\symb{transat}\ cs\ \numrep{n-1})\\
\phantom{\symb{tapehelp}\ cs\ \numrep{n}\ \numrep{p}\ \numrep{i} =\,}
}{\:
\symb{tape}\ cs\ \numrep{n-1}\ \numrep{p}}$

\medskip
$\symb{pos}\ cs\ \numrep{n} = \ifte{\numrep{n=0}}{\numrep{0}\:
}{\:
\symb{adjust}\ cs\ (\symb{pos}\ cs\ \numrep{n-1})\ 
(\symb{get2}\ (\symb{transat}\ cs\ \numrep{n-1}))}\:\:$ \\
$\symb{adjust}\ cs\ \numrep{p}\ \symb{L} = \numrep{p-1}$
\phantom{AB}
$\symb{adjust}\ cs\ \numrep{p}\ \symb{R} = \numrep{p+1}$ \\

$\symb{inputtape}\ cs\ \numrep{p} = \ifte{\numrep{p=0}}{\:\symb{B\:}}{\:
\symb{nth}\ cs\ \numrep{p-1}}$ \\
$\symb{nth}\ \nil\ \numrep{p} = \symb{B}$
\phantom{ABCDEFGHIJKLMNOPQRSTUVWXYZABC}
$\symb{bit}\ \strue = \symb{1}$ \\
$\symb{nth}\ (x\cons xs)\ \numrep{p} = \ifte{\numrep{p = 0}}{\:
\symb{bit}\ x\:}{\:\symb{nth}\ xs\ \numrep{p-1}}$
\phantom{ac}
$\symb{bit}\ \sfalse = \symb{0}$ \\
\vspace{-12pt}
\caption{Simulating a deterministic Turing Machine $(A,S,T)$}
\vspace{-12pt}
\label{fig:machine}
\end{figure}

Types of defined symbols are easily derived.  The intended meaning is
that $\symb{state}\ cs\ \numrep{n}$, for $cs$ the 
input list
and $\numrep{n}$ a number in $\{0,
\dots,\complexityfun(|cs|)\}$, returns the state of the machine at
time $\numrep{n}$; 
$\symb{pos}\ cs\ \numrep{n}$ returns the
position of the reader at time $\numrep{n}$, and $\symb{tape}\ cs
\numrep{n}\ \numrep{p}$ the symbol at time $\numrep{n}$ and
position $\numrep{p}$.

Clearly, the program is highly exponential, even when
$\complexityfun(|cs|)$ is polynomial, since the same expressions are
repeatedly evaluated.  This apparent contradiction is not problematic:
we do not claim that all cons-free programs with data order $0$ (say)
have a derivation tree of at most polynomial size.  Rather, as we will
see in Section~\ref{subsec:detalgorithm}, we can find their
\emph{result} in polynomial time by essentially using a caching
mechanism to avoid reevaluating the same expression.

\medskip
What remains is to simulate numbers and counting.  For a
machine running in $\timecomp{\complexityfun(n)}$,
it suffices to find a value $\numrep{i}$ representing $i$ for all $i
\in \{0,\dots,\complexityfun(n)\}$ and cons-free clauses to
calculate predecessor and successor functions and to perform
zero and equality checks.  This is given by a
\emph{$(\lambda n.\complexityfun(n)+1)$-counting module.}
This defines, for a given input list $cs$ of length $n$, a set of
values $\A_\pi^n$ to represent numbers and functions $\seed,\ \pred$
and $\zero$ such that (a) $\seed\ cs$ evaluates to a value which
represents $\complexityfun(n)$, (b) if $v$ represents a number $k$,
then $\pred\ cs\ v$ evaluates to a value which represents $k-1$, and
(c) $\zero\ cs\ v$ evaluates to $\strue$ or $\sfalse$ depending on
whether $v$ represents $0$.  Formally:

\begin{definition}[Adapted from~\cite{jon:01}]
For $P : \N \rightarrow \N \setminus \{0\}$, a $P$-counting module is
a tuple $C_\pi = (\numtype,\Defineds_\pi,\A_\pi,
\numinterpret{\cdot}_\pi,\prog_\pi)$ such that:
\begin{itemize}
\item $\numtype$ is a type (this will be the type of
\emph{numbers});
\item $\Defineds_\pi$ is a set of defined symbols disjoint from
$\Constructors,\Defineds,\V$, containing symbols
$\seed,\ \pred$ and $\zero$, with types
$\seed : \bits \arrtype \numtype,\ 
\pred : \bits \arrtype \numtype \arrtype \numtype$ and
$\zero : \bits \arrtype \numtype \arrtype \bool$;
\item for $n \in \N$,
$\A_\pi^n$ is a set of values of type $\numtype$, all built over
$\Constructors \cup \Defineds_\pi$ (this is the set of values used
to represent numbers);
\item for $n \in \N$,
$\numinterpret{\cdot}_\pi^n$ is a total function from
$A_{\pi}^n$ to $\N$;
\item $\prog_\pi$ is a list of cons-free clauses on the symbols in
$\Defineds_\pi$, such that, for all lists $cs : \bits \in \Data$
with length $n$:
\begin{itemize}
\item there is a unique value $v$ such that $\prog_\pi \vdashcall
  \seed\ cs \arrr v$;
\item if $\prog_\pi \vdashcall \seed\ cs \arrr v$, then $v \in
  \A_{\pi}^n$ and $\numinterpret{v}_\pi^n = P(n)-1$;
\item if $v \in \A_\pi$ and $\numinterpret{v}_\pi^n = i > 0$,
  then there is a unique value $w$ such that $\prog_\pi \vdashcall
  \pred\ cs\ v \arrr w$;
  we have $w \in \A_\pi^n$ and $\numinterpret{w}_\pi^n = i-1$;
\item for $v \in \A_\pi^n$ with $\numinterpret{v}_\pi^n = i$:
  $\prog_\pi \vdashcall \zero\ cs\ v \arrr \strue$ if and only if
  $i = 0$, and
  $\prog_\pi \vdashcall \zero\ cs\ v \arrr \sfalse$ if and only if
  $i > 0$.
\end{itemize}
\end{itemize}
\end{definition}

It is easy to see how a $P$-counting module can be plugged into the
program of Figure~\ref{fig:machine}.  We only lack successor and
equality functions, which are easily defined:

\smallskip\noindent
$\symb{succ}_\pi\ cs\ i = \symb{sc}_\pi\ cs\ (\seed\ cs)\ i$ \\
$\symb{sc}_\pi\ cs\ j\ i = \ifte{\:\symb{equal}_\pi\ cs\ 
(\symb{pred}_\pi\ cs\ j)\ i\:}{\:j\:}{\:
\symb{sc}\ cs\ (\pred\ cs\ j)\ i}$ \\
$\symb{equal}_\pi\ cs\ i\ j = \ifte{\:\zero\ cs\ i\:}{\:\zero\ cs\ j
\\\phantom{\symb{equal}_\pi\ cs\ i\ j =\,}
}{\ifte{\:\zero\ cs\ j\:}{\:\sfalse
\\\phantom{\symb{equal}_\pi\ cs\ i\ j =\,}
}{\:\symb{equal}_\pi\ cs\ (\pred\ cs\ i)\ 
(\pred\ cs\ j)}
}$

\medskip
Since the clauses in Figure~\ref{fig:machine} are
cons-free and have data order $0$, we obtain:

\begin{lemma}\label{lem:counting}
Let $x$ be a decision problem which can be decided by a
deterministic TM running in $\timecomp{\complexityfun(n)}$.
If there is a cons-free $(\lambda n.\complexityfun(n)+1)$-counting
module $C_\pi$ with data order $K$,
then $x$ is accepted by a cons-free program with data order $K$;
the program is deterministic if the counting module is.
\end{lemma}

\begin{proof}
By the argument given above.
\qed
\end{proof}

The obvious difficulty is the restriction to cons-free clauses: we
cannot simply construct a new number type, but will have to represent
numbers using only sub-expressions of the input list $cs$, and
constant data expressions.

\begin{example}\label{ex:counting}
We consider a $P$-counting module $C_\exx$ where $P(n) = 3 \cdot
(n+1)^2$.
Let $\numtype[\exx] := \bits \times \bits \times \bits$ and
for given $n$, let $\A_\pi^n := \{ (d_0,d_1,d_2) \mid d_0$ is a list
of length $\leq 2$ and $d_1,d_2$ are lists of length
$\leq n \}$.  Writing $|\ x_1\cons \dots \cons x_k
\cons\nil\ | = k$, let $\numinterpret{(d_0,d_1,d_2)}_\exx^n :=
|d_0| \cdot (n+1)^2 + |d_1| \cdot (n+1) + |d_2|$.
Essentially, we 
consider $3$-digit numbers $i_0 i_1 i_2$ in base $n+1$, with each
$i_j$ represented by a list.  $\prog_\exx$ is:
\[
\begin{array}{ll}
\seed[\exx]\ cs = (\sfalse\cons\sfalse\cons\nil,\ cs,\ cs)\ \  \\
\pred[\exx]\ cs\ (x_0,x_1,y\cons ys) = (x_0,x_1,ys) &
\zero[\exx]\ cs\ (x_0,x_1,y\cons ys) = \sfalse \\
\pred[\exx]\ cs\ (x_0,y\cons ys,\nil) = (x_0,ys,cs) &
\zero[\exx]\ cs\ (x_0,y\cons ys,\nil) = \sfalse \\
\pred[\exx]\ cs\ (y\cons ys,\nil,\nil) = (ys,cs,cs) &
\zero[\exx]\ cs\ (y\cons ys,\nil,\nil) = \sfalse \\
\pred[\exx]\ cs\ (\nil,\nil,\nil) = (\nil,\nil,\nil) &
\zero[\exx]\ cs\ (\nil,\nil,\nil) = \strue
\end{array}
\]
If
$cs = \strue\cons\sfalse\cons\strue\cons\nil$, one value
in $\A_\exx^3$ is $v = (\sfalse\cons\nil,\ \sfalse\cons\strue\cons
\nil,\ \nil)$, which is mapped to the number $1 \cdot 4^2 + 2 \cdot 4
+ 0 = 24$.  Then $\prog_\exx \vdashcall \pred[\exx]\ cs\ v \arrr w :=
(\sfalse\cons\nil,\ \strue\cons\nil,\ cs)$, which is mapped to $1
\cdot 4^2 + 1 \cdot 4 + 3 = 23$ as desired.
\end{example}

Example~\ref{ex:counting} suggests a systematic way to create
polynomial counting modules.

\edef\modulepol{\number\value{lemma}}
\begin{lemma}\label{lem:module:pol}
For any $a,b \in \N \setminus \{0\}$, there is a $(\lambda n.a \cdot
(n+1)^b)$-counting module $C_\pol$ with data order $0$.
\end{lemma}

\begin{proof}[Sketch]
A straightforward generalisation of Example~\ref{ex:counting}

(See Appendix~\ref{app:properness} for the complete proof.)
\qed
\end{proof}

By increasing type orders, we can obtain an exponential increase of
magnitude
.

\edef\moduleexp{\number\value{lemma}}
\begin{lemma}\label{lem:module:exp}
If there is a $P$-counting module $C_\pi$ of data order $K$, then
there is a $(\lambda n.2^{P(n)})$-counting module $C_\epi$ of data
order $K+1$.
\end{lemma}

\begin{proof}[Sketch]
Let $\numtype[\epi] := \numtype \arrtype \bool$; then $\typeorder{
\numtype[\epi]} \leq K+1$.  A number $i$ with bit representation $b_0
\dots b_{P(n)-1}$ (with $b_0$ the most significant digit) is
represented by a value $v$ such that, for $w$
with $\numinterpret{w}_\pi = i$:
$\prog_{\epi} \vdashcall v\ w \arrr \strue$ if{f} $b_i = 1$, and
$\prog_{\epi} \vdashcall v\ w \arrr \sfalse$ if{f} $b_i = 0$.  We use
the clauses of Figure~\ref{fig:epi}.
\begin{figure}[!ht]
\[
\begin{array}{l}
\seed[\epi]\ cs\ x = 
\strue \\
\zero[\epi]\ cs\ F = \symb{zhelp}_\epi\ cs\ F\ (\seed\ cs) \\
\symb{zhelp}_\epi\ cs\ F\ k = \ifte{\:F\ k\:}{\:\sfalse
\\\phantom{\symb{zhelp}_\epi\ cs\ k\ F =\,}}{
\ifte{\:\zero\ cs\ k\:}{\:\strue
\\\phantom{\symb{zhelp}_\epi\ cs\ k\ F =\,}
}{\:\symb{zhelp}_\epi\ cs\ F\ (\symb{pred}_\pi\ cs\ k)}} \\
\pred[\epi]\ cs\ F = \symb{phelp}_\epi\ cs\ F\ (\seed\ cs) \\
\symb{phelp}_\epi\ cs\ F\ k = \ifte{\:F\ k\:}{\:\symb{flip}_\epi\ 
cs\ F\ k
\\\phantom{\symb{phelp}_\epi\ cs\ k\ F =\,}
}{\ifte{\:\zero\ cs\ k\:}{\:\seed[\epi]\ cs
\\\phantom{\symb{phelp}_\epi\ cs\ k\ F =\,}
}{\:\symb{phelp}_\epi\ cs\ (\symb{flip}_\epi\ cs\ F\ k)\ 
(\pred\ cs\ k)}} \\
\symb{flip}_\epi\ cs\ F\ k\ i = \ifte{\:\symb{equal}_\pi\ cs\ k\ i
\:}{\:\symb{not}\ (F\ i)\:}{\:F\ i} \\
\symb{not}\ b = \ifte{\:b\:}{\:\sfalse\:}{\:\strue\:} \\
\end{array}
\]
\vspace{-12pt}
\caption{The clauses used in $\prog_\epi$, extending $\prog_\pi$ with an
exponential step.}
\vspace{-12pt}
\label{fig:epi}
\end{figure}

We also include all clauses in $\prog_\pi$.
Here, note that a bitstring $b_0 \dots b_m$ represents $0$ if each
$b_i = 0$, and that the predecessor of $b_0\dots b_i 1 0
\dots 0$ is $b_0 \dots b_i 0 1 \dots 1$.

(See Appendix~\ref{app:properness} for the complete proof.)
\qed
\end{proof}

Combining these results, we obtain:

\begin{lemma}\label{lem:deterministic:simulate}
Every decision problem in $\exptime{K}$ is accepted by a deterministic
cons-free program with data order $K$.
\end{lemma}

\begin{proof}
A decision problem is in $\exptime{K}$ if it is decided by a
deterministic TM operating in time $\exp_2^K(a \cdot n^b))$ for some
$a,b$.  By Lemma~\ref{lem:counting}, it
therefore suffices if there is a $Q$-counting module for some $Q \geq
\lambda n.\exp_2^K(a \cdot n^b)+1$, with data order $K$.
Certainly $Q(n) := \exp_2^K(a \cdot (n+1)^b)$ is large enough.  By
Lemma~\ref{lem:module:pol}, there is a $(\lambda n.a \cdot (n+1)^b)
$-counting module $C_\pol$ with data order $0$.
Applying Lemma~\ref{lem:module:exp} $K$ times, we obtain the required
$Q$-counting module $C_{\symb{e}[\dots[\symb{e}[\pol]]]}$.
\qed
\end{proof}

\begin{remark}
Our definition of a counting module significantly differs from the
one in~\cite{jon:01}, for example by representing numbers as
\emph{values} rather than \emph{expressions}, introducing the sets
$\A_\pi^n$ and imposing evaluation restrictions.  The changes enable
an easy formulation of the non-deterministic counting
module in Section~\ref{sec:elementary}.
\end{remark}

\subsection{Simulating deterministic cons-free programs using
an algorithm}\label{subsec:detalgorithm}

We now turn to the second part of characterisation:
that every decision problem solved by a deterministic
cons-free program of data order $K$ is in $\exptime{K}$.
We 
give
an algorithm which
determines the result of a fixed program
(if any) on a given input
in $\timecomp{\exp_2^K(a \cdot n^b)}$ for some $a,b$.
The algorithm is designed to extend easily to the non-deterministic
characterisations in subsequent settings.

\paragraph{Key idea.}  The principle of our algorithm is easy to
explain 
when variables
have data order $0$.  Using Lemma~\ref{lem:safetysimple}, all such
variables must be instantiated by (tuples of) elements of $\B_{d_1,
\dots,d_M}^\prog$, of which there are only polynomially many in the
input size.  Thus, we can make a comprehensive list of all
expressions that might occur as the left-hand side of a [Call] in
the derivation tree.  Now we can go over the list repeatedly, filling
in reductions to 
trace a top-down derivation of the tree.

In the higher-order setting, there are infinitely many possible
values; for example, if $\symb{id} : \bool \arrtype \bool$
has arity $1$ and $\symb{g} : (\bool \arrtype \bool) \arrtype \bool
\arrtype \bool$ has arity $2$, then $\symb{id},\ \symb{g}\ \symb{id},
\ \symb{g}\ (\symb{g}\ \symb{id})$ and so on are all values.
Therefore, instead of looking directly at values we consider
an extensional replacement.

\begin{definition}\label{def:pinterpret}
Let $\B$ be a set of data expressions closed under $\supterm$.
For $\asort \in \Sorts$, let $\pinterpret{\asort} = \{ d \in \B \mid\ 
\vdash d : \asort \}$.  Inductively, let $\pinterpret{\atype \times
\btype} = \pinterpret{\atype} \times \pinterpret{\btype}$ and
$\pinterpret{\atype \arrtype \btype}\linebreak
= \{ A_{\atype \arrtype \btype}
\mid A \subseteq \pinterpret{\atype} \times \pinterpret{\btype}
\wedge \forall e \in \pinterpret{\atype}$ there is at most one $u$
with $(e,u) \in A_{\atype \arrtype \btype} \}_{\atype \arrtype
\btype}$.
We call the elements of any $\pinterpret{\atype}$
\emph{deterministic extensional values}.
\end{definition}

Note that deterministic extensional values are data
expressions in $\B$ if $\atype$ is a sort, \emph{pairs} if $\atype$
is a pair type, and sets of pairs labelled with a type otherwise;
these sets are exactly partial functions, and can be used as such:

\begin{definition}\label{def:partialapply}
For $e \in \pinterpret{\atype_1 \arrtype \dots \arrtype \atype_n
\arrtype \btype}$ and $u_1 \in \pinterpret{\atype_1},\dots,
u_n \in \pinterpret{\atype_n}$,
we inductively define $e(u_1,
\dots,u_n) \subseteq \pinterpret{\btype}$:
\begin{itemize}
\item if $n = 0$, then $e(u_1,\dots,u_n) = e() = \{ e \}$;
\item if $n \geq 1$, then $e(u_1,\dots,u_n) =
  \bigcup_{A_{\atype_n \arrtype \btype} \in e(u_1,\dots,u_{n-1})}
  \{ o \in \pinterpret{\btype} \mid (u_n,o) \in A \}$.
\end{itemize}
\end{definition}

By induction on $n$, 
each $e(u_1,\dots,u_n)$
has at most one element 
as would be expected of a partial function.
We also
consider a form of matching.

\begin{definition}
Fix a set $\B$ of data expressions.
An \emph{extensional expression} has the form $\apps{\identifier{f}}{
e_1}{e_n}$ where $\identifier{f} : \atype_1 \arrtype \dots \arrtype
\atype_n \arrtype \btype \in \Defineds$ and each $e_i \in
\pinterpret{\atype_i}$.
Given a clause $\rho\colon\apps{\identifier{f}}{\ell_1}{\ell_k} = r$
with $\identifier{f} : \atype_1 \arrtype \dots \arrtype \atype_k
\arrtype \tau \in \F$ and variable environment $\Gamma$, an
\emph{ext-environment} for $\rho$ is a partial function $\eta$
mapping each $x : \btype \in \Gamma$ to an element of $\pinterpret{
\btype}$, such that $\ell_j\eta \in \pinterpret{\atype_j}$ for $1 \leq
j \leq n$.  Here,
\begin{itemize}
\item $\ell\eta = \eta(\ell)$ if $\ell$ is a variable
\item $\ell\eta = (\ell^{(1)}\eta,\ell^{(2)}\eta)$ if
  $\ell = (\ell^{(1)},\ell^{(2)})$;
\item $\ell\eta = \ell[x:=\eta(x) \mid x \in \Var(\ell)]$
otherwise (in this case, $\ell$ is a pattern with data order $0$,
so all its variables have data order $0$, so each $\eta(x) \in
\Data$).
\end{itemize}
Then $\ell\eta$ is a deterministic extensional value for
$\ell$ a pattern.
We say $\rho$ \emph{matches} an extensional expression
$\apps{\identifier{f}}{e_1}{e_k}$
if there is an ext-environment $\eta$ for $\rho$
such that $\ell_i\eta = e_i$ for all $1 \leq i \leq k$.
We call $\eta$ the \emph{matching ext-environment}.
\end{definition}

Finally, for technical reasons we will need an ordering on extensional
values:

\begin{definition}
We define a relation $\sqsupseteq$ on extensional values of the same
type:
\begin{itemize}
\item For $d,b \in \pinterpret{\asort}$ with $\asort \in \Sorts$:
$d \sqsupseteq b$ if $d = b$.
\item For $(e_1,e_2),(u_1,u_2) \in \pinterpret{\atype \times \btype}$:
$(e_1,e_2) \sqsupseteq (u_1,u_2)$ if each $e_i \sqsupseteq u_i$.
\item For $A_\atype,B_\atype \in \pinterpret{\atype}$ with $\atype$
functional: $A_\atype \sqsupseteq B_\atype$ if for all $(e,u) \in
B$ there is $u' \sqsupseteq u$ such that $(e,u') \in A$.
\end{itemize}
\end{definition}

\paragraph{The algorithm.}
Let us now define our algorithm.
We will present it in a general form---including
a case \ref{alg:iterate:choice} which does not apply to deterministic
programs---so we can reuse the algorithm in the non-deterministic
settings to follow.

\begin{algorithm}\label{alg:base}\normalfont
Let $\prog$ be a fixed, deterministic cons-free program, and suppose
$\identifier{f}_1$ has a type
$\asortorpair_1 \arrtype \dots \arrtype \asortorpair_M \arrtype
\asortorpair \in \F$.

{\bf Input:} data expressions $d_1 : \asortorpair_1,\dots,d_M :
\asortorpair_M$.

{\bf Output:} The set of values $b$ with $\progresult$.

\begin{enumerate}
\item\label{alg:prepare} Preparation.
\begin{enumerate}
\item\label{alg:prepare:eprog}
  Let $\eprog$ be obtained from $\prog$ by the
  transformations of Lemma~\ref{lem:proper}, and by adding a clause
  $\apps{\symb{start}}{x_1}{x_M} = \apps{\identifier{f}_1}{x_1}{x_M}$
  for a fresh symbol
  $\symb{start}$ (so that $\progresult$ if{f} $\eprog \vdashcall
  \apps{\symb{start}}{d_1}{d_M} \arrr b$).
\item\label{alg:prepare:statements}
  Denote $\B := \B_{d_1,\dots,d_M}^\prog$ and let $\X$ be the set
  of all ``statements'':
  \begin{enumerate}
  \item $\vdash \apps{\identifier{f}}{e_1}{e_n} \leadsto o$ for
    (a) $\identifier{f} \in \Defineds$ with $\identifier{f} :
    \atype_1 \arrtype \dots \arrtype \atype_m \arrtype \asortorpair'
    \in \F$, (b) $0 \leq n \leq \arity(\identifier{f})$ such that
    $\typeorder{\atype_{n+1} \arrtype \dots \arrtype \atype_m
    \arrtype \asortorpair'} \leq K$, (c) $e_i
    \in \pinterpret{\atype_i}$ for $1 \leq i \leq n$ and (d) $o \in
    \pinterpret{\atype_{n+1} \arrtype \dots \arrtype \atype_m \arrtype
    \asortorpair'}$;
  \item $\eta \vdash t \leadsto o$ for (a) $\rho\colon
    \apps{\identifier{f}}{\ell_1}{\ell_k} = s$ a clause in $\eprog$,
    (b) $s \suptermeq t : \tau$, (c) $o \in \pinterpret{\btype}$ and
    (d) $\eta$ an ext-environment for $\rho$.
  \end{enumerate}
\item\label{alg:prepare:base}
  Mark statements of the form $\eta \vdash t \leadsto o$ in $\X$
  as confirmed if :
  \begin{enumerate}
  \item\label{alg:prepare:base:var}
    $t \in \V$ and $\eta(t) \sqsupseteq o$,
    \emph{or}
  \item\label{alg:prepare:base:constructor}
    $t = \apps{\identifier{c}}{t_1}{t_m}$ with $\identifier{c}
    \in \Constructors$ and $t\eta = o$.
  \end{enumerate}
  All statements not of either form are marked unconfirmed.
\end{enumerate}
\item\label{alg:iterate}
Iteration: repeat the following steps, until no further changes are made.
\begin{enumerate}
\item\label{alg:iterate:value}
  For all unconfirmed statements $\vdash \apps{\identifier{f}}{e_1}{
  e_n} \leadsto o$ in $\X$ with $n < \arity(\identifier{f})$: write
  $o = O_\atype$ and mark the statement as confirmed if for all
  $(e_{n+1},u) \in O$ there exists $u' \sqsupseteq u$ such that
  $\vdash \apps{\identifier{f}}{e_1}{e_{n+1}} \leadsto u'$
  is marked confirmed.
\item\label{alg:iterate:lhs}
  For all unconfirmed statements $\vdash \apps{\identifier{f}}{e_1}{
  e_k} \leadsto o$ in $\X$ with $k = \arity(\identifier{f})$:
  \begin{enumerate}
  \item\label{alg:iterate:lhs:find}
    find the first clause $\rho\colon
    \apps{\identifier{f}}{\ell_1}{\ell_k} = s$ in $\eprog$ that
    matches $\apps{\identifier{f}}{e_1}{e_k}$ and let $\eta$ be the
    matching ext-environment (if any);
  \item\label{alg:iterate:lhs:rhs}
    determine whether $\eta \vdash s \leadsto o$ is confirmed and if
    so, mark the statement $\apps{\identifier{f}}{e_1}{e_k} \leadsto
    o$ as confirmed.
  \end{enumerate}
\item\label{alg:iterate:ifte}
  For all unconfirmed statements of the form $\eta \vdash
  \ifte{s_1}{s_2}{s_3} \leadsto o$ in $\X$, mark
  the statement confirmed if
  \begin{enumerate}
  \item\label{alg:iterate:ifte:true}
    both $\eta \vdash s_1 \leadsto \strue$ and $\eta \vdash s_2
    \leadsto o$ are confirmed, or
  \item\label{alg:iterate:ifte:false}
    both $\eta \vdash s_1 \leadsto \sfalse$ and $\eta \vdash s_3
    \leadsto o$ are confirmed.
  \end{enumerate}
\item\label{alg:iterate:choice}
  For all unconfirmed statements $\eta \vdash
  \apps{\choice}{s_1}{s_n} \leadsto o$ in $\X$, mark the statement
  as confirmed if $\eta \vdash s_i \leadsto o$ for any $i \in \{1,
  \dots,n\}$.
\item\label{alg:iterate:pair}
  For all unconfirmed statements $\eta \vdash (s_1,s_2) \leadsto
  (o_1,o_2)$ in $\X$, mark the statement confirmed if both
  $\eta \vdash s_1 \leadsto o_1$ and $\eta \vdash s_2 \leadsto o_2$
  are confirmed.
\item\label{alg:iterate:rhs:var}
  For all unconfirmed statements $\eta \vdash \apps{x}{s_1}{s_n}
  \leadsto o$ in $\X$ with $x \in \V$, mark the statement as confirmed
  if there are $e_1 \in \pinterpret{\atype_1},\dots,e_n \in
  \pinterpret{\atype_n}$ such that each $\eta \vdash s_i \leadsto e_i$
  is marked confirmed, and there exists $o' \in \eta(x)(e_1,
  \dots,e_n)$ such that $o' \sqsupseteq o$.
\item\label{alg:iterate:rhs:func}
  For all unconfirmed statements $\eta \vdash \apps{\identifier{f}}{
  s_1}{s_n} \leadsto o$ in $\X$ with $\identifier{f} \in \Defineds$,
  mark the statement as confirmed if there are $e_1 \in
  \pinterpret{\atype_1},\dots,e_n \in \pinterpret{\atype_n}$ such that
  each $\eta \vdash s_i \leadsto e_i$ is marked confirmed, and:
  \begin{enumerate}
  \item\label{alg:iterate:rhs:call}
    $n \leq \arity(\identifier{f})$ and $\vdash \apps{\identifier{f}}{e_1}{
    e_n} \leadsto o$ is marked confirmed,
    \emph{or}
  \item\label{alg:iterate:rhs:extra}
    $n > k := \arity(\identifier{f})$ and
    there are $u,o'$ such that $\vdash \apps{\identifier{f}}{e_1}{
    e_k} \leadsto u$ is marked confirmed and $u(e_{k+1},\dots,e_n)
    \ni o' \sqsupseteq o$.
  \end{enumerate}
\end{enumerate}
\item\label{alg:completion}
  Completion: return $\{ b \mid b \in \B \wedge \vdash \apps{\symb{start}}{d_1}{d_M}
  \leadsto b$ is marked confirmed$\}$.
\end{enumerate}
\end{algorithm}

Note that, for programs of data order $0$, this algorithm closely
follows the earlier sketch.  Values of a higher type are abstracted to
deterministic extensional values.  The use of
$\sqsupseteq$ is needed because a value of higher type is associated
to many extensional values; e.g., to confirm a
statement $\vdash \symb{plus}\ \symb{3} \leadsto \{ (\symb{1},
\symb{4}), (\symb{0},\symb{3}) \}_{\nat \arrtype \nat}$ in some
program, it may be necessary to first confirm $\vdash \symb{plus}\ 
\symb{3} \leadsto \{ (\symb{0},\symb{3}) \}_{\nat \arrtype \nat}$.

The complexity of the algorithm relies on the following key
observation:

\edef\complexitylem{\number\value{lemma}}
\begin{lemma}\label{lem:complexitycore}
Let $\prog$ be a cons-free program of data order $K$.
Let $\Sigma$ be the set of all types $\atype$ with $\typeorder{\atype}
\leq K$ which occur as part of an argument type, or as an output type
of some $\identifier{f} \in \Defineds$.
Suppose that, given input of total size $n$,
$\pinterpret{\atype}$ has cardinality at most $F(n)$ for all $\atype
\in \Sigma$, and testing whether $e_1 \sqsupseteq e_2$ for $e_1,e_2
\in \interpret{\atype}$ takes at most $F(n)$ steps.
Then Algorithm~\ref{alg:base} runs in $\timecomp{a \cdot F(n)^b}$ for
some $a,b$.
\end{lemma}

Here, the cardinality $\Card(A)$ of a set $A$ is just the number of
elements of $A$.

\begin{proof}[Sketch]
Due to the use of $\eprog$,
all intensional values occurring in Algorithm~\ref{alg:base}
are in $\bigcup_{\atype \in \Sigma} \pinterpret{\atype}$.
Writing \textsf{a} for the greatest number of arguments any defined
symbol $\identifier{f}$ or variable $x$ in $\eprog$ may take
and \textsf{r} for the greatest number of sub-expressions of any
right-hand side in $\eprog$ (which is independent of the input!),
$\X$ contains at most $\textsf{a} \cdot |\Defineds|
\cdot F(n)^{\textsf{a}+1} + |\eprog| \cdot \textsf{r} \cdot
F(n)^{\textsf{a}+1}$ statements.
Since in all but the last step of the iteration at least one
statement is flipped from unconfirmed to confirmed, there are at most
$|\X|+1$ iterations, each considering $|\X|$ statements.  It is easy
to see that the individual steps in both the preparation and iteration
are all polynomial in $|\X|$ and $F(n)$, resulting in a polynomial
overall complexity.

(See Appendix~\ref{app:complexity} for the complete proof.)
\qed
\end{proof}

The result follows as $\Card(\pinterpret{\atype})$ is
given by a tower of exponentials in 
$\typeorder{\atype}$:

\edef\pcomplexitylem{\number\value{lemma}}
\begin{lemma}\label{lem:pinterpretcard}
If $1 \leq \Card(\B) < N$, then for each $\atype$ of length $L$
(where the length of a type is the number of sorts occurring in it,
including repetitions), with
$\typeorder{\atype} \leq K$: $\Card(\pinterpret{\atype}) <
\exp_2^K(N^L)$. 
Testing $e \sqsupseteq u$ for $e,u \in \pinterpret{\atype}$ takes
at most $\exp_2^K(N^{(L+1)^3})$ comparisons between elements of $\B$.
\end{lemma}

\begin{proof}[Sketch]
An easy induction on the form of $\atype$, using that
$\exp_2^K(X) \cdot \exp_2^K(Y) \leq \exp_2^K(X \cdot Y)$ for $X \geq
2$, and that for $A_{\atype_1 \arrtype \atype_2}$, each key $e \in
\pinterpret{\atype_1}$ is assigned one of
$\Card(\pinterpret{\atype_2})+1$ choices: an element $u$ of
$\pinterpret{\atype_2}$ such that $(e,u) \in A$, or non-membership.
The second part (regarding $\sqsupseteq$) uses the first.

(See Appendix~\ref{app:complexity} for the complete proof.)
\qed
\end{proof}

We will postpone showing correctness of the algorithm until
Section~\ref{subsec:proof}, where we can show the result together
with the one for non-deterministic programs.  Assuming correctness
for now, we may conclude:

\begin{lemma}\label{lem:deterministic:algorithm}
Every decision problem accepted by a deterministic cons-free program
$\prog$ with data order $K$ is in $\exptime{K}$.
\end{lemma}

\begin{proof}
We will see in Lemma~\ref{lem:basecorrectness} in
Section~\ref{subsec:proof} that $\progresult$ if and only if
Algorithm~\ref{alg:base} returns the set $\{b\}$.  For a program
of data order $K$, Lemmas~\ref{lem:complexitycore}
and~\ref{lem:pinterpretcard} together give that
Algorithm~\ref{alg:base} operates in $\timecomp{\exp_2^K(n)}$.
\qed
\end{proof}

\begin{theorem}\label{thm:deterministic}
The class of deterministic cons-free programs with data order $K$
characterises $\exptime{K}$ for all $K \in \N$.
\end{theorem}

\begin{proof}
A combination of Lemmas~\ref{lem:deterministic:simulate}
and~\ref{lem:deterministic:algorithm}.
\qed
\end{proof}

\section{Non-deterministic characterisations}\label{sec:elementary}

A natural question is what happens if we do not limit interest to
deterministic programs.  For data order $0$,
Bonfante~\cite{bon:06} shows that adding the choice operator to
Jones' language does not increase expressivity.   We will recover
this result for our generalised language in
Section~\ref{sec:nopartialvar}.
However, in the higher-order setting, non-deterministic choice
\emph{does} increase expressivity---dramatically so.  We have:

\smallskip\noindent
\begin{tabular}{c|c|c|c|c|c}
& \textbf{data order 0} &
\textbf{data order 1} &
\textbf{data order 2} &
\textbf{data order 3} & \dots
\\
\cline{1-6}
\vphantom{$2^{x^y}$}
\textbf{cons-free} &
$\pclass$ &
$\elementary$ &
$\elementary$ &
$\elementary$ & \dots
\\
\cline{1-6}
\end{tabular}

\medskip
As before, we will show the result---for data orders $1$ and
above---in two parts:
in Section~\ref{subsec:elementary:counting} we see that cons-free
programs of data order $1$ suffice to accept all problems in
$\elementary$; in Section~\ref{subsec:elementary:algorithm} we see
that they cannot go beyond.

\subsection{Simulating TMs using (non-deterministic)
cons-free programs}
\label{subsec:elementary:counting}

We start by showing how Turing Machines in $\elementary$ can be
simulated by non-deterministic cons-free programs.  For this, we
reuse the core simulation from Figure~\ref{fig:machine}.  The
reason for the jump in expressivity lies in
Lemma~\ref{lem:counting}: by taking advantage of non-determinism,
we can count up to arbitrarily high numbers.

\begin{lemma}\label{lem:nondetmodule}
If there is a $P$-counting module $C_\pi$ with data order $K \leq 1$,
there is a (non-deterministic) $(\lambda n.2^{P(n)-1})$-counting
module $C_{\er}$ with data order $1$.
\end{lemma}

\begin{proof}
We let $\numtype[\er] := \bool \arrtype \numtype$ (which has type
order $\max(1,\typeorder{\numtype})$), and:
\begin{itemize}
\item $\A_{\er}^n :=$ the set of those values $v : \numtype[\er]$
  such that:
  \begin{itemize}
  \item there is $w \in \A_\pi$ with $\numinterpret{w}_\pi^n = 0$
    such that $\prog_{\er} \vdashcall \app{v}{\strue} \arrr w$;
  \item there is $w \in \A_\pi$ with $\numinterpret{w}_\pi^n = 0$
    such that $\prog_{\er} \vdashcall \app{v}{\sfalse} \arrr w$;
  \end{itemize}
  and for all $1 \leq i < P(n)$ exactly one of the following holds:
  \begin{itemize}
  \item there is $w \in \A_\pi^n$ with $\numinterpret{w}_\pi^n = i$
    such that $\prog_{\er} \vdashcall \app{v}{\strue} \arrr w$;
  \item there is $w \in \A_\pi^n$ with $\numinterpret{w}_\pi^n = i$
    such that $\prog_{\er} \vdashcall \app{v}{\sfalse} \arrr w$;
  \end{itemize}
  We will say that $\app{v}{\strue} \mapsto i$ or
  $\app{v}{\sfalse} \mapsto i$ respectively.
\item $\numinterpret{v}_{\er}^n := \sum_{i=1}^{P(n)-1} \{
    2^{P(n)-1-i} \mid \app{v}{\strue} \mapsto i \}$;
\item $\prog_{\er}$ be given by
  Figure~\ref{fig:nondetcount} appended to $\prog_\pi$, and
  $\Defineds_{\er}$ by the symbols in $\prog_{\er}$.
\end{itemize}
So, we interpret a value $v$ as the number given by the bitstring
$b_1 \dots b_{P(n)-1}$ (most significant digit first), where
$b_i$ is $1$ if $\app{v}{\strue}$ evaluates to a value representing
$i$ in $C_\pi$, and $b_i$ is $0$ otherwise---so exactly if
$\app{v}{\sfalse}$ evaluates to such a value.
\qed
\end{proof}

\begin{figure}[!htb]
\vspace{-12pt}
-- core elements; $\symb{st}i\ n\ F$ sets bit $n$ in $F$ to the
value $i$ \\
$\symb{base}_{\er}\ x\ b = x$ \\
\begin{minipage}{0.5\linewidth}
$\symb{st1}_{\er}\ n\ F\ \strue = \choice\ n\ (F\ \strue)$ \\
$\symb{st1}_{\er}\ n\ F\ \sfalse = F\ \sfalse$ \\
\end{minipage}
\begin{minipage}{0.5\linewidth}
$\symb{st0}_{\er}\ n\ F\ \strue = \app{F}{\strue}$ \\
$\symb{st0}_{\er}\ n\ F\ \sfalse = \choice\ n\ (F\ \sfalse)$ \\
\end{minipage}
\vspace{-9pt}

-- testing bit values (using non-determinism and non-termination) \\
$\symb{bitset}_{\er}\ cs\ F\ i = \ifte{\:\symb{equal}_\pi\ cs\ (F\
  \strue)\ i\:}{\:\strue
  \\\phantom{\symb{bitset}_{\er}\ cs\ F\ i =\,}}{
  \ifte{\:\symb{equal}_\pi\ cs\ (F\ \sfalse)\ i\:}{\:\sfalse
  \\\phantom{\symb{bitset}_{\er}\ cs\ F\ i =\,}}{\:
  \symb{bitset}_{\er}\ cs\ F\ i}}$

-- the seed function \\
$\symb{nul}_\pi\ cs = \symb{nul'}_\pi\ cs\ (\seed\ cs)$ \\
$\symb{nul'}_\pi\ cs\ n = \ifte{\:\zero\ cs\ n\:}{\:n\:}{\:
  \symb{nul'}_\pi\ cs\ (\pred\ cs\ n)}$ \\
$\seed[\er]\ cs = \symb{seed'}_{\er}\ cs\ (\seed\ cs)\
  (\symb{base}_{\er}\ (\symb{nul}_\pi\ cs))$ \\
$\symb{seed'}_{\er}\ cs\ i\ F = \ifte{\:\zero\ cs\ i\:}{\:F\:}{\:
  \symb{seed'}_{\er}\ cs\ (\pred\ cs\ i)\ (\symb{st1}_{\er}\ i\ F)}$

-- the zero test \\
$\zero[\er]\ cs\ F = \symb{zero'}_{\er}\ cs\ F\ (\seed\ cs)$ \\
$\symb{zero'}_{\er}\ cs\ F\ i = \ifte{\:\zero\ i\:}{\:\strue
  \\\phantom{\symb{zero'}_{\er}\ cs\ F\ i =\,}
  }{\ifte{\:\symb{bitset}_{\er}\ cs\ F\ i\:}{\:\sfalse
  \\\phantom{\symb{zero'}_{\er}\ cs\ F\ i =\,}
  }{\:\symb{zero'}_{\er}\ cs\ F\ (\pred\ cs\ i)}}$

-- the predecessor \\
$\pred[\er]\ cs\ F = \symb{pr}_{\er}\ cs\ F\ (\seed\ cs)\
  (\symb{base}_{\er}\ (\symb{nul}_\pi\ cs))$ \\
$\symb{pr}_{\er}\ cs\ F\ i\ G =
  \ifte{\:\symb{bitset}_{\er}\ cs\ F\ i\:}{\:
  \symb{cp}_{\er}\ cs\ F\ (\pred\ cs\ i)\ (\symb{st0}_{\er}\ i\ G)
  \\\phantom{\symb{pr}_{\er}\ cs\ F\ i\ G =\ }
  }{\:\symb{pr}_{\er}\ cs\ F\ (\pred\ cs\ i)\ (\symb{st1}_{\er}\ i\ G)}$ \\
$\symb{cp}\ cs\ F\ i\ G = \ifte{\zero\ cs\ i\:}{\:G
  \\\phantom{\symb{cp}\ cs\ F\ i\ G =\,}}{
  \ifte{\:\symb{bitset}_{\er}\ cs\ F\ i\:}{\:
  \symb{cp}_{\er}\ cs\ F\ (\pred\ cs\ i)\ (\symb{st1}_{\er}\ i\ G)
  \\\phantom{\symb{cp}\ cs\ F\ i\ G =\,}
  }{\:\symb{cp}_{\er}\ cs\ F\ (\pred\ cs\ i)\ (\symb{st0}_{\er}\ i\ G)}
  }$
\caption{Clauses for the counting module $C_{\er}$.}
\vspace{-12pt}
\label{fig:nondetcount}
\end{figure}

To understand the counting program, consider $4$, with
bit representation $100$.  If $0,1,2,3$ are represented
in $C_\pi$ by values $O,w_1,w_2,w_3$ respectively, then
in $C_{\er}$, the number $4$ corresponds for example to $Q$:
\[
\symb{st1}\ w_1\ (\symb{st0}\ w_2\ (\symb{st0}\ w_3\ 
(\symb{base}_{\er}\ O)))
\]
The null-value $O$ functions as a default, and is a possible value
of both $\app{Q}{\strue}$ and $\app{Q}{\sfalse}$ for any function $Q$
representing a bitstring.

The non-determinism comes into play when determining whether
$\app{Q}{\strue} \mapsto i$ or not: we can evaluate $\app{F}{\strue}$ to
\emph{some} value, but this may not be the value we\linebreak need.
Therefore, we find some value of both $\app{F}{
\strue}$ and $\app{F}{\sfalse}$; if either represents $i$ in $C_\pi$, then
we have confirmed or rejected that $b_i = 1$.  If both
evaluations give a different value, we repeat the test.  This
gives a non-terminating program, but there is always exactly one
value $b$ such that $\prog_{\er} \vdashcall
\symb{bitset}_{\er[\pi]}\ cs\ F\ i \arrr b$.

The $\seed[\er]$ function generates the bit string $1\dots 1$, so the
function $F$ with $\app{F}{\strue} \mapsto i$ for all $i \in \{0,
\dots,P(n)-1\}$ and $\app{F}{\sfalse} \mapsto i$ for only $i = 0$.
The $\zero[\er]$ function iterates through $b_{P(n)-1},b_{P(n)-2},
\dots,b_1$ and tests whether all bits are set to $0$.  The clauses
for $\pred[\er]$ assume given a bitstring $b_1\dots b_{i-1}10\dots0$,
and recursively build $b_1 \dots b_{i-1}01\cdots 1$ in the parameter
$G$.

\begin{example}\label{ex:er}
Consider an input string of length 3, say $\sfalse\cons\sfalse\cons
\strue\cons\nil$.
Recall from Lemma~\ref{lem:module:pol} that there is a
$(\lambda n.n+1)$-counting module $C_{\linear}$ representing
$i \in \{0,\dots,3\}$ as suffixes of length $i$ from the input string.
Therefore, there is also a second-order $(\lambda n.2^n)
$-counting module $C_{\er[\linear]}$ representing $i \in \{0,\dots,
7\}$.
The number $6$---with bitstring $110$---is represented
by the value $w_6$:
\[
\begin{array}{c}
w_6 = \symb{st1}_{\er[\linear]}\ (\strue\cons\nil)\ (\ 
\symb{st1}_{\er[\linear]}\ (\sfalse\cons\strue\cons\nil)\ ( \\
\symb{st0}_{\er[\linear]}\ (\sfalse\cons\sfalse\cons\strue\cons\nil)\ 
(\ \symb{cons}_{\er[\linear]}\ \nil\ )\ )\ ) : \bool \arrtype \bits
\end{array}
\]
But then there is also a $(\lambda n.2^{2^n-1})$-counting module
$C_\eelin$, representing $i \in \{0,\dots,2^7-1\}$.  For example
97---with bit vector 1100001---is represented by:
\[
\begin{array}{c}
S = \symb{st1}_\eelin\ w_1\ (\ \symb{st1}_\eelin\ w_2\ 
(\ \symb{st0}_\eelin\ w_3\ (\\
\phantom{S =}\ \symb{st0}_\eelin\ w_4\ 
(\ \symb{st0}_\eelin\ w_5\ (\ \symb{st0}_\eelin\ w_6\ (\\
  \symb{st1}_\eelin\ w_7\ 
  (\ \symb{cons}_\eelin\ w_7\ )\ )\ )\ )\ )\ )\ )
\end{array}
\]
Here $\symb{st1}_\eelin$ and $\symb{st0}_\eelin$ have the type
$(\bool \arrtype \bits)
\arrtype (\bool \arrtype \bool \arrtype \bits) \arrtype \bool \arrtype
\bool \arrtype \bits$
and each $w_i$ 
represents $i$ in
$C_{\er[\linear]}$, as shown for $w_6$ above.
Note: $S\ \strue \mapsto w_1,w_2,w_7$ and $S\ \sfalse \mapsto
w_3,w_4,w_5,w_6$.
\end{example}

Since $2^{2^m-1}-1 \geq 2^m$ for all $m \geq 2$, we can count up to
arbitrarily high bounds using this module.  Thus, already with data
order $1$, we can simulate Turing Machines operating in
$\timecomp{\exp_2^K(n)}$ for any $K$.

\begin{lemma}\label{lem:nondetelementary}
Every decision problem in $\elementary$ is accepted by a
non-deterministic cons-free program with data order $1$.
\end{lemma}

\begin{proof}
A decision problem is in $\elementary$ if it is in some $\exptime{K}$
which, by Lemma~\ref{lem:counting}, is certainly the case if for any
$a,b$ there is a $Q$-counting module with $Q \geq \lambda n.\exp_2^K(
a \cdot n^b)$.  Such a module exists for data order $1$ by
Lemma~\ref{lem:nondetmodule}.
\qed
\end{proof}

\subsection{Simulating cons-free programs using an algorithm}
\label{subsec:elementary:algorithm}

Towards a characterisation, we must also see that every
decision problem accepted by a cons-free program is in $\elementary$%
---so that the result of every such program can be found by an
algorithm operating in $\timecomp{\exp_2^K(a \cdot n^b)}$
for some $a,b,K$.
We can reuse Algorithm~\ref{alg:base} by altering the definition of
$\pinterpret{\atype}$.

\begin{definition}\label{def:interpret}
Let $\B$ be a set of data expressions closed under $\supterm$.
For $\asort \in \Sorts$, let $\interpret{\asort} = \{ d \in \B \mid\ 
\vdash d : \asort \}$.  Inductively, define $\interpret{\atype \times
\btype} = \interpret{\atype} \times \interpret{\btype}$ and
$\interpret{\atype \arrtype \btype} = \{ A_{\atype \arrtype \btype}
\mid A \subseteq \interpret{\atype} \times \interpret{\btype} \}$.
We call the elements of any $\interpret{\atype}$
\emph{non-deterministic extensional values}.
\end{definition}

Where the elements of $\pinterpret{\atype \arrtype \btype}$ are
partial functions, 
$\interpret{\atype \arrtype \btype}$ contains arbitrary
relations: a value $v$ is associated to
a set of pairs $(e,u)$ such that $v\ e$ \emph{might} evaluate to $u$.
The notions of extensional expression, $e(u_1,\dots,u_n)$ and
$\sqsupseteq$ immediately extend to non-deterministic extensional
values.
Thus we can define:

\begin{algorithm}\label{alg:general}
Let $\prog$ be a fixed, non-deterministic cons-free program, with
$\identifier{f}_1 : \asortorpair_1 \arrtype \dots \arrtype
\asortorpair_M \arrtype \asortorpair \in \F$.

{\bf Input:} data expressions $d_1 : \asortorpair_1,\dots,d_M :
\asortorpair_M$.

{\bf Output:} The set of values $b$ with $\progresult$.

Execute Algorithm~\ref{alg:base}, but using $\interpret{
\atype}$ in place of $\pinterpret{\atype}$.
\end{algorithm}

In Section~\ref{subsec:proof}, we will see that indeed
$\progresult$ if and only if Algorithm~\ref{alg:general} returns a set
containing $b$.  But as before, we first consider complexity.  To
properly analyse this, we introduce the new notion of \emph{arrow
depth}.

\begin{definition}
A type's \emph{arrow depth} is given by:
$\mathit{depth}(\asort) = 0,\ 
\mathit{depth}(\atype \times \btype) = \max(\mathit{depth}(\atype),
  \mathit{depth}(\btype))$ and
$\mathit{depth}(\atype \arrtype \btype) = 1 +
  \max(\mathit{depth}(\atype),\mathit{depth}(\btype))$.
\end{definition}

Now the cardinality of each $\interpret{\atype}$ can be expressed
using its arrow depth:

\edef\ncomplexitylem{\number\value{lemma}}
\begin{lemma}\label{lem:interpretcard}
If $1 \leq \Card(\B) < N$, then for each $\atype$ of length $L$, with 
$\mathit{depth}(\atype) \leq K$:
$\Card(\interpret{\atype}) < \exp_2^K(N^L)$.
Testing $e \sqsupseteq u$ for $e,u \in \interpret{\atype}$ takes at
most $\exp_2^K(N^{(L+1)^3})$ comparisons.
\end{lemma}

\begin{proof}[Sketch]
A straightforward induction on the form of $\atype$, like
Lemma~\ref{lem:pinterpretcard}.

(See Appendix~\ref{app:complexity} for the complete proof.)
\qed
\end{proof}

Thus, once more assuming correctness for now, we may conclude:

\begin{lemma}\label{lem:nondeterministic:algorithm}
Every decision problem accepted by a non-deterministic
cons-free program $\prog$ is in $\elementary$.
\end{lemma}

\begin{proof}
We will see in Lemma~\ref{lem:generalcorrectness} in
Section~\ref{subsec:proof} that $\progresult$ if and only if
Algorithm~\ref{alg:general} returns a set containing $b$.
Since all types have an arrow depth and the set $\Sigma$ in
Lemma~\ref{lem:complexitycore} is finite, Algorithm~\ref{alg:general}
operates in some $\timecomp{\exp_2^K(n)}$.  Thus, the 
problem is in 
$\exptime{K} \subseteq \elementary$.
\qed
\end{proof}

\begin{theorem}\label{thm:nondeterministic}
The class of non-deterministic cons-free programs with data order $K$
characterises $\elementary$ for all $K \in \N \setminus \{0\}$.
\end{theorem}

\begin{proof}
A combination of Lemmas~\ref{lem:nondetelementary}
and~\ref{lem:nondeterministic:algorithm}.
\qed
\end{proof}

\subsection{Correctness proofs of Algorithms~\ref{alg:base}
and~\ref{alg:general}}\label{subsec:proof}

Algorithms~\ref{alg:base} and~\ref{alg:general} are the same---merely
parametrised with a different set of extensional values to be
used in step~\ref{alg:prepare:statements}.  Due to this
similarity, and because $\pinterpret{\atype} \subseteq \interpret{
\atype}$, we can largely combine their correctness proofs.
The proofs are somewhat intricate, however; all details are provided 
in~ Appendix~\ref{app:correctness}.

We begin with \emph{soundness}:

\edef\soundnesslem{\number\value{lemma}}
\begin{lemma}\label{lem:generalsound}
If Algorithm~\ref{alg:base} or~\ref{alg:general} returns a set $A \cup
\{b\}$, then $\progresult$.
\end{lemma}

\begin{proof}[Sketch]
We define for every value $v : \atype$ and $e \in \interpret{\atype}$:
$\down{v}{e}$ if{f}:
  (a)
  $\atype \in \Sorts$ and $v = e$;
  or (b)
  $\atype = \atype_1 \times \atype_2$ and $v = (v_1,v_2)$ and
  $e = (e_1,e_2)$ with $\down{v_1}{e_1}$ and $\down{v_2}{e_2}$;
  or (c)
  $\atype = \atype_1 \arrtype \atype_2$ and $e = A_\atype$ with
  $A \subseteq \{ (u_1,u_2) \mid u_1 \in
  \interpret{\atype_1} \wedge u_2 \in \interpret{\atype_2} \wedge$ for
  all values $w_1 : \atype_1$ with $\down{w_1}{u_1}$ there is some
  value $w_2 : \atype_2$ with $\down{w_2}{u_2}$ such that $\eprog
  \vdashcall \app{v}{w_1} \arrr w_2 \}$.

We now prove two statements together by induction on
the confirmation time in Algorithm~\ref{alg:base}, which
we consider equipped with \emph{unspecified} subsets $[\atype]$ of
$\interpret{\atype}$:

\begin{enumerate}
\item
Let:
  (a) $\identifier{f} : \atype_1 \arrtype \dots \arrtype \atype_m
  \arrtype \asortorpair \in \F$ be a defined symbol;
  (b) $v_1 : \atype_1,\dots,v_n : \atype_n$ be values, for
  $1 \leq n \leq \arity(\identifier{f})$;
  (c) $e_1 \in \interpret{\atype_1},\dots,e_n \in
  \interpret{\atype_n}$ be such that each $\down{v_i}{e_i}$;
  (d) $o \in \interpret{\atype_{n+1} \arrtype \dots \arrtype
  \atype_m \arrtype \asortorpair}$.
If $\vdash \apps{\identifier{f}}{e_1}{e_n} \leadsto o$ is
eventually confirmed, then $\eprog \vdashcall
\apps{\identifier{f}}{v_1}{v_n} \arrr w$ for some $w$ with
$\down{w}{o}$.
\item
Let:
  (a) $\rho\colon\app{\identifier{f}}{\vec{\ell}} = s$ be a clause
  in $\eprog$;
  (b) $t : \btype$ be a sub-expression of $s$;
  (c) $\eta$ be an ext-environment for $\rho$;
  (d) $\gamma$ be an environment such that $\down{\gamma(x)}{\eta(x)}$
  for all $x \in \Var(\app{\identifier{f}}{\vec{\ell}})$;
  (e) $o \in \interpret{\btype}$.
If the statement $\eta \vdash t \leadsto o$ is eventually confirmed,
then $\eprog,\gamma \vdash t \arrr w$ for some $w$ with
$\down{w}{o}$.
\end{enumerate}

Given the way $\eprog$ is defined from $\prog$, the lemma
follows from the first statement.  The induction is easy, but
requires minor sub-steps such as transitivity of $\sqsupseteq$.
\qed
\end{proof}

The harder part, where the algorithms diverge, is \emph{completeness}:

\edef\ncompletenesslem{\number\value{lemma}}
\begin{lemma}\label{lem:generalcomplete}
If $\progresult$, then Algorithm~\ref{alg:general} returns a set
$A \cup \{b\}$.
\end{lemma}

\begin{proof}[Sketch]
If $\progresult$, then $\eprog \vdashcall
\apps{\symb{start}}{d_1}{d_M} \arrr b$.
We label the nodes in the derivation trees with strings of numbers
(a node with label $\aindex$ has immediate subtrees of the form
$\aindex \cdot i$), and let $>$ denote lexicographic comparison of
these strings, and $\succ$ lexicographic comparison without prefixes
(e.g., $1 \cdot 2 > 1$ but not $1 \cdot 2 \succ 1$).  We define the
following function:
\begin{itemize}
\item $\psi(v,\aindex) = v$ if $v \in \B$, and
  $\psi((v_1,v_2),\aindex) = (\psi(v_1,\aindex),\psi(v_2,\aindex))$;
\item for $\apps{\identifier{f}}{v_1}{v_n} : \btype = \atype_{n+1}
  \arrtype \dots \arrtype \atype_m \arrtype \asortorpair$ with $m >
  n$, let $\psi(\apps{\identifier{f}}{v_1}{v_n},\aindex) =$ \\
  $\{ (e_{n+1},u) \mid \exists \cindex \succ \bindex >
    \aindex\ [$the subtree with index $\bindex$ has a root $\eprog
    \vdashcall \apps{\identifier{f}}{v_1}{v_{n+1}} \arrr w$ with
    $\psi(w,\cindex) = u$ and $e_{n+1} \ssupseteq \psi(v_{n+1},
    \bindex)] \}_\btype$.
\end{itemize}
Here, $\ssupseteq$ is defined the same as $\sqsupseteq$, except that
$A_\atype \ssupseteq B_\atype$ if{f} $A \supseteq B$.  Note that clearly
$A \ssupseteq B$ implies $A \sqsupseteq B$, and that $\ssupseteq$ is
transitive by transitivity of $\supseteq$.
Then, using induction on the labels of the tree in reverse
lexicographical order (so going through the tree right-to-left,
top-to-bottom), we can prove:
\begin{enumerate}
\item If the subtree labelled $\aindex$ has root $\eprog \vdashcall
  \apps{\identifier{f}}{v_1}{v_n} \arrr w$, then for all $e_1,\dots,
  e_n$ such that each $e_i \ssupseteq \psi(v_i,\aindex)$, and for all
  $\bindex \succ \aindex$ there exists $o \ssupseteq \psi(w,\bindex)$
  such that $\vdash \apps{\identifier{f}}{e_1}{e_n} \leadsto o$ is
  eventually confirmed.
\item
  If the subtree labelled $\aindex$ has root $\eprog,\gamma
  \vdash t \arrr w$ and $\eta(x) \ssupseteq \psi(\gamma(x),
  \aindex)$ for all $x \in \Var(t)$, then for all $\bindex \succ
  \aindex$ there exists $o \ssupseteq \psi(w,\bindex)$ such that
  $\eta \vvdash t \leadsto o$ is eventually confirmed.
\end{enumerate}
Assigning the main tree a label $0$ (to secure that $p \succ 0$
exists), we obtain that $\vdash \apps{\symb{start}}{d_1}{d_M} \leadsto
b$ is eventually confirmed, so $b$ is indeed returned.
\qed
\end{proof}

By Lemmas~\ref{lem:generalsound} and~\ref{lem:generalcomplete} together
we may immediately conclude:

\begin{lemma}\label{lem:generalcorrectness}
$\progresult$ if{f} Algorithm~\ref{alg:general} returns a set containing
$b$.
\end{lemma}


The proof of the general case provides a basis for the deterministic
case:

\edef\dcompletenesslem{\number\value{lemma}}
\begin{lemma}\label{lem:basecomplete}
If $\progresult$ and $\prog$ is deterministic, then
Algorithm~\ref{alg:base} returns a set $A \cup \{b\}$.
\end{lemma}

\begin{proof}[Sketch]
We define a consistency measure $\consistent{}{}$ on non-deterministic
extensional values: $\consistent{e}{u}$ if{f} $e = u \in \B$, or
$e = (e_1,e_2),\ u = (u_1,u_2),\ \consistent{e_1}{u_1}$ and
$\consistent{e_2}{u_2}$, or
$e = A_\atype,\ u = B_\atype$ and for all $(e_1,u_1) \in A$ and
$(e_2,u_2) \in B$: $\consistent{e_1}{e_2}$ implies
$\consistent{u_1}{u_2}$.

In the proof of Lemma~\ref{lem:generalcomplete}, we trace a derivation
in the algorithm.  In a deterministic program, we can see that if
both $\vdash \apps{\identifier{f}}{e_1}{e_n} \arrr o$ and
$\vdash \apps{\identifier{f}}{e_1'}{e_n'} \arrr o'$ are confirmed, and
each $\consistent{e_i}{e_n'}$, then $\consistent{o}{o'}$---and similar
for statements $\eta \vdash s \too o$.
We use this to remove statements which are not necessary, ultimately
leaving only those which use deterministic extensional
values as used in Algorithm~\ref{alg:base}.
\qed
\end{proof}

\begin{lemma}\label{lem:basecorrectness}
$\progresult$ if{f} Algorithm~\ref{alg:base} returns a set containing
$b$.
\end{lemma}

\begin{proof}
This is a combination of Lemmas~\ref{lem:generalsound}
and~\ref{lem:basecomplete}.
\qed
\end{proof}

Note that it is a priori not clear that
Algorithm~\ref{alg:base} returns only one value; however, this is
obtained as a consequence of Lemma~\ref{lem:basecorrectness}.

\section{Recovering the $\exptime{}$ hierarchy}\label{sec:nopartialvar}

While interesting, Lemma~\ref{lem:nondetelementary} exposes a
problem: non-determinism is unexpectedly powerful in the higher-order
setting.  If we still want to use non-deterministic programs towards
characterising non-deterministic complexity classes, we must surely
start by considering restrictions which avoid this explosion of
expressivity.

One direction is to consider \emph{arrow depth} instead of data order.
Using Lemma~\ref{lem:interpretcard}, we easily recover the original
hierarchy---and obtain the last line of Figure~\ref{fig:overview}.

\smallskip\noindent
\begin{tabular}{c|c|c|c|c}
& \textbf{arrow depth 0} &
\textbf{arrow depth 1} &
\textbf{arrow depth 2} &
\dots
\\
\cline{1-5}
\textbf{cons-free} &
$\pclass = \exptime{0}$ &
$\expclass = \exptime{1} $ &
$\exptime{2}$ & \dots
\\
\cline{1-5}
\end{tabular}

\edef\arrowdepththm{\number\value{theorem}}
\begin{theorem}\label{thm:arrowdepth}
The class of non-deterministic cons-free programs where all
variables are typed with a type of arrow depth $K$ characterises
$\exptime{K}$.
\end{theorem}

\begin{proof}[Sketch]
Both in the base program in Figure~\ref{fig:machine}, and in
the counting modules of Lemmas~\ref{lem:module:pol}
and~\ref{lem:module:exp}, type order and arrow depth coincide.  Thus
every decision problem in $\exptime{K}$ is accepted by a cons-free
program with ``data arrow depth'' $K$.
For the other direction, the proof of Lemma~\ref{lem:proper} is
trivially adapted to use arrow depth rather than type order.  Thus,
altering the preparation step in Algorithm~\ref{alg:general} gives an
algorithm which determines the possible outputs of a program with
data arrow depth $K$, with the desired complexity by
Lemma~\ref{lem:interpretcard}.
\qed
\end{proof}

A downside is that, by moving away from data order, this
result is hard to compare with other characterisations using cons-free
programs.  An alternative is to impose a restriction alongside
cons-freeness: \emph{unitary variables}.  This gives no restrictions
in the setting with data order $0$---thus providing the
first column in the table from Section~\ref{sec:elementary}---and
brings us the second-last line in Figure~\ref{fig:overview}:

\smallskip\noindent
\begin{tabular}{c|c|c|c|c}
& \textbf{data order 0} &
\textbf{data order 1} &
\textbf{data order 2} &
\textbf{data order 3} 
\\
\cline{1-5}
\textbf{cons-free} &
$\pclass =$ &
$\expclass =$ &
\multirow{2}{*}{$\exptime{2}$} &
\multirow{2}{*}{$\exptime{3}$} 
\\
\textbf{unitary variables} &
$\exptime{0}$ &
$\exptime{1}$ &
\\
\cline{1-5}
\end{tabular}

\smallskip
\begin{definition}\label{def:unaryvar}
A program $\prog$ has \emph{unitary variables} if clauses
are typed with an assignment mapping each variable $x$ to a type
$\asortorpair$ or $\atype \arrtype \asortorpair$, with
$\typeorder{\asortorpair} = 0$.
\end{definition}

Thus, in a program with unitary variables, a variable of a type
$(\bits \times \bits \times \bits) \arrtype \bits$ is admitted, but
$\bits \arrtype \bits \arrtype \bits \arrtype \bits$ is not.  The
crucial difference is that the former must be applied to all its
arguments at the same time, while the latter may be partially applied.
This avoids the problem of Lemma~\ref{lem:nondetmodule}.

\begin{theorem}\label{thm:unitary}
The class of (deterministic or non-deterministic) cons-free programs
with unitary variables of data order $K$ characterises $\exptime{K}$.
\end{theorem}

\begin{proof}[Sketch]
Both the base program in Figure~\ref{fig:machine} and the counting
modules of Lemmas~\ref{lem:module:pol} and~\ref{lem:module:exp} have
unitary variables, and are deterministic---this gives one direction.
For the other, let a recursively unitary type be $\asortorpair$
or $\atype \arrtype \asortorpair$ with $\typeorder{\asortorpair} = 0$
and $\atype$ recursively unitary.
The transformations of Lemma~\ref{lem:proper} are easily extended to
transform a program with unitary variables of type order $\leq K$ to one
where all (sub-)expressions have a recursively unitary type.
Since here data order and arrow depth are the same in this case, we
complete with Theorem~\ref{thm:arrowdepth}.
\qed
\end{proof}

\section{Conclusion and future work}

We have studied the effect of combining higher types and
non-determinism for cons-free programs.  This has resulted in the---%
highly surprising---conclusion that naively adding non-deterministic
choice to a language that characterises the $\exptime{K}$ hierarchy
for increasing data orders immediately increases the expressivity of
the language to $\elementary$.  Recovering a more fine-grained
complexity hierarchy can be done, but at the cost of further
syntactical restrictions. 

The primary goal that we will pursue in future work is to
use non-deterministic cons-free programs to characterise hierarchies
of \emph{non}-deterministic complexity classes such as
$\nexptime{K}$ for $K \in \N$.
In addition, it would be worthwhile to make a full study
of the ramifications of imposing restrictions on recursion, such as
tail-recursion or primitive recursion, in combination with
non-determinism and higher types (akin to the study of primitive
recursion in a successor-free language done in \cite{DBLP:journals/mst/KristiansenM12}).
We also intend to study
characterisations of
classes more restrictive than $\pclass$, such as $\logtime$ and
$\logspace$.

Finally, given the surprising nature of our results, we
urge readers to investigate the effect of adding non-determinism to
other programming languages used in implicit complexity that
manipulate higher-order data. We conjecture that the effect on
expressivity there will essentially be the same as what we have
observed.

\bibliography{references}

\begin{thebibliography}{10}

\bibitem{Bellantoni:thesis}
S.~Bellantoni.
\newblock PhD thesis, University of Toronto, 1993.

\bibitem{DBLP:journals/cc/BellantoniC92}
S.~Bellantoni and S.~Cook.
\newblock A new recursion-theoretic characterization of the polytime functions.
\newblock {\em Computational Complexity}, 2:97--110, 1992.

\bibitem{DBLP:conf/icalp/Ben-AmramP98}
A.~Ben{-}Amram and H.~Petersen.
\newblock {CONS}-free programs with tree input (extended abstract).
\newblock In {\em ICALP}, volume 1443 of {\em LNCS}, pages 271--282, 1998.

\bibitem{bon:06}
G.~Bonfante.
\newblock Some programming languages for logspace and ptime.
\newblock In {\em AMAST}, volume 4019 of {\em LNCS}, pages 66--80, 2006.

\bibitem{Clote:handbook}
P.~Clote.
\newblock Computation models and function algebras.
\newblock In {\em Handbook of Computability Theory}, pages 589--681. Elsevier,
  1999.

\bibitem{coo:71}
S.A. Cook.
\newblock Characterizations of pushdown machines in terms of time-bounded
  computers.
\newblock {\em journal of the ACM}, 18(1):4--18, 1971.

\bibitem{car:sim:14}
D.~de~Carvalho and J.~Simonsen.
\newblock An implicit characterization of the polynomial-time decidable sets by
  cons-free rewriting.
\newblock In {\em RTA-TLCA}, volume 8560 of {\em LNCS}, pages 179--193, 2014.

\bibitem{DBLP:journals/iandc/Goerdt92}
A.~Goerdt.
\newblock Characterizing complexity classes by general recursive definitions in
  higher types.
\newblock {\em Information and Computation}, 101(2):202--218, 1992.

\bibitem{DBLP:journals/tcs/Goerdt92a}
A.~Goerdt.
\newblock Characterizing complexity classes by higher type primitive recursive
  definitions.
\newblock {\em Theoretical Computer Science}, 100(1):45--66, 1992.

\bibitem{Immerman99descriptivecomplexity}
N.~Immerman.
\newblock {\em Descriptive Complexity}.
\newblock Springer-Verlag, 1999.

\bibitem{Jones:CompComp}
N.~Jones.
\newblock {\em Computability and Complexity from a Programming Perspective}.
\newblock MIT Press, 1997.

\bibitem{jon:01}
N.~Jones.
\newblock The expressive power of higher-order types or, life without {CONS}.
\newblock {\em Journal of Functional Programming}, 11(1):55--94, 2001.

\bibitem{DBLP:journals/jacm/KfouryTU94}
A.~J. Kfoury, J.~Tiuryn, and P.~Urzyczyn.
\newblock An analysis of {ML} typability.
\newblock {\em Journal of the {ACM}}, 41(2):368--398, 1994.

\bibitem{DBLP:conf/rta/KopS16}
C.~Kop and J.~Simonsen.
\newblock Complexity hierarchies and higher-order cons-free rewriting.
\newblock In {\em FSCD}, volume~52 of {\em LIPIcs}, pages 23:1--23:18, 2016.

\bibitem{DBLP:journals/mst/KristiansenM12}
L.~Kristiansen and B.M.W. Mender.
\newblock Non-determinism in g{\"{o}}del's system {T}.
\newblock {\em Theory of Computing Systems}, 51(1):85--105, 2012.

\bibitem{KRISTIANSEN2004139}
L.~Kristiansen and K.-H. Niggl.
\newblock Implicit computational complexity on the computational complexity of
  imperative programming languages.
\newblock {\em Theoretical Computer Science}, 318(1):139 -- 161, 2004.

\bibitem{DBLP:journals/njc/KristiansenV05}
L.~Kristiansen and P.J. Voda.
\newblock Programming languages capturing complexity classes.
\newblock {\em Nordic Journal of Computing}, 12(2):89--115, 2005.

\bibitem{DalLago2012}
U.~Dal Lago.
\newblock A short introduction to implicit computational complexity.
\newblock In {\em Lectures on Logic and Computation: ESSLLI 2010/2011}, pages
  89--109. 2012.

\bibitem{DBLP:journals/apal/Oitavem11}
I.~Oitavem.
\newblock A recursion-theoretic approach to {NP}.
\newblock {\em Annals of Pure and Applied Logic}, 162(8):661--666, 2011.

\bibitem{Papadimitriou:complexity}
C.~Papadimitriou.
\newblock {\em Computational Complexity}.
\newblock Addison-Wesley, 1994.

\bibitem{Sipser:comp}
M.~Sipser.
\newblock {\em Introduction to the Theory of Computation}.
\newblock Thomson Course Technology, 2006.

\end{thebibliography}

\clearpage
\appendix

\edef\savedcounter{\number0}
\newcommand{\startappendixcounters}{
\setcounter{lemma}{\savedcounter}
\renewcommand{\thelemma}{\Alph{section}\arabic{lemma}}
}
\newcommand{\oldcounter}[1]{
\edef\savedcounter{\number\value{lemma}}
\setcounter{lemma}{#1}
\renewcommand{\thelemma}{\arabic{lemma}}
}
\startappendixcounters

\noindent
\emph{This appendix contains full proofs of the results presented in the
text.}

\section{Matching expression and function order with data order
(Section~\ref{sec:semantics})}\label{app:properness}

In this first section, we consider Lemma~\ref{lem:proper}, which gives
a way to translate a propram which merely has data order $K$ to one
where all sub-expressions in all clauses have a type of order at most
$K$, and where for defined symbols $\identifier{f} : \atype_1 \arrtype
\dots \arrtype \atype_m \arrtype \asortorpair$ both each $\atype_i$
and $\asortorpair$ also have type order $\leq K$.

\medskip
The work in this appendix may initially seem to be rather more
detailed than necessary.  However, we must be very precise because we
will reuse the proofs to obtain the same results for \emph{arrow
depth} and \emph{unitary variables} in Appendix~\ref{app:arrowdepth}.
To easily combine these proofs, we define:

\begin{definition}\label{def:propertype}
In Appendix~\ref{app:properness}, a type $\atype$ is \emph{proper} if
$\typeorder{\atype} \leq K$ for some fixed non-negative integer $K$.
A program is proper if it is well-formed, and all clauses are typed so
that variables are assigned a proper type.

Note: types of order $0$ are proper, and $\atype \times \btype$ is
proper if and only if both $\atype$ and $\btype$ are proper.
\end{definition}

Henceforth, we will refer only to ``proper'' programs, not to type
orders.  This allows the lemmas to easily translate to different
notions of ``proper'' (which satisfy the two requirements mentioned
as notes.)

\medskip
A crucial insight to understand the lemma is that if an expression of
a certain type is used, then there has to be a variable of at least a
matching type order.

\begin{lemma}\label{lem:ordexpress}
Given a proper program $\prog$ and a derivation tree with root
$\prog,\gamma \vdash \apps{t}{s_1}{s_n} \arrr w$ (for $t$ an
expression),
the type of $t$ is $\atype_1 \arrtype \dots \arrtype \atype_n
\arrtype \btype$ with $\atype_i$ proper for all $1 \leq i \leq n$.
\end{lemma}

Here, we speak of \emph{the} type of $t$, because within the context
of the derivation, a unique type is associated to all expressions,
even variables.
Formally, we could for instance consider the type of a variable $x$ to
be the type of the value $\gamma(x)$.

\begin{proof}
By induction on $n$; for $n = 0$ there is nothing to prove.  For
larger $n$, note that $\apps{t}{s_1}{s_n}$ is an application, so the
result can only be derived by [Appl]:
\begin{prooftree}
\AxiomC{$\prog,\gamma \vdash \apps{t}{s_1}{s_{n-1}} \arrr w_1$}
\AxiomC{$\prog,\gamma \vdash s_n \arrr v_{i+1}$}
\AxiomC{$\prog \vdashcall \apps{\identifier{g}}{v_1}{v_{i+1}} \arrr
  w$}
\TrinaryInfC{$\prog,\gamma \vdash \apps{t}{s_1}{s_n} \arrr w$}
\end{prooftree}
Here, $w_1 = \apps{\identifier{g}}{v_1}{v_i}$.  By the induction
hypothesis on the first premise, each of $\atype_1,\dots,
\atype_{n-1}$ is typed properly; in addition, $v_{i+1} : \atype_n$
and since
$\arity(\identifier{g}) > i$ there must be a rule $\apps{\identifier{
g}}{\ell_1}{\ell_{i+1}} \cdots \ell_k = s$.  If $\atype_n$ is proper
we are done; otherwise, the pattern $\ell_{i+1}$ can only be a
variable or a pair.  Since at least one of the immediate subtypes of
an improper product type is also improper, $\ell_{i+1}$ must contain
a variable with an improper type; as this contradicts the properness
of $\prog$, indeed $\atype_n$ is proper.
\qed
\end{proof}

With this insight, we turn to a series of transformations, as used in
the proof sketch in the text.
The first step---increasing arities of clauses with certain output
types---is the hardest.  Essentially, we can increase the arities of
clauses whose type $\atype \arrtype \btype$ has order $> K$ because,
when a value $w$ of type $\atype \arrtype \btype$ is generated,
$w$ is eventually applied on some value $v$ of type $\atype$---and
not copied before that.

\begin{lemma}\label{lem:increasearity}
Given a proper program $\prog$,
let $\eprog$ be obtained from $\prog$ by replacing all clauses
$\identifier{f}\ \vec{\ell} = s$ where $s$ has an \emph{improper}
type $\atype \arrtype \btype$ with $\atype$ itself proper, by
$\identifier{f}\ \vec{\ell}\ x = s\ x$ for some fresh variable $x$.
Then $\eprog$ is a well-formed program with data order $K$, and
$\progresult$ if and only if
$\progeval{\eprog}{d_1,\dots,d_M} \mapsto b$.
\end{lemma}

\begin{proof}
The preservation of properness is clear, since $x$ can only have a
proper type $\atype$ and well-formedness is preserved because all
clauses with root symbol $\identifier{f}$ have the same type, so
are affected in the same way.

First, if $\progeval{\eprog}{d_1,\dots,d_M} \mapsto b$ then
$\progresult$ follows easily by induction on the size of the
derivation tree of $\progeval{\eprog}{d_1,\dots,d_M} \mapsto b$; the
only non-trivial step, an [Appl] where the third premise uses one of
the altered clauses, is handled by using the original rule instead
and shifting the subtrees around.

If $\progresult$, we get $\progeval{\eprog}{d_1,\dots,d_M} \mapsto
b$ by proving by induc\-tion on the derivation that $\eprog,\gamma \vdash
\apps{s}{t_1}{t_n} \arrr w_n$ if the following properties hold:
\begin{itemize}
\item $s : \atype_1 \arrtype \dots \arrtype \atype_n \arrtype \btype$
  with $\btype$ a proper type;
\item $t_1 : \atype_1,\dots,t_n : \atype_n$ are expressions and
  $v_1 : \atype_1,\dots,v_n : \atype_n$ are values;
\item $w_0,\dots,w_n$ are values with each $w_i : \atype_{i+1}
  \arrtype \dots \arrtype \atype_n \arrtype \btype$;
\item $\prog,\gamma \vdash s \arrr w_0$;
\item both $\eprog,\gamma \vdash t_i \arrr v_i$ and $\eprog \vdashcall
  w_{i-1}\ v_i \arrr w_i$ for all $1 \leq i \leq n$.
\end{itemize}

This gives the required result for $n = 0$ and $s = \apps{\identifier{
f}_1}{x_1}{x_M}$ (which has a type of order $0$).  Consider the rule
used to derive $\prog,\gamma \vdash s \arrr w_0$.

\begin{description}
\item[Instance] $w_0 = \gamma(s)$; then also $\eprog,\gamma \vdash s
  \arrr w_0$, and we complete with [Appl].
\item[Constructor] $n = 0$ and we complete with the
  induction hypothesis and [Constructor] (all sub-expressions $s_i$
  have a type of order $0$, which is proper).
\item[Pair] $n = 0$; we complete with the induction hypothesis and
  [Pair] (as the direct subtypes of a proper product types are also
  proper).
\item[Choice] $s = \apps{\choice}{s_1}{s_m}$ and the immediate subtree
  has root $\prog,\gamma \vdash s_j \arrr w_0$ for some $i$; by the
  induction hypothesis, $\eprog,\gamma \vdash \apps{s_j}{t_1}{t_n}
  \arrr w_n$.  If $n = 0$, then $\eprog,\gamma \vdash s \arrr w_n$ by
  [Choice].  Otherwise, $\eprog,\gamma \vdash \apps{s_j}{t_1}{t_n}
  \arrr w_n$ can only be obtained by [Appl]; thus, there are $w_0',
  \dots,w_n'=w_n$ and $v_1',\dots,v_n'$ such that $\eprog,\gamma
  \vdash s_j \arrr w_0'$ and $\eprog,\gamma \vdash t_i \arrr v_i'$ and
  $\eprog,\gamma \vdash w_{i-1}'\ v_i' \arrr w_i'$ for $1 \leq i \leq
  n$.  This gives first $\eprog,\gamma \vdash s \arrr w_0'$ by
  [Choice] and then $\eprog,\gamma \vdash \apps{s}{t_1}{t_n} \arrr
  w_n' = w_n$ by $n$ uses of [Appl].
\item[Conditional] Whether obtained by [If-True] or [If-False],
  this follows like [Choice].
\item[Function] $s = \identifier{f}$ and the immediate subtree has
  root $\prog \vdashcall \identifier{f} \arrr w_0$.
\begin{itemize}
\item If $\arity(\identifier{f}) > 0$, then $w_0 =
  \identifier{f}$ and since $\mathtt{arity}_{\eprog}(\identifier{f})
  \geq \arity(\identifier{f})$, also $\eprog,\gamma \vdash
  s \arrr w_0$; we complete with [Appl].
\item If $\arity(\identifier{f}) = \mathtt{arity}_{\eprog}(
  \identifier{f}) = 0$, then $\prog \vdashcall \identifier{f} \arrr
  w_0$ holds because $\prog,[] \vdash t \arrr w_0$ for some clause
  \pagebreak
  $\identifier{f} = t$.  Observing that $\prog,[x_1:=v_1,\dots,x_n:=
  v_n] \vdash \apps{t}{x_1}{x_n} \arrr w_n$ by the induction
  hypothesis, we follow the reasoning from the [Choice] case to
  obtain $w_0',\dots,w_n'=w_n$ such that $\eprog,\gamma \vdash
  s \arrr w_0'$ and $w_{i-1}'\ v_i \arrr w_i'$ for $1 \leq i \leq n$
  (here, the $v_i$ are unaltered since variables can only be
  evaluated in one way); we complete with [Appl] once more.
\item If $\arity(\identifier{f}) = 0$ but $\mathtt{arity}_{\eprog}(
  \identifier{f}) = 1$, then $n > 0$ since only clauses with an
  improper type were altered.  As $\prog,[] \vdash t \arrr w_0$
  for some clause $\identifier{f} = t$, the induction
  hypothesis gives $\prog,[x_1:=v_1,\dots,x_n:=v_n] \vdash
  \apps{t}{x_1}{x_n} \arrr w_n$.
  
  As in the [Choice] case, but considering $t\ x_1$ as the head, we
  find $w_1',\dots,w_n'=w_n$ such that $\eprog,[x_1:=v_1,\dots,x_n:=
  v_n] \vdash t\ x_1 \arrr w_1'$ and $w_{i-1}'\ v_i \arrr w_i'$ for
  $1 < i \leq n$.
  Since $x_2,\dots,x_n$ do not occur in $t\ x_1$, we can adapt the
  first of these trees to have a root $\eprog,[x_1:=v_1] \vdash t\ 
  x_1 \arrr w_1'$.

  Then we obtain $\eprog,\gamma \vdash \identifier{f}\ t_1 \arrr
  w_1'$ from the three subtrees $\eprog,\gamma \vdash \identifier{f}
  \arrr \identifier{f}$ (obtained using [Function] and [Closure]),
  $\eprog,\gamma \vdash t_1 \arrr v_1$ and $\eprog \vdashcall
  \identifier{f}\ v_1 \arrr w_1'$ (obtained using [Call] from
  $\eprog,[x_1:=v_1] \vdash t\ x_1 \arrr w_1'$).  Using this,
  $\apps{s\ v_1}{v_2}{v_n} \arrr w_n$ follows by [Appl].
\end{itemize}
\item[Appl] $s = s_1\ s_2$ and the immediate subtrees have roots
  $\prog,\gamma \vdash s_1 \arrr \apps{\identifier{g}}{v_1'}{v_j'}$
  and $\prog,\gamma \vdash s_2 \arrr v_0$ and $\prog \vdashcall
  \apps{\identifier{g}}{v_1'}{v_j'}\ v_0 \arrr w_0$.
  By Lemma~\ref{lem:ordexpress}, $s_2$ has a proper type, so we
  obtain $\eprog,\gamma \vdash s_2 \arrr v_0$ by the induction
  hypothesis.

  Now, using the same reasoning as with [Function], we obtain $w_0',
  \dots,w_n'=w_n$ such that $\eprog \vdashcall \apps{\identifier{g}}{
  v_1'}{v_j'}\ v_0 \arrr w_0'$ and $\eprog \vdashcall w_{i-1}'\ v_i
  \arrr w_i$ for $1 \leq i \leq n$ (if $\arity(\identifier{g}) = j+1$
  and $\mathtt{arity}_{\eprog}(\identifier{g}) = j+2$, then $w_0' =
  \apps{\identifier{g}}{v_1'}{v_j'}\ v_0$).  Now we may use the
  induction hypothesis on the subtree $\prog,\gamma \vdash s_1 \arrr
  \apps{\identifier{g}}{v_1'}{v_j'}$ to obtain
  $\eprog,\gamma \vdash \apps{s_1\ s_2}{t_1}{t_n} \arrr w_n$.
\qed
\end{description}
\end{proof}

Repeating the transformation of Lemma~\ref{lem:increasearity} until it
is no longer applicable, we obtain a proper program where all clauses
either have a proper type, or a type $\atype_{k+1} \arrtype \dots
\arrtype \atype_m \arrtype \asortorpair$ where $\atype_{k+1}$ is
improper; these latter clauses will be removed in one of the
following steps.

The next step is the removal of $\symb{if}$ and $\choice$ expressions
at the head of an application.  This is straightforward, and does not
depend on properness.

\begin{lemma}\label{lem:progifte}
Let $\eprog$ be obtained from $\prog$ by replacing all occurrences of
expressions $\apps{(\ifte{s_1}{s_2}{s_3})}{t_1}{t_n}$ with $n > 0$ in
the right-hand side of any clause by
$\ifte{s_1}{(\apps{s_2}{t_1}{t_n})}{(\apps{s_3}{t_1}{t_n})}$,
and by similarly replacing occurrences of
$\apps{(\apps{\choice}{s_1}{s_m})}{t_1}{t_n}$ with $n > 0$ by
$\apps{\choice}{(\apps{s_1}{t_1}{t_n})}{(\apps{s_m}{t_1}{t_n})}$.

Then, if $\prog$ is a proper program also $\eprog$ is a proper program
and $\progresult$ if{f} $\progeval{\eprog}{d_1,\dots,d_M} \mapsto b$.
\end{lemma}

\begin{proof}
Let $\mathit{fix}(s)$ be the result of replacing all
sub-expressions of the form $\apps{(\ifte{b}{s_1}{s_2})}{t_1}{t_n}$ in
$s$ by $\ifte{b}{(\apps{s_1}{t_1}{t_n})}{(\apps{s_2}{t_1}{t_n})}$, and
expressions $\apps{(\apps{\choice}{s_1}{s_m})}{t_1}{t_n}$ by
$\apps{\choice}{(\apps{s_1}{t_1}{t_n})}{(\apps{s_m}{t_1}{t_n})}$.
Then we see, by induction on the size of the derivation tree, that:
\begin{itemize}
\item $\prog \vdashcall \apps{\identifier{f}}{v_1}{v_n} \arrr w$ if{f}
  $\eprog \vdashcall \apps{\identifier{f}}{v_1}{v_n} \arrr w$, and
\item $\prog,\gamma \vdash s \arrr w$ if{f} $\prog,\gamma \vdash
  \mathit{fix}(s) \arrr w$.
\end{itemize}
The case where $s$ has one of the fixable forms merely requires
swapping some subtrees.
Typing is clearly not affected, nor the other properties of
well-formedness, and properness is unaltered because variables
are left alone.
\qed
\end{proof}

The effect of this step is to remove sub-expressions of a functional
type which, essentially, occur at the head of an application (and
therefore have a larger type than is necessary).  Since an expression
$\ifte{b}{s_1}{s_2}$ or $t = \apps{\choice}{s_1}{s_m}$ shares the
type of each $s_i$, this transformation guarantees that the
\emph{outermost} expression of a given improper type $\atype$
in a clause cannot occur as the direct sub-expression of an
$\ifte{}{}{}$ or $\choice$.  Thus, in a clause
$\apps{\identifier{f}}{\ell_1}{\ell_k} = s$ such an outermost
expression is either $s$ itself, or is $s_i$ in some context
$\apps{a}{s_1}{s_n}$ with $a \in \Defineds \cup \V$.
Following the transformation of Lemma~\ref{lem:increasearity}, the
former situation can only occur if $\identifier{f}$ has a type
$\atype_1 \arrtype \dots \arrtype \atype_m \arrtype \asortorpair$
where some $\atype_i$ or $\asortorpair$ is improper---%
symbols which we will remove in the final transformation.

Before that, however, we perform one further modification: we alter
clauses to remove those sub-expressions which cannot be used
following Lemma~\ref{lem:ordexpress}: if $\apps{t}{s_1}{s_n}$
occurs in the right-hand side of a clause and some $s_i$ has an
improper type, then this sub-expression can never occur in a
derivation tree.  Either the clause itself is never used, or the
sub-expression occurs below an $\symb{if}$ or $\choice$ which is
never selected.  Thus, we can safely replace those sub-expressions by
a fresh, unusable symbol.  This is done in Lemma~\ref{lem:properbot}.

\begin{lemma}\label{lem:properbot}
Given a proper program $\prog$, such that for all clauses
$\identifier{f}\ \vec{\ell} = s$ there is no sub-expression
$\apps{(\ifte{s_1}{s_2}{s_3})}{t_1}{t_n}$
or $\apps{(\apps{\choice}{s_1}{s_m})}{t_1}{t_n}$ with $n > 0$, let
$\eprog$ be obtained from $\prog$ by altering all clauses
$\apps{\identifier{f}}{\ell_1}{\ell_k} = s$ as follows: if $s
\suptermeq \apps{a}{s_1}{s_n} =: t$ where $a \in \V \cup \Defineds$
and some $s_i$ has an improper
type although $t$ itself has a proper type, and $t$ is the leftmost
outermost such sub-expression, then replace $t$ in the clause by a
fresh symbol $\bot_\atype$, typed $\bot_\atype : \atype$, and add a
clause $\bot_\atype = \bot_\atype$ (to ensure $\bot_\atype \in
\Defineds$ with arity $0$).

Then $\eprog$ is proper, and
$\progresult$ if{f} $\progeval{\eprog}{d_1,\dots,d_M}
\mapsto b$.
\end{lemma}

\begin{proof}
Replacing a sub-expression by a different one of the same type (but
potentially fewer variables) cannot affect well-formedness, and as
variables are left alone, properness of the variables types is not
altered.  For the
``only if'' part, note that the derivation tree for $\progresult$ has
no subtree with root $\prog,\gamma \vdash \apps{a}{s_1}{s_n} \arrr w$
by Lemma~\ref{lem:ordexpress}.  Therefore, $\apps{a}{s_1}{s_n}$ occurs
only as a \emph{strict sub-expression} of expressions in the tree for
$\progresult$, and may be replaced in all these places by
$\bot_\atype$ without consequence to obtain a derivation for
$\progeval{\eprog}{d_1,\dots,d_M} \mapsto b$.
Similarly, for the ``if'' part, the derivation tree for
$\progeval{\eprog}{d_1,\dots,d_M}
\mapsto b$ cannot have a subtree with root $\eprog,\gamma \vdash
\bot_\atype \arrr w$ since the only clause for $\bot_\atype$ does not
allow for such a conclusion.  Nor can the new clause $\bot_\atype =
\bot_\atype$ be used in it at all due to non-termination.
\qed
\end{proof}

Note that the transformation from Lemma~\ref{lem:properbot} is
terminating, as the size of the affected clause decreases; thus, it
can be repeated until no sub-expressions of the given form remain.
This gives step~\ref{lem:proper:occurrence} of the proof sketch of
Lemma~\ref{lem:proper}.

All in all, after these first three steps we still have a well-formed
program of the same data order, such that $\progresult$ can be
derived for exactly the same $d_1,\dots,d_M,b$.  For
step~\ref{lem:proper:removal}, we observe that the offending symbols
do not occur in any other clauses anymore.

\begin{lemma}\label{lem:properchangesworked}
Assume given a proper program $\prog$ such that for all
clauses $\apps{\identifier{f}}{\ell_1}{\ell_k} = s$ in $\prog$:
\begin{enumerate}
\item\label{lem:properworked:arity}
  if $\identifier{f} : \atype_1 \arrtype \dots \arrtype \atype_m
  \arrtype \asortorpair \in \F$ and all $\atype_i$ and $\asortorpair$
  are proper, then $\atype_{k+1} \arrtype \dots \arrtype \atype_m
  \arrtype \asortorpair$ is proper as well;
\item\label{lem:properworked:ifchoice}
  $s$ does not have a sub-expression of the form
  $\apps{(\ifte{s_1}{s_2}{s_3})}{t_1}{t_n}$ or
  $\apps{(\apps{\choice}{s_1}{s_m})}{t_1}{t_n}$ with $n > 0$;
\item\label{lem:properworked:nosub}
  $s$ does not have a sub-expression $t = \apps{a}{s_1}{s_n}$ with $a
  \in \V \cup \Defineds$ where $t$ itself has a proper type, but
  with some $s_i$ having an improper type.
\end{enumerate}
Let $\mathit{Bad}$ be the set of defined symbols $\identifier{g}$
which are assigned a type $\atype_1 \arrtype \dots \arrtype \atype_m
\arrtype \asortorpair$ such that some $\atype_i$ or $\asortorpair$ is
improper.
Then for all clauses $\apps{\identifier{f}}{\ell_1}{\ell_k} = s$ in
$\prog$ with $\identifier{f} \notin \mathit{Bad}$: none of the
symbols in $\mathit{Bad}$ occur in $s$, and $s$ does not have any
sub-expressions whose type has an order $> K$.
\end{lemma}

Note that assumption~\ref{lem:properworked:arity} is given by the
transformation of Lemma~\ref{lem:increasearity},
assumption~\ref{lem:properworked:ifchoice} is given by the
transformation of Lemma~\ref{lem:progifte}, and
assumption~\ref{lem:properworked:nosub} is given by the transformation
of Lemma~\ref{lem:properbot}.

\begin{proof}
We first observe: if $\identifier{f} \notin \mathit{Bad}$, then by
the first assumption, the type of $s$ is proper.  For such
$s$, which moreover does not have $\symb{if}$ or $\choice$
expressions at the head of an application, we prove by induction that
$s$ does not use elements of $\mathit{Bad}$.
If $s = \ifte{s_1}{s_2}{s_3}$ of $s = \apps{\choice}{s_1}{s_m}$, then
each $s_i$ has a proper type (either the type of $s$ or $\bool$), so
we complete by induction.
If $s = \apps{\identifier{c}}{s_1}{s_m}$ with $\identifier{c} \in
\Constructors$, each $s_i$ has a type of order $0$, which is proper.
Otherwise, $s = \apps{a}{s_1}{s_n}$ with $a \in \V \cup \Defineds$.
By the third assumption, all $s_i$ have a proper type, so no bad
symbols occur in them by the induction hypothesis.  Moreover, if $a
\in \Defineds$ and $a : \btype_1 \arrtype \dots \arrtype \btype_n
\arrtype \ctype \in \F$, then each $\btype_i$ is proper by that
same assumption, and $\ctype$ is the type of $s$, so is proper.
Thus, also $a \notin \mathit{Bad}$.
\qed
\end{proof}

Now, a trivial induction shows that the derivation of any conclusion
of the form $\progresult$ cannot use a clause with a bad root symbol;
removing these clauses therefore has no effect.  As the bad symbols
do not occur at all in the remaining clauses, the symbols can also
be safely removed.  We conclude:

\oldcounter{\propernesslem}
\begin{lemma}
Given a well-formed program $\prog$ with data order $K$, there is a
well-formed program $\eprog$ such that
$\progresult$ iff $\progeval{\eprog}{d_1,\dots,d_M} \mapsto b$
for any $b_1,\dots,b_M,d$ and:
(a) all defined symbols in $\eprog$ have a type $\atype_1 \arrtype
  \dots \arrtype \atype_m \arrtype \asortorpair$ such that both
  $\typeorder{\atype_i} \leq K$ for all $i$ and
  $\typeorder{\asortorpair} \leq K$, and
(b) in all clauses, all sub-expressions of the right-hand side have a
  type of order $\leq K$ as well.
\end{lemma}
\startappendixcounters

\begin{proof}
We apply the transformations from
Lemmas~\ref{lem:increasearity}--\ref{lem:properbot} and then remove
all ``bad'' symbols and corresponding clauses following
Lemma~\ref{lem:properchangesworked}, as described in the text above.
As we have seen, the resulting program $\eprog$ is still proper,
which means that it is well-formed and has the same data order $K$;
in addition, it has properties (a) and (b) by
Lemma~\ref{lem:properchangesworked}.
\qed
\end{proof}

In Appendix~\ref{app:arrowdepth} we will use variations of
Lemma~\ref{lem:proper} for other notions of ``proper''.

\section{Properties of cons-free programs (Section~\ref{sec:consfree})}
\label{app:consfree}

Lemma~\ref{lem:safetysimple} demonstrates that any data encountered
during the execution of a cons-free program was either part of the
input, or occurred directly as part of a clause; that is, every such
data expression is in the set $\B_{d_1,\dots,d_M}^\prog$.

As a helper result, we start by proving that \emph{patterns} occurring
in clauses can only be instantiated to data.  This is important to
demonstrate the harmlessness of allowing sub-expressions of the
left-hand sides of clauses to occur on the right.

\begin{lemma}\label{lem:patterninstantiate}
Let $T$ be a derivation tree with root $\prog,\gamma \vdash s \arrr
w$.  If $s$ a pattern, then $s\gamma = w$.
\end{lemma}

\begin{proof}
By induction on the form of $T$.
The roots of [Function], [Choice] and [Conditional] have the wrong
shape. [Instance] immediately gives the required result, and the cases
for [Constructor] and [Pair] follow by the induction hypothesis.

Finally, we show by induction on $n$ that [Appl] is not applicable:
if $n = 1$, then [Appl] requires a subtree
$\prog,\gamma \vdash \identifier{c} \arrr \identifier{f}$ with
$\identifier{f} \in \Defineds$, for which there are no inference
rules.  If $n > 1$, then [Appl] requires a subtree $\prog,\gamma
\vdash \apps{\identifier{c}}{s_1}{s_{n-1}} \arrr
\apps{\identifier{f}}{v_1}{v_i}$ which, by the induction hypothesis,
must be obtained by an inference rule other than [Appl]; again, there
are no suitable inference rules.
\qed
\end{proof}

Rather than immediately proving Lemma~\ref{lem:safetysimple}, we will
present---in Lemma~\ref{lem:safety}---a variation which gives a little
more information.  As this result will be used in some of the later
proofs in the appendix, it pays to be precise.  First we define a
notion of value which is limited to data in
$\B_{d_1,\dots,d_M}^\prog$.

\begin{definition}\label{def:VValue}
Fixing a program $\prog$ and data expressions $d_1,\dots,d_M$, let the
set $\VValue$ be given by the grammar:

\begin{tabular}{rcl}
$v,w \in \VValue$ & ::= & $d \in \B_{d_1,\dots,d_M}^\prog
\mid (v,w) \mid \apps{f}{v_1}{v_n}$ ($n < \arity(f)$) \\
\end{tabular}
\end{definition}

Note that clearly $\texttt{Value} \subseteq \VValue$.  The following
lemma makes the notion ``the only data expressions encountered during
the execution of a cons-free program $\prog$ are in
$\B_{d_1,\dots,d_M}^\prog$'' precise, by requiring all values to be
in $\VValue$:

\begin{lemma}\label{lem:safety}
Let $T$ be a derivation tree for $\progresult$.  Then for all
subtrees $T'$ of $T$:
\begin{itemize}
\item if $T'$ has root $\prog,\gamma \vdash s \arrr w$, then both $w$
and all $\gamma(x)$ are in $\VValue$;
\item if $T'$ has root $\prog,\gamma \vdashif d, s_1, s_2 \arrr w$,
then $d \in \B_{d_1,\dots,d_M}^\prog$ and both $w$ and all
$\gamma(x)$ are in $\VValue$;
\item if $T'$ has a root $\prog \vdashcall \apps{\identifier{f}}{v_1}{
v_n} \arrr w$ with $\identifier{f} \in \Defineds$, then both $w$ and
all $v_i$ are in $\VValue$;
\item if $T'$ has a root $\prog,\gamma \vdash
\apps{\identifier{c}}{s_1}{s_m} \arrr
\apps{\identifier{c}}{b_1}{b_m}$ with $\identifier{c} \in
\Constructors$, then each $s_i\gamma = b_i \in \Data$ and
$\apps{\identifier{c}}{b_1}{b_m} \in \B_{d_1,\dots,d_M}^\prog$.
\end{itemize}
\end{lemma}

\begin{proof}
For brevity, let $\B := \B_{d_1,\dots,d_M}^\prog$.
We show by induction on the depth of $T'$ that: (**) the properties
in the lemma statement hold for both $T'$ and all its strict subtrees
if $\treeroot(T')$ has one of the following forms:
\begin{itemize}
\item $\prog,\gamma \vdash s \arrr w$ with all $\gamma(x) \in
  \VValue$, and $t\gamma \in \B$ for all sub-expressions
  $t \subtermeq s$ such that $t = \apps{\identifier{c}}{s_1}{s_m}$
  for some $\identifier{c} \in \Constructors$;
\item $\prog,\gamma \vdashif d,s_1,s_2 \arrr w$ with $d \in
  \B$ and all $\gamma(x) \in \VValue$, and $t\gamma \in \B$
  for all 
  $t \subtermeq s_1$ or $t \subtermeq s_2$
  such that $t = \apps{\identifier{c}}{s_1}{s_m}$ for some
  $\identifier{c} \in \Constructors$;
\item $\prog \vdashcall \apps{f}{v_1}{v_n} \arrr w$ with all
  $v_i \in \VValue$.
\end{itemize}
Note that proving this suffices: the immediate subtree $T'$ of $T$
has a root $\prog,\gamma \vdash \apps{\identifier{f}_1}{x_1}{x_M}
\arrr b$, where each $\gamma(x_i) = d_i \in \B$, and
$\apps{\identifier{f}_1}{x_1}{x_M}$ has no sub-expressions with a
data constructor at the head.  Thus, the lemma holds for both $T'$ and
all its strict subtrees, which implies that it holds for $T$.

We prove (**).  Assume that $\treeroot(T')$ has one of the
given forms, and consider the rule used to obtain this root.
\begin{description}
\item[Instance] Then $T'$ has a root $\prog,\gamma \vdash x \arrr
\gamma(x)$; the requirement that all $\gamma(y) \in \VValue$
is satisfied by the assumption, and this also gives that the
right-hand side $\gamma(x) \in \VValue$.
\item[Function] Then $T'$ has a root $\prog,\gamma \vdash 
\identifier{f} \arrr v$ and a subtree $\prog \vdashcall
\identifier{f} \arrr v$; by the induction hypothesis, the properties
hold for this subtree, which also implies that $v \in
\VValue$ and therefore the properties hold for $T'$ as well.
\item[Constructor] Then $T'$ has a root $\prog,\gamma \vdash
\apps{\identifier{c}}{s_1}{s_m} \arrr \apps{\identifier{c}}{b_1}{
b_m}$ with $\identifier{c} \in \Constructors$, and the immediate
subtrees have the form $\prog,\gamma \vdash s_i \arrr b_i$; by the
induction hypothesis (and the assumption), the properties are
satisfied for each such subtree.  Also by the assumption,
$(\apps{\identifier{c}}{s_1}{s_m})\gamma \in \B$, so necessarily
each $s_i\gamma \in \B \subseteq \Data$.  By
Lemma~\ref{lem:patterninstantiate}, each $s_i\gamma =
b_i$, and $\apps{\identifier{c}}{b_1}{b_m} =
(\apps{\identifier{c}}{s_1}{s_m})\gamma \in \B$.
\item[Pair] Then $T'$ has a root $\prog,\gamma \vdash (s_1,s_2) \arrr
(w_1,w_2)$ and subtrees with roots $\prog,\gamma \vdash s_1 \arrr
w_1$ and $\prog,\gamma \vdash s_2 \arrr w_2$.  The assumption and
induction hypothesis give that the properties are satisfied for both
subtrees, and therefore both $w_1$ and $w_2$ are in
$\VValue$, giving also $(w_1,w_2) \in \VValue$.
\item[Choice] Then $T'$ has a root $\prog,\gamma \vdash \apps{
\choice}{s_1}{s_n} \arrr v$ and a subtree $\prog,\gamma \vdash s_i
\arrr v$ for some $i$.  By the induction hypothesis, the properties
hold for the subtree, and therefore $v \in \VValue$.
\item[Conditional] Then $T'$ has a root $\prog,\gamma \vdash
\ifte{s_1}{s_2}{s_3} \arrr w$ and subtrees with roots
$\prog,\gamma \vdash s_1 \arrr d$ and $\prog,\gamma \vdashif
d,s_2,s_3 \arrr w$.  The requirement that all $\gamma(x) \in
\VValue$ is satisfied by the assumption, and by both the
assumption and the induction hypothesis, the lemma is satisfied for
the first subtrees.  Thus, $d \in \VValue$; for typing
reasons $d \in \B$.  We may apply the induction hypothesis on
the second subtree, which gives that the lemma is satisfied for it,
and that $w \in \VValue$.
\item[If-True or If-False] Then $\treeroot(T')$ has the form
$\prog,\gamma \vdashif d,s_1,s_2 \arrr w$.  The requirement that $d
\in \B$ and all $\gamma(x) \in \VValue$ is satisfied by the
assumption.  $T'$ has one immediate subtree $T''$, whose root is
either $\prog,\gamma \vdash s_2 \arrr w$ or $\prog,\gamma \vdash
s_3 \arrr w$.  Since the assumptions are satisfied, $T''$ satisfies
the lemma by the induction hypothesis, which also gives that $w
\in \VValue$.
\item[Appl] Then $T'$ has a root $\prog,\gamma \vdash s\ t
\arrr w$ and subtrees $\prog,\gamma \vdash s \arrr
\apps{\identifier{f}}{v_1}{v_n}$ and $\prog,\gamma \vdash t \arrr
v_{n+1}$ and $\prog \vdashcall \apps{\identifier{f}}{v_1}{v_{n+1}}
\arrr w$.  The assumption gives that all $\gamma(x) \in
\VValue$, and the assumption and induction hypothesis
together give that the lemma is satisfied for the first two
subtrees.  Since this implies that all $v_i \in \VValue$, we
may also apply the induction hypothesis on the last subtree, which
gives that $w \in \VValue$.
\item[Closure] Then $\treeroot(T')$ has the form $\prog,\vdashcall
\apps{\identifier{f}}{v_1}{v_n} \arrr w$ with $\identifier{f} \in
\Defineds$.  All $v_i$ are in $\VValue$ by the assumption;
thus, $w = \apps{f}{v_1}{v_n}
\in \VValue$ as well, and there are no strict subtrees.
\item[Call] Then $\treeroot(T')$ has the form $\prog,
\vdashcall \apps{\identifier{f}}{v_1}{v_k} \arrr w$ with
$\identifier{f} \in \Defineds$, and there exist a clause
$\apps{\identifier{f}}{\ell_1}{\ell_k} = s$ and an environment
$\gamma$ with domain $\Var(\apps{\identifier{f}}{\ell_1}{\ell_k})$
such that each $v_i = \ell_i\gamma$, and $T'$ has one immediate
subtree $T''$ with root $\prog,\gamma \vdash s \arrr w$.  Then, for
$1 \leq i \leq n$ we observe that $v_i \suptermeq \gamma(x)$ for all
$x \in \Var(\ell_i)$, since (by definition of a pattern) $\ell_i
\suptermeq x$ for all such $x$.  Since all sub-expressions of
a value in $\VValue$ are themselves in $\VValue$, we thus have: each
$\gamma(x) \in \VValue$.

Moreover, for $s \suptermeq t = \apps{\identifier{c}}{s_1}{s_m}$
with $\identifier{c} \in \Constructors$, also $\ell_i \suptermeq t$
for some $i$ by definition of cons-free.  But then also
$\ell_i\gamma = v_i \suptermeq t\gamma$.
Thus, we can apply the induction hypothesis, and obtain that the
lemma is satisfied for $T''$.  This implies that $w \in
\VValue$, so the last requirement on the root of $T'$ is
satisfied.
\qed
\end{description}
\end{proof}

It remains to prove Lemma~\ref{lem:safetysimple} from the
text---which is just a (slightly weakened) reformulation
of Lemma~\ref{lem:safety}.

\oldcounter{\safetypreservelem}
\begin{lemma}
Let $\prog$ be a cons-free program, and suppose that $\progresult$ is
obtained by a derivation tree $T$.  Then for all statements $\prog,
\gamma \vdash s \arrr w$ or $\prog,\gamma \vdashif b',s_1,s_2 \arrr w$
or $\prog \vdashcall \apps{\identifier{f}}{v_1}{v_n} \arrr w$ in $T$,
and all expressions $t$ such that (a) $w \suptermeq t$, (b) $b'
\suptermeq t$, (c) $\gamma(x) \suptermeq t$ for some $x$ or (d) $v_i
\suptermeq t$ for some $i$:
if $t$ has the form $\apps{\identifier{c}}{b_1}{b_m}$ with
$\identifier{c} \in \Constructors$, then $t \in
\B_{d_1,\dots,d_M}^\prog$.
\end{lemma}
\startappendixcounters

\begin{proof}
Immediately by Lemma~\ref{lem:safety}, as the
only sub-expressions of an element of $\VValue$ with a data
constructor as head symbol, are in $\B_{d_1,\dots,d_M}^\prog$.
\qed
\end{proof}

\section{Counting modules (Section~\ref{subsec:counting})}\label{app:counting}

We discuss the counting modules from Section~\ref{subsec:counting}
in more detail.
To start, we use the ideas of Example~\ref{ex:counting} to create
counting modules surpassing any polynomial.

\oldcounter{\modulepol}
\begin{lemma}
For any $a,b \in \N \setminus \{0\}$, there is a $(\lambda n.a \cdot
(n+1)^b)$-counting module $C_\pol$ with data order $0$.
\end{lemma}
\startappendixcounters

\begin{proof}
Using pairing in a right-associative way---so $(x,y,z)$ should
be read as $(x,(y,z))$---we let:
\begin{itemize}
\item $\numtype[\pol] := \bits^{b+1}$; that is,
$\bits \times \dots \times \bits$ with $b+1$ occurrences of
$\bits$
\item $\A_\pol^n := \{ (d_0,\dots,d_b) \mid$ all $d_i$ are boolean
lists, with $|d_0| < a$ and $|d_i| \leq n$ for $1 \leq i \leq b$;
here, we say $|x_1\cons \dots \cons x_k \cons \nil| = k$
\item $\numinterpret{(d_0,\dots,d_b)}_\pol^n := \sum_{i = 0}^b
|d_i| \cdot (n+1)^{b-i}$
\item $\Defineds_\pol = \{ \seed[\pol], \pred[\pol],
\zero[\pol] \}$
\item let $\symb{alist}$ be a list of length $a-1$, e.g.,
$\sfalse\cons\dots\cons\sfalse\cons\nil$ and let
$\prog_\pol$ consist of the following clauses:

\medskip
$\seed[\pol]\ cs = (\symb{alist},cs,\dots,cs)$ \\
\ \\
$\pred[\pol]\ cs\ (x_0,\dots,x_{b-1},y\cons ys) =
  (x_0,\dots,x_{b-1},ys)$ \\
$\pred[\pol]\ cs\ (x_0,\dots,x_{b-2},y\cons ys,\nil) =
  (x_0,\dots,x_{b-2},ys,cs)$ \\
\dots \\
$\pred[\pol]\ cs\ (y\cons ys,\nil,\dots,\nil) = (ys,cs,\dots,cs)$ \\
$\pred[\pol]\ cs\ (\nil,\nil,\dots,\nil) = (\nil,\nil,\dots,\nil)$ \\
\ \\
$\zero[\pol]\ cs\ (x_0,\dots,x_{b-1},y\cons ys) = \sfalse$ \\
$\zero[\pol]\ cs\ (x_0,\dots,x_{b-2},y\cons ys,\nil) = \sfalse$ \\
\dots \\
$\zero[\pol]\ cs\ (y\cons ys,\nil,\dots,\nil) = \sfalse$ \\
$\zero[\pol]\ cs\ (\nil,\dots,\nil) = \strue$
\end{itemize}
It is easy to see that the requirements on evaluation are satisfied.
For example, $\prog_\pol \vdashcall \seed[\pol]\ cs \arrr
(\symb{alist},cs,\dots,cs)$,
which consists of $b+1$ boolean lists with the right
lengths, and $\numinterpret{(\symb{alist},cs,\dots,cs)}_\pol^n = (a-1
) \cdot (n+1)^b + n \cdot (n+1)^{b-1} \linebreak
+ \dots + n \cdot (n+1)^{b-b} =
(a \cdot (n+1)^b - (n+1)^b) + ((n+1)^b - (n+1)^{b-1}) + \cdots +
((n+1)^1 -\linebreak (n+1)^0) = a \cdot (n+1)^b - 1$; as the program
is deterministic, this is the only possible result.
The requirements for $\pred[\pol]$ and $\zero[\pol]$ are
similarly easy.
\qed
\end{proof}

Note that the clauses in $\prog_\pol$ correspond to those in
Example~\ref{ex:counting}.  The other counting module of
Section~\ref{subsec:counting} allows us to build on an existing
counting module so as to obtain an exponential increase in magnitude
of the boundary $P$---and to be applied repeatedly for arbitrarily
high bounds.

\oldcounter{\moduleexp}
\begin{lemma}
If there is a $P$-counting module $C_\pi$ of data order $K$, then
there is a $(\lambda n.2^{P(n)})$-counting module $C_\epi$ of data
order $K+1$.
\end{lemma}
\startappendixcounters

\begin{proof}
We let:
\begin{itemize}
\item $\numtype[\epi] := \numtype \arrtype \bool$; then $\typeorder{
\numtype[\epi]} \leq K+1$;
\item $\A_\epi^n := \{$values $F$ such that, (a) for all $v \in
\A_\pi^n$: either $\prog_\epi \vdashcall F\ v \arrr \strue$ or
$\prog_\epi \vdashcall F\ v \arrr \sfalse$ (but not both), and (b)
for all $v,w
\in \A_\pi^n$: if $\numinterpret{v}_\pi^n = \numinterpret{w}_\pi^n$
then $\prog_\epi \vdashcall F\ v \arrr b$ and $\prog_\epi \vdashcall
F\ w \arrr d$ implies $b = d\}$; that is, $\A_\epi^n$ is the set of
functions from $\numtype$ to $\bool$ such that $F\ \numrep{i}$ is
uniquely defined for any representation $\numrep{i}$ of $i \in \{0,
\dots,P(n)-1\}$ in $C_\pi$;
\item $\numinterpret{F}_\epi^n = \sum_{i = 0}^{P(n)-1} \{ 2^{P(n)-1-
i} \mid \exists v \in \A_\pi^n [\numinterpret{v}_\pi^n = i \wedge
\prog_\epi \vdashcall F\ i \arrr \strue] \}$; that is, $F$ is mapped
to the number $i$ with a bitstring $b_0\dots b_{P(n)-1}$, where $b_i
= 1$ if and only if $F\ \numrep{i}$ has value $\strue$;
\item $\Defineds_\epi = \Defineds_\pi \cup \{\symb{not}\} \cup
\{ \identifier{f}_\epi \mid \identifier{f}_\epi$ used in
$\prog_\epi$ below$\}$
\item $\prog_\epi$ consists of the following clauses, followed by the
clauses in $\prog_\pi$:

\smallskip
\texttt{//} $2^{P(n)}-1$ \texttt{corresponds to the bitvector which is
1 at all bits} \\
$\seed[\epi]\ cs = \symb{alwaystrue}_\epi$ \\
$\symb{alwaystrue}_\epi\ x = \strue$

\bigskip
\texttt{// to test whether} $b_0\dots b_{P(n)-1}$ \texttt{is 0, check
each} $b_i = 0$ \\
\texttt{// start in} $b_{P(n)-1}$ \texttt{and count down to test all
bits.} \\
$\zero[\epi]\ cs\ F = \symb{zhelp}_\epi\ cs\ F\ (\seed\ cs)$ \\
$\symb{zhelp}_\epi\ cs\ F\ k = \ifte{\:F\ k\:}{\:\sfalse
\\\phantom{\symb{zhelp}_\epi\ cs\ k\ F =\,}}{
\ifte{\:\zero\ cs\ k\:}{\:\strue
\\\phantom{\symb{zhelp}_\epi\ cs\ k\ F =\,}
}{\:\symb{zhelp}_\epi\ cs\ F\ (\symb{pred}_\pi\ cs\ k)}}$

\bigskip
\texttt{// the predecessor of} $b_0\dots b_i 1 0 \dots 0$
\texttt{is} $b_0 \dots b_i 0 1\dots 1$\texttt{, so go down} \\
\texttt{// through the bits, and flip them until you encounter a 1} \\
$\pred[\epi]\ cs\ F = \symb{phelp}_\epi\ cs\ F\ (\seed\ cs)$ \\
$\symb{phelp}_\epi\ cs\ F\ k = \ifte{\:F\ k\:}{\:\symb{flip}_\epi\ 
cs\ F\ k
\\\phantom{\symb{phelp}_\epi\ cs\ k\ F =\,}
}{\ifte{\:\zero\ cs\ k\:}{\:\seed[\epi]\ cs
\\\phantom{\symb{phelp}_\epi\ cs\ k\ F =\,}
}{\:\symb{phelp}_\epi\ cs\ (\symb{flip}_\epi\ cs\ F\ k)\ 
(\pred\ cs\ k)}}$ \\
$\symb{flip}_\epi\ cs\ F\ k\ i = \ifte{\:\symb{equal}_\pi\ cs\ k\ i
\:}{\:\symb{not}\ (F\ i)\:}{\:F\ i}$ \\
$\symb{not}\ b = \ifte{\:b\:}{\:\sfalse\:}{\:\strue\:}$
\end{itemize}
By standard bitvector arithmetic, the evaluation
requirements are satisfied.
\qed
\end{proof}

\section{Algorithm complexity (Sections~\ref{subsec:detalgorithm}
and~\ref{subsec:elementary:algorithm})}\label{app:complexity}

Now, we turn to proving the complexity of both the core and general
algorithms (Algorithm~\ref{alg:base} and~\ref{alg:general}).  The
key result---which is formulated for Algorithm~\ref{alg:base} but
immediately extends to Algorithm~\ref{alg:general}, is the first
lemma of Section~\ref{subsec:detalgorithm}:

\oldcounter{\complexitylem}
\begin{lemma}
Let $\prog$ be a cons-free program of data order $K$.
Let $\Sigma$ be the set of all types $\atype$ with $\typeorder{\atype}
\leq K$ which occur as part of an argument type, or as an output type
of some $\identifier{f} \in \Defineds$.
Suppose that, given input of total size $n$,
$\pinterpret{\atype}$ has cardinality at most $F(n)$ for all $\atype
\in \Sigma$, and testing whether $e_1 \sqsupseteq e_2$ for $e_1,e_2
\in \interpret{\atype}$ takes at most $F(n)$ steps.
Then Algorithm~\ref{alg:base} runs in $\timecomp{a \cdot F(n)^b}$ for
some $a,b$.
\end{lemma}
\startappendixcounters

\begin{proof}
We first observe that, for any $e \in \interpret{\atype}$ occurring
in the algorithm, $\atype \in \Sigma$.  This is due to the preparation
step where $\prog$ is replaced by $\eprog$.

Write \textsf{a} for the greatest number of arguments any defined
symbol $\identifier{f}$ or variable $x$ occurring in $\eprog$ may
take, and write \textsf{r} for the greatest number of sub-expressions
of any right-hand side in $\eprog$ (which does not depend on the
input!).
We start by observing that $\X$ contains at most $\textsf{a} \cdot
|\Defineds| \cdot F(n)^{\textsf{a}+1}$ statements
$\apps{\identifier{f}}{e_1}{e_n} \leadsto o$, and at most $|\eprog|
\cdot \textsf{r} \cdot F(n)^{\textsf{a}+1}$ statements $t\eta \leadsto
o$.

We observe that step~\ref{alg:prepare:eprog} does not depend on the
input, so takes a constant number of steps.
Step~\ref{alg:prepare:statements} and~\ref{alg:prepare:base} both
take $|\X|$ steps.  The exact time cost of each step depends on
implementation concerns, but is certainly limited by some polynomial
of $F(n)$, by the assumption on $\sqsupseteq$..  Thus, the
preparation step is polynomial in $F(n)$; say its cost is $P_1(F(n))$.

In every step of the iteration, at least one statement is flipped from
unconfirmed to confirmed, or the iteration ends.  Thus, there are at
most $|\X| + 1$ iterations.  In each iteration,
Step~\ref{alg:iterate:value} has a cost limited by $\mathit{Card}(O)
\cdot |\X| \cdot \langle$cost of checking $u' \sqsupseteq u\rangle
\leq F(n)^3 \cdot |\X| \cdot \langle$ some implementation-dependent
constant$\rangle$.
Step~\ref{alg:iterate:lhs} has a cost limited by $|\eprog| \cdot
\langle$cost of matching$\rangle \cdot |\X| \cdot \langle$some
implementation-dependent constant$\rangle$.
Both Steps~\ref{alg:iterate:ifte} and~\ref{alg:iterate:pair}
are limited by $2 \cdot \langle$some constant$\rangle \cdot |\X|$ as
well (the cost for looking up confirmation status of two given
statements), and Step~\ref{alg:iterate:choice} is certainly limited by
$r \cdot \langle$some constant$\rangle \cdot |\X|$.

For each statement $s\eta \leadsto o$ in Steps~\ref{alg:iterate:rhs:var}
and~\ref{alg:iterate:rhs:func}, we must check all suitable tuples
$(e_1,\dots,e_{n'})$---of which there are at most $F(n)^{\textsf{a}}
$---and test confirmation for each $s_i\eta \leadsto e_i$.
In Step~\ref{alg:iterate:rhs:var}, we must additionally do
$\sqsupseteq$ tests for all $o' \in \eta(x)(e_1,\dots,e_{n'})$ for all
tuples; even if we ignore that $\eta(x)$ is a partial function, this
takes at most $F(n)^{\textsf{a}} \cdot F(n)^{\textsf{a}} \cdot F(n)
\cdot \langle$some constant$\rangle$ steps.
In Step~\ref{alg:iterate:rhs:call}, a single lookup over $|\X|$
statements must be done; in Step~\ref{alg:iterate:rhs:extra} this is
combined with a lookup.  Both cases certainly stay below
$F(n)^{2 \cdot \textsf{a}+2} \cdot \langle$some constant$\rangle$
steps.

In total, the cost of iterating is thus limited by $(|\X| + 1)
\cdot |\X| \cdot \langle$some constant$\rangle \cdot \max(
F(n)^3 \cdot |\X|,
|\eprog| \cdot |\X|,
2 \cdot |\X|,
r \cdot |\X|,
F(n)^{2 \cdot \textsf{a}+2})$.  Since $|\X|$ is a polynomial in
$F(n)$, this is certainly bounded by $P_2(F(n))$ for some polynomial
$P_2$.

Finally, completion requires at most $|\X|$ tests.  Overall, all
steps together gives a polynomial time complexity in $F(n)$.
\qed
\end{proof}

Thus, complexity of Algorithm~\ref{alg:base} relies on the size of
each $\pinterpret{\atype}$, and complexity of
Algorithm~\ref{alg:general} on the sizes of $\interpret{\atype}$.

To determine these sizes as well as the complexity of testing
$\sqsupseteq$ on two given extensional values, we first obtain a
simple helper lemma for calculation:

\begin{lemma}\label{lem:expmultiply}
If $X,Y \geq 2$, then $\exp_2^K(X) \cdot \exp_2^K(Y) \geq
\exp_2^K(X \cdot Y)$ for $K \in \N$.
\end{lemma}

\begin{proof}
We start by observing that for $X,Y \geq 2$
always (**) $X \cdot Y \geq X + Y$:
\begin{itemize}
\item $2 \cdot 2 = 4 = 2 + 2$;
\item if $X \cdot Y \geq X + Y$, then $X \cdot (Y+1) = X \cdot Y + X
  \geq (X + Y) + X \geq X + (Y + 1)$;
\item if $X \cdot Y \geq X + Y$, then $(X + 1) \cdot Y \geq Y + X + 1$
  in the same way.
\end{itemize}

By induction on $K$ we also see:
(***) if $X \geq 2$ then $\exp_2^K(X) \geq 2$ for all $K$.

Now the lemma follows by another induction on $K$:
\begin{itemize}
\item for $K = 0$: $\exp_2^K(X) \cdot \exp_2^K(Y) = X \cdot Y =
  \exp_2^K(X \cdot Y)$;
\item for $K \geq 0$:
  $\exp_2^{K+1}(X) \cdot \exp_2^{K+1}(Y) = 2^{\exp_2^K(X)} \cdot
  2^{\exp_2^K(Y)} = 2^{\exp_2^K(X) + \exp_2^K(Y)} \leq
  2^{\exp_2^K(X) \cdot \exp_2^K(Y)}$ by (**) and (***),
  $\leq 2^{\exp_2^K(X) \cdot \exp_2^K(Y)} = \exp_2^{K+1}(X \cdot Y)$
  by the induction hypothesis.
  \qed
\end{itemize}
\end{proof}

For the first part of Lemma~\ref{lem:pinterpretcard}, we consider the
cardinality of each $\pinterpret{\atype}$.

\begin{lemma}\label{lem:pinterpretsize}
If $1 \leq \Card(\B) < N$, then for each $\atype$ with
$\typeorder{\atype} \leq K$ such that $L$ sorts occur in $\atype$
(including repetitions) we have:
$\Card(\pinterpret{\atype}) < \exp_2^K(N^L)$.
\end{lemma}

\begin{proof}
By induction on the form of $\atype$.

For $\atype \in \Sorts$, $\pinterpret{\atype} \subseteq \B$ so
$\Card(\pinterpret{\atype}) \leq \Card(\B) < N$.

For $\atype = \atype_1 \times \atype_2$ with $\atype_1$ having $L_1$
sorts and $\atype_2$ having $L_2$, we have
\[
\begin{array}{rcl}
\Card(\pinterpret{\atype_1 \times \atype_2})
& = & \Card(\pinterpret{\atype_1}) \cdot \Card(\pinterpret{\atype_2}) \\
& < & \exp_2^K(N^{L_1}) \cdot \exp_2^K(N^{L_2}) \\
& \leq & \exp_2^K(N^{L_1} \cdot N^{L_2})\ 
         \text{by Lemma~\ref{lem:expmultiply}} \\
& = & \exp_2^K(N^{L_1 + L_2}) = \exp_2^K(L) \\
\end{array}
\]
For $\atype = \atype_1 \arrtype \atype_2$ with $\atype_1$ having $L_1$
sorts and $\atype_2$ having $L_2$, each element of
$\pinterpret{\atype}$ can be seen as a \emph{total} function from
$\pinterpret{\atype_1}$ to $\pinterpret{\atype_2} \cup \{\bot\}$.
Therefore,
\[
\begin{array}{rcl}
\Card(\pinterpret{\atype_1 \arrtype \atype_2})
& = & (\Card(\pinterpret{\atype_2})+1)^{\Card(\pinterpret{\atype_1})} \\
& \leq & \exp_2^K(N^{L_2})^{\Card(\pinterpret{\atype_1})} \\
& < & \exp_2^K(N^{L_2})\text{\textasciicircum}(\exp_2^{K-1}(N^{L_1})) \\
& = & 2\text{\textasciicircum}(\ \exp_2^{K-1}(N^{L_2}) \cdot
  \exp_2^{K-1}(N^{L_1})\ ) \\
& \leq & 2\text{\textasciicircum}(\ \exp_2^{K-1}(N^L)\ )\ 
         \text{by Lemma~\ref{lem:expmultiply}} \\
& = & \exp_2^K(N^L) \\
\vspace{-28pt}
\end{array}
\]
\qed
\end{proof}

The cardinality of each $\interpret{\atype}$ (as used in
Lemma~\ref{lem:interpretcard}) is obtained similarly.

\begin{lemma}\label{lem:interpretsize}
If $1 \leq \Card(\B) < N$, then for each $\atype$ with
$\mathit{depth}(\atype) \leq K$ such that $L$ sorts occur in $\atype$
(including repetitions) we have:
$\Card(\interpret{\atype}) < \exp_2^K(N^L)$.
\end{lemma}

\begin{proof}
By induction on the form of $\atype$.

For $\atype \in \Sorts$, $\interpret{\atype} \subseteq \B$ so
$\Card(\interpret{\atype}) \leq \Card(\B) < N$.

For $\atype = \atype_1 \times \atype_2$, we obtain $\Card(
\interpret{\atype}) \leq \exp_2^K(N^L)$ in exactly the same way as
in Lemma~\ref{lem:pinterpretsize}.

For $\atype = \atype_1 \arrtype \atype_2$ with $\atype_1$ having
$L_1$ sorts and $\atype_2$ having $L_2$, each
element of $\interpret{\atype}$ is a subset of $\interpret{\atype_1}
\times \interpret{\atype_2}$; therefore,
\[
\begin{array}{rcl}
\Card(\pinterpret{\atype_1 \arrtype \atype_2})
& = & 2\text{\textasciicircum}(\ \Card(\pinterpret{\atype_1} \times
  \pinterpret{\atype_2})\ ) \\
& \leq & 2\text{\textasciicircum}(\ \exp_2^{K-1}(N^{L_1}) \cdot
  \exp_2^{K-1}(N^{L_2})\ ) \\
& \leq & 2\text{\textasciicircum}(\ \exp_2^{K-1}(N^L)\ )\ 
         \text{by Lemma~\ref{lem:expmultiply}} \\
& = & \exp_2^K(N^L) \\
\vspace{-28pt}
\end{array}
\]
\qed
\end{proof}

Aside from the cardinalities of $\pinterpret{\atype}$ and
$\interpret{\atype}$, Lemmas~\ref{lem:pinterpretcard}
and~\ref{lem:interpretcard} also consider the complexity of deciding
$e \sqsupseteq u$ for two (deterministic or non-deterministic)
extensional values.  This complexity we consider for both lemmas
together:

\begin{lemma}\label{lem:sqsupcomplexity}
Let $[\atype]$ be one of $\pinterpret{\atype}$ or
$\interpret{\atype}$, and suppose that we know that for all subtypes
of $\atype$ containing $L$ sorts: $\Card([\atype]) < \exp_2^K(N^L)$
for some fixed $K$, and $N \geq 2$.  Then for any $e,u \in [\atype]$:
testing $e \sqsupseteq u$ requires $< \exp_2^K(N^{(L+1)^3})$
comparisons between elements of $\B$.
\end{lemma}

\begin{proof}
We let $C_\atype$ be the maximum cost of either $\sqsupseteq$ tests or
equality tests for elements of $[\atype]$.  We first observe:
\begin{enumerate}
\item $(X+Y+1)^3 = X^3 + Y^3 + 3X^2Y + 3XY^2 + 3X^2 + 3Y^2 + 6XY +
  3X + 3Y + 1$;
\item $(X+1)^3 = X^3 + 3X^2 + 3X + 1$;
\item\label{third:compare}
  $(X+Y+1)^3 - (X+1)^3 - (Y+1)^3 = 3X^2Y + 3XY^2 + 6XY - 1$.
\end{enumerate}

Now, $C_\asort = 1 < N^8 = \exp_2^K(N^{2^3})$ for $\asort \in \Sorts$.
Writing $L_1$ for the number of sorts in $\atype_1$ and $L_2$ for the
number of sorts in $\atype_2$, we have:
\[
\begin{array}{rcl}
C_{\atype_1 \times \atype_2}
& = & C_{\atype_1} + C_{\atype_2} \\
& < & \exp_2^K(N^{(L_1+1)^3}) + \exp_2^K(N^{(L_2+1)^3})\ \text{by
  the induction hypothesis} \\
& \leq & \exp_2^K(N^{(L_1+1)^3} \cdot N^{(L_2+1)^3})\ \text{because
  both sides are at least 2} \\
& \leq & \exp_2^K(N^{(L_1+1)^3 + (L_2+1)^3})\ \text{by
  Lemma~\ref{lem:expmultiply}} \\
& \leq & \exp_2^K(N^{(L_1+L_2+1)^3})\ \text{by observation
  \ref{third:compare} above} \\
& = & \exp_2^K(N^{(L+1)^3}) \\
\end{array}
\]
To compare $A_{\atype_1 \arrtype \atype_2}$ and
$B_{\atype_1 \arrtype \btype_1}$, we may for instance do the
following:
\begin{itemize}
\item for all $(u_1,u_2) \in B$:
  \begin{itemize}
  \item for all $(e_1,e_2) \in A$,
    test $e_1 = u_1$ and either $e_2 = e_2$ or $e_2 \sqsupseteq u_2$;
  \item conclude failure if we didn't find a match
  \end{itemize}
  \item in the case of $\sqsupseteq$, conclude success if we haven't
    concluded failure yet; in the case of $=$, also do the test in the
    other direction
\end{itemize}
This gives, roughly:
\[
\begin{array}{rcl}
C_{\atype \arrtype \btype}
& \leq & 2 \cdot \Card([\atype_1 \times \atype_2]) \cdot
  \Card([\atype_1 \times \atype_2]) \cdot (C_{\atype_1} +
  C_{\atype_2}) \\
& \leq & 2 \cdot \exp_2^K(N^L) \cdot \exp_2^K(N^L) \cdot
  (C_{\atype_1} + C_{\atype_2}) \\
& < & 2 \cdot \exp_2^K(N^L) \cdot \exp_2^K(N^L) \cdot
  \exp_2^K(N^{(L_1+1)^3 + (L_2+1)^3})\ \text{as above} \\
& \leq & 2 \cdot \exp_2^K(N^{2 \cdot L + (L_1+1)^3 + (L_2+1)^3})\ 
  \text{by Lemma~\ref{lem:expmultiply}} \\
& \leq & \exp_2^K(N^{2 \cdot L + (L_1+1)^3 + (L_2+1)^3 + 1})\ 
  \text{because $N \geq 2$} \\
& \leq & \exp_2^K(N^{(L_1+L_2+1)^3})\ \text{by observation
  \ref{third:compare} above} \\
& & \text{because}\ (X + 6L_1L_2 - 1) - (2 L_1 + 2L_2 + 1) \geq 0\ 
  \text{when}\ L_1,L_2 \geq 1 \\
\vspace{-24pt}
\end{array}
\]
\qed
\end{proof}

All parts now proven, Lemmas~\ref{lem:pinterpretcard}
and~\ref{lem:interpretcard} follow immediately.

\oldcounter{\pcomplexitylem}
\begin{lemma}
If $1 \leq \Card(\B) < N$, then for each $\atype$ of length $L$
(where the length of a type is the number of sorts occurring in it,
including repetitions), with
$\typeorder{\atype} \leq K$: $\Card(\pinterpret{\atype}) <
\exp_2^K(N^L)$. 
Testing $e \sqsupseteq u$ for $e,u \in \pinterpret{\atype}$ takes
at most $\exp_2^K(N^{(L+1)^3})$ comparisons between elements of $\B$.
\end{lemma}
\startappendixcounters

\begin{proof}
The first part is Lemma~\ref{lem:pinterpretsize}; using this, the
second part follows by Lemma~\ref{lem:sqsupcomplexity}.
\qed
\end{proof}

\oldcounter{\ncomplexitylem}
\begin{lemma}
If $1 \leq \Card(\B) < N$, then for each $\atype$ of length $L$, with 
$\mathit{depth}(\atype) \leq K$:
$\Card(\interpret{\atype}) < \exp_2^K(N^L)$.
Testing $e \sqsupseteq u$ for $e,u \in \interpret{\atype}$ takes at
most $\exp_2^K(N^{(L+1)^3})$ comparisons.
\end{lemma}
\startappendixcounters

\begin{proof}
The first part is Lemma~\ref{lem:interpretsize}; using this, the
second part follows by Lemma~\ref{lem:sqsupcomplexity}.
\qed
\end{proof}

\section{Algorithm correctness (Section~\ref{subsec:proof})}
\label{app:correctness}

We prove that for both Algorithm~\ref{alg:base} and
Algorithm~\ref{alg:general}: $\progresult$ if and only if $b$ is in
the set returned by the algorithm.  We do this in four steps:

\begin{description}
\item[Section~\ref{subsec:algprop}] we obtain some properties on
  (deterministic or non-deterministic) extensional values and $\eprog$;
\item[Section~\ref{subsec:soundness}] we prove that for both algorithms:
  if $b$ is returned by the algorithm, then $\progresult$;
\item[Section~\ref{subsec:completeness}] we prove that for
  Algorithm~\ref{alg:general}: if $\progresult$, then $b$ is returned by
  the algorithm;
\item[Section~\ref{subsec:detcomplete}] we adapt this proof to the
  deterministic setting.
\end{description}

In this, we break from the order in the main text: where the text
considers the deterministic case first (Algorithm~\ref{alg:base}), we
will show completeness first for the non-deterministic case
(Algorithm~\ref{alg:general}).  The reason for this choice is that our
algorithm has been  designed particularly for the non-deterministic
cases (both the general non-deterministic setting which results in a
classification of $\elementary$, and the result for arrow depth in
Section~\ref{sec:nopartialvar}) for which no algorithm yet existed in
the literature.  This results in a significantly simpler proof.

We do also handle the deterministic case, but this requires an extra
proof step to replace the sets $\interpret{\atype}$ of
non-deterministic extensional values by the sets $\pinterpret{\atype}$
of deterministic extensional values.

\medskip
Note that all deterministic extensional values are also
non-deterministic extensional values.  In this appendix,
\emph{extensional values} may refer to either deterministic or
non-deterministic extensional values.

\subsection{Properties of extensional values and $\eprog$}\label{subsec:algprop}

We begin by deriving some properties relevant to both the soundness and
completeness proofs.  First, the following lemma will be invaluable
when matching extensional values against the left-hand sides of clauses.

\begin{lemma}\label{lem:downhelp}
Fix a set $\B$ of data expressions, closed under taking
sub-expressions.
Let $\down{}{}$ be a relation, relating values $v$ of type $\atype$
to extensional values $e \in \interpret{\atype}$, notation $\down{v}{e}$,
such that:
\begin{itemize}
\item $\down{v}{e}$ for $v,e$ data if and only if $v = e$, and
\item $\down{(v,w)}{(e,u)}$ if and only if both $\down{v}{e}$ and
$\down{w}{u}$.
\end{itemize}
Let $v_1:\atype_1,\dots,v_k:\atype_k$ and $e_1 \in \interpret{\atype_1
},\dots,e_k \in \interpret{\atype_k}$ be such that $\down{s_i}{e_i}$
for each $i$, and let $\rho\colon\apps{\identifier{f}}{\ell_1}{\ell_k}
= s$ be a clause.
Then there is an environment $\gamma$ such that each $v_i = \ell_i
\gamma$ if and only if there is an ext-environment $\eta$ such that
each $e_i = \ell_i\eta$, and if both are satisfied then
$\down{\gamma(x)}{\eta(x)}$ for all $x \in
\Var(\apps{f}{\ell_1}{\ell_k})$.
\end{lemma}

Essentially, this lemma says that no matter how we associate values of
a higher type to extensional values, if data and pairing are handled
as expected, then matching is done in the natural way.

\begin{proof}
For $\ell$ a pattern of type $\atype$, $v : \atype$ a value and $e
\in \interpret{\atype}$ such that $\down{v}{e}$, the lemma follows
easily once we prove the following by induction on $\ell$:
\begin{itemize}
\item If $v = \ell\gamma$ for some $\gamma$, then there exists $\eta$
on domain $\Var(\ell)$ such that $e = \ell\eta$ and
$\down{\gamma(x)}{\eta(x)}$ for all $x$ in the domain:
\begin{itemize}
\item If $\ell$ is a variable, then $\gamma(\ell) = v$, so choose
  $\eta := [\ell:=e]$.
\item If $\ell$ is a pair $(\ell_1,\ell_2)$, then $v = (v_1,v_2)$
  and therefore $e = (e_1,e_2)$ with both $\down{v_1}{
  e_1}$ and $\down{v_2}{e_2}$; by the induction hypothesis, we
  find $\eta_1$ and $\eta_2$ on domains $\Var(\ell_1)$ and
  $\Var(\ell_2)$ respectively; we are done with $\eta := \eta_1
  \cup \eta_2$.
\item If $\ell = \apps{\identifier{c}}{\ell_1}{\ell_m}$ with
  $\identifier{c} \in \Constructors$, then $v$ and $e$ are both
  data expressions, so $v = e$; since the argument types of
  constructors have order $0$, all $x \in \Var(\ell)$
  have type order $0$, so we can choose $\eta(x) := \gamma(x
  )$ for such $x$.
\end{itemize}
\item If $e = \ell\eta$ for some $\eta$, then there exists $\gamma$
on domain $\Var(\ell)$ such that $s = \ell\gamma$ and
$\down{\gamma(x)}{\eta(x)}$ for $x$ in $\Var(\ell)$;
\pagebreak
this reasoning is parallel to the case above.
\qed
\end{itemize}
\end{proof}

Next we move to transivity of $\sqsupseteq$.  Note that $\sqsupseteq$
for two extensional values $A_\atype$ and $B_\atype$ is \emph{not} set
containment $A \supseteq B$, but slightly different.

\begin{lemma}\label{lem:sqsuptrans}
$\sqsupseteq$ is transitive.
\end{lemma}

\begin{proof}
Let $e \sqsupseteq u \sqsupseteq o$ with $e,u,o \in
\interpret{\atype}$; we prove that $e \sqsupseteq o$ by induction
on the form of $\atype$.  The induction is entirely straightforward:
\begin{itemize}
\item if $\atype \in \Sorts$, then $e = u = o$;
\item if $\atype = \atype_1 \times \atype_2$, then $e = (e_1,e_2),\ 
  v = (u_1,u_2)$ and $o = (o_1,o_2)$ with both $e_1 \sqsupseteq u_1
  \sqsupseteq o_1$ and $e_2 \sqsupseteq u_2 \sqsupseteq o_2$; by the
  induction hypothesis indeed $e_1 \sqsupseteq o_1$ and $e_2
  \sqsupseteq o_2$;
\item if $\atype = \atype_1 \arrtype \atype_2$, then we can write
  $e = A_\atype$, $u = B_\atype$ and $o = C_\atype$ and:
  \begin{itemize}
  \item for all $(o_1,o_2) \in C$ there exists $u_2 \sqsupseteq o_2$
    such that $(o_1,u_2) \in B$;
  \item for all $(o_1,u_2) \in B$ there exists $e_2 \sqsupseteq u_2$
    such that $(o_1,e_2) \in A$.
  \end{itemize}
  As the induction hypothesis gives $e_2 \sqsupseteq o_2$, also
  $e \sqsupseteq o$.
  \qed
\end{itemize}
\end{proof}

Finally, we show how $\prog$ and $\eprog$ in Algorithm~\ref{alg:base}
relate:

\begin{lemma}\label{lem:progeprog}
$\progresult$ if and only if $\eprog \vdashcall
\apps{\symb{start}}{d_1}{d_M}
\arrr b$.
\end{lemma}

\begin{proof}
By Lemma~\ref{lem:proper} and the observation that the fresh symbol
$\symb{start}$ does not occur in any other clauses, $\progresult$ if
and only if $\progeval{\eprog}{d_1,\dots,d_M} \mapsto b$, which by
definition is the case if and only if $\eprog,[x_1:=d_1,\dots,x_M:=
d_M] \vdash \apps{\identifier{f}_1}{x_1}{x_M} \arrr b$.  As there is
only one clause for $\symb{start}$ in $\eprog$, this is the case
if and only if $\eprog \vdashcall \apps{\identifier{start}}{d_1}{d_M}
\arrr b$.
\qed
\end{proof}

\subsection{Soundness of Algorithms~\ref{alg:base} and~\ref{alg:general}}
\label{subsec:soundness}

We turn to soundness.  We will see that for every $b$ in the
output set of Algorithms~\ref{alg:base} and~\ref{alg:general} indeed
$\progresult$.  Since each $\pinterpret{\atype} \subseteq
\interpret{\atype}$---and therefore the statements considered in
Algorithm~\ref{alg:base} are a subset of those considered in
Algorithm~\ref{alg:general}---it suffices to prove this for the
non-deterministic algorithm, as the deterministic case follows
directly.

To achieve this end, we first give a definition to relate values and
extensional values in line with Lemma~\ref{lem:downhelp}, and obtain
two further helper results:

\begin{definition}
For a value $v : \atype$ and an extensional value $e \in
\interpret{\atype}$, we recursively define $\down{v}{e}$ if one of the
following holds:
\begin{itemize}
\item $\atype \in \Sorts$ and $v = e$;
\item $\atype = \atype_1 \times \atype_2$ and $v = (v_1,v_2)$ and
  $e = (e_1,e_2)$ with $\down{v_1}{e_1}$ and $\down{v_2}{e_2}$;
\item $\atype = \atype_1 \arrtype \atype_2$ and $e = A_\atype$ with
  $A \subseteq \varphi(v) := \{ (u_1,u_2) \mid u_1 \in
  \interpret{\atype_1} \wedge u_2 \in \interpret{\atype_2}
  \wedge$ for all values $w_1 : \atype_1$ with $\down{w}{u_1}$ there
  is some value $w_2 : \atype_2$ with $\down{w_2}{u_2}$ such that
  $\eprog \vdashcall \app{v}{w_1} \arrr w_2 \}$.
\end{itemize}
\end{definition}

It is easy to see that $\Downarrow$ satisfies the requirements
of Lemma~\ref{lem:downhelp}.

The first helper lemma essentially states the following: if a value $v$
is associated to an extensional value $e$ (in the sense that
$\down{v}{e}$), and $v_1,\dots,v_n$ are associated to extensional
values $u_1,\dots,u_n$, then the set $e(u_1,\dots,u_n)$ contains only
(extensional values associated to) the possible results of evaluating
$\apps{v}{v_1}{v_n}$.  Thus, if $\down{v}{e}$ then $e$ represents $v$
in the expected sense: by defining the same ``function''.

(To make it easier to use this lemma in the proof of
Lemma~\ref{lem:generalsound}, however, it is formulated in a slightly
more general way than this sketch: the lemma considers an expression
$s$ which evaluates to $v$, and similarly expressions $t_1,\dots,t_n$
which evaluate to each $v_1,\dots,v_n$.  We show that $\apps{s}{t_1}{
t_n}$ evaluates to the elements of $e(u_1,\dots,u_n)$.)

\begin{lemma}\label{lem:slowcall}
Assume given an environment $\gamma$.
Let $s : \atype_1 \arrtype \dots \arrtype \atype_n \arrtype \btype$,
and $e \in \interpret{\atype_1 \arrtype \dots \arrtype \atype_n
\arrtype \btype}$ be such that $\down{v}{e}$ for some value $v$ with
$\eprog,\gamma \vdash s \arrr v$.
For $1 \leq i \leq n$, let $t_i,v_i : \atype_i$ and $u_i \in
\interpret{\atype_i}$ be such that $\eprog,\gamma \vdash t_i \arrr
\down{v_i}{u_i}$.
Then for any $o \in e(u_1,\dots,u_n)$ there exists $w : \btype$ such
that $\down{w}{o}$ and $\eprog,\gamma \vdash \apps{s}{t_1}{t_n} \arrr
w$.
\end{lemma}

\begin{proof}
By induction on $n \geq 0$.

If $n = 0$, then $o = e$ and $\eprog,\gamma \vdashcall s \arrr v$ is
given; we choose $w := v$.

If $n \geq 1$, then there is some $o' := A_{\atype_n \arrtype \btype}
\in e(u_1,\dots,u_{n-1})$ such that $(u_n,o) \in A$.  By the
induction hypothesis, there exists a value $w'$ such that $\eprog,
\gamma \vdash \apps{s}{t_1}{t_{n-1}} \arrr \down{w'}{o'}$.  Since also
$\down{v_n}{u_n}$, the definition of $\Downarrow$ provides a value
$w$ such that $\eprog \vdashcall \app{w'}{v_n} \arrr \down{w}{o}$.
As $w'$ is a value of higher type, it must have a form
$\apps{\identifier{f}}{w_1}{w_i}$, so we can apply [Appl] to obtain
$\eprog,\gamma \vdash \app{(\apps{s}{t_1}{t_{n-1}})}{t_n} \arrr w$.
\qed
\end{proof}

The second helper lemma states the following: if a value $v$ is
associated to an extensional value $e$, then it is also associated to
all ``smaller'' extensional values: if $e \sqsupseteq u$, then $u$
simply has less information about the value described.  The property
is closely related to transitivity of $\sqsupseteq$:

\begin{lemma}\label{lem:downsqsup}
For any value $v : \atype$ and extensional values $e,u \in
\interpret{\atype}$: if $\down{v}{e} \sqsupseteq u$
then $\down{v}{u}$.
\end{lemma}

\begin{proof}
By induction on the form of $\atype$:
\begin{itemize}
\item if $v$ is data, then $v = e = u$;
\item if $v = (v_1,v_2)$, then $\down{v}{e} \sqsupseteq u$
  implies $e = (e_1,e_2)$ and $u = (u_1,u_2)$ with $\down{v_i}{e_i}
  \sqsupseteq u_i$ for $i \in \{1,2\}$, so $\down{v_i}{u_i}$ by the
  induction hypothesis;
\item if $v$ is a functional value, then $e = A_\atype$ and $u =
  B_\atype$, and for all $(o_1,o_2) \in B$ there exists $o_2'
  \sqsupseteq o_2$ such that $(o_1,o_2') \in A$; thus, for all
  values $\down{w_1}{o_1}$, the property that $\down{v}{e}$ gives
  some $w_2$ such that $\eprog \vdashcall \app{v}{w_1} \arrr
  \down{w_2}{o_2'} \sqsupseteq o_2$, which by the induction hypothesis
  implies $\down{w_2}{o_2}$ as well.  Thus, indeed $\down{v}{u}$.
  \qed
\end{itemize}
\end{proof}

With these preparations, we are ready to tackle the soundness proof:

\oldcounter{\soundnesslem}
\begin{lemma}
If Algorithm~\ref{alg:base} or~\ref{alg:general} returns a set $A \cup
\{b\}$, then $\progresult$.
\end{lemma}
\startappendixcounters

\begin{proof}
We prove the lemma by obtaining the following results:
\begin{enumerate}
\item\label{lem:soundness:lhs}
Let:
\begin{itemize}
\item $\identifier{f} : \atype_1 \arrtype \dots \arrtype \atype_m
  \arrtype \asortorpair \in \F$ be a defined symbol;
\item $v_1 : \atype_1,\dots,v_n : \atype_n$ be values, for
  $1 \leq n \leq \arity(\identifier{f})$;
\item $e_1 \in \interpret{\atype_1},\dots,e_n \in
  \interpret{\atype_n}$ be such that each $\down{v_i}{e_i}$;
\item $o \in \interpret{\atype_{n+1} \arrtype \dots \arrtype
  \atype_m \arrtype \asortorpair}$.
\end{itemize}
If the statement $\vdash \apps{\identifier{f}}{e_1}{e_n} \leadsto o$
is eventually confirmed, then we can derive $\eprog \vdashcall
\apps{\identifier{f}}{v_1}{v_n} \arrr w$ for some $w$ with
$\down{w}{o}$.
\item\label{lem:soundness:rhs}
Let:
\begin{itemize}
\item $\rho\colon\apps{\identifier{f}}{\ell_1}{\ell_k} = s$ be a clause
  in $\eprog$;
\item $t : \btype$ be a sub-expression of $s$;
\item $\eta$ be an ext-environment for $\rho$;
\item $\gamma$ be an environment such that $\down{\gamma(x)}{
  \eta(x)}$ for all $x \in \Var(\apps{\identifier{f}}{\ell_1}{\ell_k})$;
\item $o \in \interpret{\btype}$.
\end{itemize}
If the statement $\eta \vdash t \leadsto o$ is eventually confirmed,
then we can derive $\eprog,\gamma \vdash t \arrr w$ for some $w$ with
$\down{w}{o}$.
\end{enumerate}
This proves the lemma: if the algorithm returns $b$, then
$\apps{\identifier{start}}{d_1}{d_M} \leadsto b$ is confirmed, so
$\eprog \vdashcall \apps{\symb{start}}{d_1}{d_M} \mapsto
b$.  By Lemma~\ref{lem:progeprog}, $\progresult$.

We prove both statements together by induction on the algorithm.
\begin{enumerate}
\item $\apps{\identifier{f}}{e_1}{e_n} \leadsto o$ can only be
confirmed in two ways:
\begin{description}
\item[(\ref{alg:iterate:value})] $n < \arity(\identifier{f})$,
  $o = O_{\atype_{n+1} \arrtype \dots \arrtype \atype_m \arrtype
  \asortorpair}$ and for all $(e_{n+1},u) \in O$ there is some
  $u' \sqsupseteq u$ such that also
  $\apps{\identifier{f}}{e_1}{e_{n+1}} \leadsto u'$ is confirmed.
  By the induction hypothesis, this implies that for all such
  $e_{n+1}$ and $u'$, and for all $v_{n+1} : \atype_{n+1}$ with
  $\down{v_{n+1}}{e_{n+1}}$, there exists $w'$ with $\down{w'}{u'}$
  such that $\eprog \vdashcall \apps{\identifier{f}}{v_1}{v_{n+1}}
  \arrr w'$.  By Lemma~\ref{lem:downsqsup}, also $\down{w'}{u}$.
  Thus, $O \subseteq \varphi(\apps{\identifier{f}}{v_1}{v_n})$,
  and $\down{(\apps{\identifier{f}}{v_1}{v_n})}{o}$.
  We are done choosing $w := \apps{\identifier{f}}{v_1}{v_n}$,
  since $\eprog \vdashcall \apps{\identifier{f}}{v_1}{v_n}
  \arrr \apps{\identifier{f}}{v_1}{v_n}$ by [Closure].
\item[(\ref{alg:iterate:lhs})] $n = \arity(\identifier{f})$ and,
  for $\rho\colon \apps{\identifier{f}}{\ell_1}{\ell_k} = s$ the
  first matching clause in $\eprog$ and $\eta$ the matching
  ext-environment,
  $\eta \vdash s \leadsto 
  o$ is confirmed.
  Following Lemma~\ref{lem:downhelp}, there exists an environment
  $\gamma$ on domain $\Var(\apps{\identifier{f}}{\ell_1}{\ell_k})$
  with each $\ell_j\gamma = v_j$ and $\down{\gamma(x)}{\eta(x)}$
  for each $x$ in the mutual domain.
  By the induction hypothesis, we can derive $\eprog,\gamma \vdash
  s \arrr w$ for some $w$ with $\down{w}{o}$; by [Call]
  therefore $\eprog \vdashcall \apps{\identifier{f}}{v_1}{v_n} \arrr
  w$ (necessarily $n = k$).
\end{description}
\item $\eta \vdash t \leadsto o$ can be confirmed in eight ways:
\begin{description}
\item[(\ref{alg:prepare:base:var})]
  $t \in \V$ and $\eta(t) \sqsupseteq o$; choosing $w = \gamma(t)$,
  we have $\eprog,\gamma \vdash t \arrr w$ by [Instance], and
  $\down{w}{o}$ by Lemma~\ref{lem:downsqsup}.
\item[(\ref{alg:prepare:base:constructor})]
  $t = \apps{\identifier{c}}{t_1}{t_m}$ with $\identifier{c} \in
  \Constructors$ and $t\eta = o$; choosing $w = t\gamma = o$, we
  clearly have $\down{w}{o}$ and $\eprog,\gamma \vdash t \arrr w$
  by [Constructor].
\item[(\ref{alg:iterate:ifte})]
  $t = \ifte{t_1}{t_2}{t_3}$ and either
  \begin{description}
  \item[(\ref{alg:iterate:ifte:true})] $\eta \vdash t_1 \leadsto
    \strue$ and $\eta \vdash t_2 \leadsto o$ are both confirmed; by
    the induction hypothesis, $\eprog,\gamma\vdash t_1 \arrr \strue$
    and $\eprog,\gamma \vdash t_2 \arrr w$ for some $w$ with
    $\down{w}{o}$;
  \item[(\ref{alg:iterate:ifte:false})] $\eta \vdash t_1 \leadsto
    \sfalse$ and $\eta \vdash t_3 \leadsto o$ are both confirmed; by
    the induction hypothesis, $\eprog,\gamma\vdash t_1 \arrr \strue$
    and $\eprog,\gamma \vdash t_3 \arrr w$ for some $w$ with
    $\down{w}{o}$.
  \end{description}
  In either case we complete with [Conditional], using
  [Cond-True] in the former and [Cond-False] in the latter case.
\item[(\ref{alg:iterate:choice})]
  $t = \apps{\choice}{t_1}{t_n}$ and $\eta \vdash t_i \arrr o$ is
  confirmed for some $i$; by the induction hypothesis, $\eprog,
  \gamma \vdash t_i \arrr w$ for a suitable $w$, so $\eprog,\gamma
  \vdash t \arrr w$ by [Choice].
\item[(\ref{alg:iterate:pair})]
  $t = (t_1,t_2)$ and $o = (o_1,o_2)$ and $\eta \vdash t_i \arrr o_i$
  is confirmed for $i \in \{1,2\}$; by the induction hypothesis,
  $\eprog,\gamma \vdash t_i \arrr \down{w_i}{o_i}$ for both $i$, so
  $\eprog,\gamma \vdash t \arrr \down{(w_1,w_n)}{o}$ by [Pair].
\item[(\ref{alg:iterate:rhs:var})] $t = \apps{x}{t_1}{t_n}$ with
  $x \in \V$ and $n > 0$, and there are $e_1,\dots,e_n$ such that
  $\eta \vdash t_i \leadsto e_i$ is confirmed for all $i$, and
  $\eta(x)(e_1,\dots,e_n) \ni o' \sqsupseteq o$ for some $o'$;
  by the induction hypothesis, there are $v_1,\dots,v_n$ such that
  $\eprog,\gamma \vdash t_i \arrr v_i$ for all $i$.
  Since also $\eprog,\gamma \vdash x \arrr \down{\gamma(x)}{\eta(
  x)}$ by [Instance], Lemma~\ref{lem:slowcall} provides $w$ such
  that $\eprog,\gamma \vdash \apps{x}{t_1}{t_n} \arrr \down{w}{o'}$;
  by Lemma~\ref{lem:downsqsup}, also $\down{w}{o}$.
\item[(\ref{alg:iterate:rhs:call})] $t = \apps{\identifier{f}}{t_1
  }{t_n}$ with $\identifier{f} \in \Defineds$ and $0 \leq n \leq
  \arity(\identifier{f})$, and there are $e_1,\dots,e_n$ such that
  $\eta \vdash t_i \leadsto e_i$ is confirmed for all $i$, and
  $\vdash \apps{\identifier{f}}{e_1}{e_n} \leadsto o$ is marked
  confirmed.
  By the second induction hypothesis, there are $v_1,\dots,v_n$
  such that $\eprog,\gamma \vdash t_i \arrr \down{v_i}{e_i}$ for
  all $i$, and
  therefore by the first induction hypothesis, there is $w$
  such that $\eprog \vdashcall \apps{\identifier{f}}{v_1}{v_n}
  \arrr \down{w}{o}$.  Combining this with [Function] and $n$
  [Appl]s, we have $\eprog,\gamma \vdash \apps{\identifier{f
  }}{t_1}{t_n} \arrr w$ as well.
\item[(\ref{alg:iterate:rhs:extra})] $t = \apps{\identifier{f}}{
  t_1}{t_n}$ with $\identifier{f} \in \Defineds$ and $n > k :=
  \arity(\identifier{f})$, and there are $e_1,\dots,e_n$ such that,
  just as in the previous two cases, $\eprog,\gamma \vdash t_i \arrr
  \down{v_i}{e_i}$ for each $i$.  Moreover, $u(e_{k+1},\dots,e_n) \ni
  o' \sqsupseteq o$ for some $u$ with $\apps{\identifier{f}}{e_1}{e_k}
  \leadsto u$ confirmed.  As in the previous case, there exists $v$
  such that $\eprog,\gamma \vdash \apps{\identifier{f}}{t_1}{t_k}
  \arrr \down{v}{u}$.  Lemma~\ref{lem:slowcall} provides $w$ with
  $\eprog,\gamma \vdash \apps{\identifier{f}}{t_1}{t_n} \arrr
  \down{w}{o'}$; since $o' \sqsupseteq o$ also $\down{w}{o}$ by
  Lemma~\ref{lem:downsqsup}.
\qed
\end{description}
\end{enumerate}
\end{proof}

\subsection{Completeness of Algorithm~\ref{alg:general}}
\label{subsec:completeness}

We turn to completeness; in particular, the property that if
$\progresult$ then Algorithm~\ref{alg:general} returns a set
containing $b$ (in Section~\ref{subsec:detcomplete} we will see that
for deterministic programs also Algorithm~\ref{alg:base} returns such
a set).  We will do this by induction on the derivation tree;
specifically, by going from the tree \emph{right-to-left,
top-to-bottom}.  To make this induction formal, we will need to
\emph{label} the nodes; to obtain the desired order of traversing the
nodes, we label them with strings, ordered in reverse lexicographic
order.

\begin{definition}\label{lem:labeling}
For a given derivation tree $T$, we label the nodes by strings of
numbers as follows: the root is labelled $0$, and for a tree
\begin{prooftree}
\AxiomC{$T_1$}
\AxiomC{\dots}
\AxiomC{$T_n$}
\TrinaryInfC{$\pi$}
\end{prooftree}
if node $\pi$ is labelled with $\aindex$, then we label each $T_i$
with $\aindex \cdot i$.

We say that $\aindex > \bindex$ if $\aindex$ is larger
than $\bindex$ in the lexicographic ordering (with $\aindex \cdot i
> \aindex$), and $\aindex \succ \bindex$ if $\aindex > \bindex$ but
$\bindex$ is not a prefix of $\aindex$.
\end{definition}

Thus, for nodes labelled
$\aindex$ and $\bindex$, we have $\aindex \succ \bindex$ if $\aindex$
occurs to the right of $\bindex$, and $\aindex > \bindex$ if $\aindex$
occurs to the right or above of $\bindex$.
We have $1 \succ \aindex$ for all $\aindex$ in the tree.

\medskip
In the soundness proof (Lemma~\ref{lem:generalsound}), we essentially
recursed over the steps in the algorithm, and associated to every
extensional value an expression value.  Now, we must go in the other
direction, and associate to every value an extensional value.  If a
value occurs at multiple places in the derivation tree, we do not need
to select the same extensional value every time---just as the
soundness proof did not always associate the same value to a given
extensional value.

In order to choose a suitable extensional value for each position in
the derivation tree, we define the function $\psi$ which considers
the tree above and to the right of a given position.  As a result,
functional values are associated to ever larger extensional values
as we traverse the tree right-to-left, top-to-bottom.

\begin{definition}\label{def:psi}
Let $T$ be a derivation tree and $L$ the set of its labels, which
must all have the form $0 \cdot \aindex$.
For any $v \in \VValue$ (see Definition~\ref{def:VValue}) and
$\aindex \in L \cup \{1\}$, let:
\begin{itemize}
\item $\psi(v,\aindex) = v$ if $v \in \B$
\item $\psi(v,\aindex) = (\psi(v_1,\aindex),\psi(v_2,\aindex))$ if
  $v = (v_1,v_2)$
\item for $\apps{\identifier{f}}{v_1}{v_n} : \btype = \atype_{n+1}
  \arrtype \dots \arrtype \atype_m \arrtype \asortorpair$ with $m >
  n$, let $\psi(\apps{\identifier{f}}{v_1}{v_n},\aindex) =$ \\
  $\{ (e_{n+1},u) \mid \exists \cindex \succ \bindex >
    \aindex\ [$the subtree with index $\bindex$ has a root $\eprog
    \vdashcall \apps{\identifier{f}}{v_1}{v_{n+1}} \arrr w$ with
    $\psi(w,\cindex) = u$ and $e_{n+1} \ssupseteq \psi(v_{n+1},
    \bindex)] \}_\btype$.  In this, $\cindex$ is allowed to
    be $1$ (but $\bindex$ is not).
\end{itemize}
Here, $\ssupseteq$ is defined the same as $\sqsupseteq$, except that
$A_\atype \ssupseteq B_\atype$ if{f} $A \supseteq B$.  Note that clearly
$e \ssupseteq u$ implies $e \sqsupseteq u$, and that $\ssupseteq$ is
transitive by transitivity of $\supseteq$.
\end{definition}

Thus, $\psi(v,\aindex) \in \interpret{\atype}$ for $v : \atype$, but
\emph{not} $\psi(v,\aindex) \in \pinterpret{\atype}$.
Note that $\psi(v,\aindex) \ssupseteq \psi(v,\bindex)$ if $\bindex >
\aindex$ by transitivity of $>$.
Note also that, in the derivation tree for $\eprog \vdashcall
\apps{\symb{start}}{d_1}{d_M} \arrr b$, all values
are in $\VValue$ as all $d_i$ are in $\B$.

\begin{remark}
Some choices in Definition~\ref{def:psi} may well confound the reader.

First, the special label $1$ is used because we will make statements
of the form ``for all $\bindex \succ \aindex$ there exists $o
\ssupseteq \psi(w,\bindex)$ with property P'': if we did not include
$1$ in this quantification, it would give no information about, e.g.,
the root of the tree.
Note that this is already used in the definition of $\psi$.

Second, one may wonder why we use $\ssupseteq$ rather than
$\sqsupseteq$.  This is purely for the sake of the proof: the simpler
and more restrictive relation $\ssupseteq$ works better in the
induction because whenever $A_\atype \ssupseteq B_\atype$, all
elements of $B$ are also in $A$.
In fact, in Algorithm~\ref{alg:general} we could replace all uses of
$\sqsupseteq$ by $\ssupseteq$ without affecting the algorithm's
correctness.  However, this \emph{would} be problematic for
Algorithm~\ref{alg:base}; in the completeness proof in
Section~\ref{subsec:detcomplete}, we will for instance use that
$\{(e,u)\}_\atype \sqsupseteq \{(e,u),(e,o)\}_\atype$ if $u
\sqsupseteq o$, something which does not hold for $\ssupseteq$.
\end{remark}

Now, we could easily follow the proof sketch in the running text and
prove directly that Algorithm~\ref{alg:general} is complete, by
showing for each subtree with a label $\aindex$ and root $\eprog
\vdashcall \apps{\identifier{f}}{v_1}{v_n} \arrr w$ that for all
$e_1,\dots,e_n$ such that each $e_i \ssupseteq \psi(v_i,\aindex)$
and for all $\bindex \succ \aindex$ there exists $o \ssupseteq
\psi(w,p)$ such that $\vdash \apps{\identifier{f}}{e_1}{e_n} \leadsto
o$ is eventually confirmed (and similar for subtrees $\eprog,\gamma
\vdash s \arrr w$).
However, the proof for this is quite long, and we would have to
essentially repeat it with some minor changes when proving
completeness of Algorithm~\ref{alg:base}.

Instead, we will take a slight detour.  We present a new set of
derivation rules built on extensional values, which directly
corresponds to the algorithm.  By these derivation rules---as
presented in Figure~\ref{fig:extensional}---it is easy to see that
$\eprog \vvdashcall \apps{\identifier{f}}{e_1}{e_k} \too o$ if and
only if $\vdash \apps{\identifier{f}}{e_1}{e_k} \leadsto o$ is
eventually confirmed in Algorithm~\ref{alg:general}, and
$\eprog,\eta \vvdashcall s \too o$ if and only if $\eta \vdash s
\leadsto o$ is eventually confirmed.  Thus, the primary work is in
showing that such a derivation exists.

\begin{figure}[!htb]
\vspace{-12pt}
\begin{prooftree}
\AxiomC{}
\LeftLabel{[Constructor]\quad}
\UnaryInfC{$\eprog,\eta \vvdash \apps{\identifier{c}}{s_1}{s_m}
\too \apps{\identifier{c}}{(s_1\eta)}{(s_m\eta)}$}
\end{prooftree}

\begin{prooftree}
\AxiomC{$\eprog,\eta \vvdash s \too o_1$}
\AxiomC{$\eprog,\eta \vvdash t \too o_2$}
\LeftLabel{[Pair]\quad}
\BinaryInfC{$\eprog,\eta \vvdash (s,t) \too (o_1,o_2)$}
\end{prooftree}

\begin{prooftree}
\AxiomC{$\eprog,\eta \vvdash s_i \too o$}
\LeftLabel{[Choice]\quad}
\RightLabel{for $1 \leq i \leq n$}
\UnaryInfC{$\eprog,\eta \vvdash \apps{\choice}{s_1}{s_n} \too o$}
\end{prooftree}

\begin{prooftree}
\AxiomC{$\eprog,\eta \vdash s_1 \too \strue$}
\AxiomC{$\eprog,\eta \vvdash s_2 \too o$}
\LeftLabel{[Cond-True]}
\BinaryInfC{$\eprog,\eta \vvdash \ifte{s_1}{s_2}{s_3} \too o$}
\end{prooftree}

\begin{prooftree}
\AxiomC{$\eprog,\eta \vdash s_1 \too \sfalse$}
\AxiomC{$\eprog,\eta \vvdash s_3 \too o$}
\LeftLabel{[Cond-False]}
\BinaryInfC{$\eprog,\eta \vvdash \ifte{s_1}{s_2}{s_3} \too o$}
\end{prooftree}

\begin{prooftree}
\AxiomC{$\eprog,\eta \vvdash s_i \too e_i$ for $1 \leq i \leq n$}
\LeftLabel{[Variable]\quad}
\RightLabel{$\exists o' \in \eta(x)(e_1,\dots,e_n) [o' \sqsupseteq o]$}
\UnaryInfC{$\eprog,\eta \vvdash \apps{x}{s_1}{s_n} \too o$}
\end{prooftree}

\begin{prooftree}
\AxiomC{$\eprog,\eta \vvdash s_i \too e_i$ for $1 \leq i \leq n$}
\AxiomC{$\eprog \vvdashcall \apps{\identifier{f}}{e_1}{e_n} \too o$}
\RightLabel{\begin{tabular}{l}
  for $\identifier{f} \in \Defineds$, \\
  $n \leq \arity(\identifier{f})$
\end{tabular}}
\LeftLabel{[Func]\quad}
\BinaryInfC{$\eprog,\eta \vvdash \apps{\identifier{f}}{s_1}{s_n} \too o$}
\end{prooftree}

\begin{prooftree}
\AxiomC{$\eprog,\eta \vvdash s_i \too e_1$ for $1 \leq i \leq n$}
\AxiomC{$\eprog \vvdashcall \apps{\identifier{f}}{e_1}{e_k} \too u$}
\RightLabel{\begin{tabular}{l}
  for $\identifier{f} \in \Defineds$, \\
  $n > \arity(\identifier{f})$, \\
  \hspace{-25pt}
  $o' \in u(e_{k+1},\dots,e_n)$, \\
  $o' \sqsupseteq o$
\end{tabular}}
\LeftLabel{[Applied]\quad}
\BinaryInfC{$\eprog,\eta \vvdash \apps{\identifier{f}}{s_1}{s_n} \too o$}
\end{prooftree}

\begin{prooftree}
\AxiomC{$\eprog \vvdashcall \apps{\identifier{f}}{e_1}{e_{n+1}} \too
u' \sqsupseteq u$ for all $(e_{n+1},u) \in O$}
\RightLabel{if $n < \arity(\identifier{f})$}
\LeftLabel{[Value]\quad}
\UnaryInfC{$\eprog \vvdashcall \apps{\identifier{f}}{e_1}{e_n} \too
O_\atype$}
\end{prooftree}

\begin{prooftree}
\AxiomC{$\eprog,\eta \vvdash s \too o$}
\RightLabel{\begin{tabular}{l}
if $\apps{\identifier{f}}{\ell_1}{\ell_k} = s$ is the
first clause in $\eprog$ which \\
matches $\apps{\identifier{f}}{e_1}{e_k}$, and $\eta$ is the
matching \\
ext-environment
\end{tabular}}
\LeftLabel{[Call]\quad}
\UnaryInfC{$\eprog\vvdashcall \apps{\identifier{f}}{e_1}{e_k} \too o$}
\end{prooftree}

\caption{Alternative semantics using (non-deterministic) extensional
values}
\vspace{-12pt}
\label{fig:extensional}
\end{figure}

Thus, we come to the main result needed for completeness:

\begin{lemma}\label{lem:maketree}
If $\progresult$, then
$\eprog \vvdashcall \apps{\symb{start}}{d_1}{d_M} \too b$.
\end{lemma}

\begin{proof}
Given $\progresult$, Lemma~\ref{lem:progeprog} allows us to assume
that $\eprog \vdashcall \apps{\symb{start}}{d_1}{d_M} \arrr b$.  Let
$T$ be the derivation tree with this root (with root label $0$) and
$L$ the set of its labels.
We prove, by induction on $\aindex$ with greater labels handled first
(which is well-founded because $T$ has only finitely many subtrees):
\begin{enumerate}
\item\label{complete:call}
If the subtree with label $\aindex$ has root $\eprog \vdashcall
\apps{\identifier{f}}{v_1}{v_n} \arrr w$, then for all $e_1,\dots,
e_n$ such that each $e_i \ssupseteq \psi(v_i,\aindex)$, and for all
$\bindex \succ \aindex$ there exists $o \ssupseteq \psi(w,\bindex)$
such that $\eprog \vvdashcall \apps{\identifier{f}}{e_1}{e_n} \too
o$.
\item\label{complete:rhs}
If the subtree with label $\aindex$ has root $\eprog,\gamma
\vdash t \arrr w$ and $\eta(x) \ssupseteq \psi(\gamma(x),
\aindex)$ for all $x \in \Var(t)$, then for all $\bindex \succ
\aindex$ there exists $o \ssupseteq \psi(w,\bindex)$ such that
$\eprog,\eta \vvdash t \too o$.
\end{enumerate}
Here, for $\bindex \succ \aindex$ we allow $\bindex \in L \cup \{1\}$.
Therefore, in both cases, there must exist a suitable $o \ssupseteq
\psi(w,1)$ if $w$ is a data expression; this $o$ can only be $w$
itself.
The first item gives the desired result for $\aindex = 0$, as $o
\ssupseteq \psi(b,1)$ implies $o = b$.

We prove both items together by induction on $\aindex$, with greater
labels handled first.  Consider the first item.  There are two cases:

\begin{itemize}
\item If $\eprog \vdashcall \apps{\identifier{f}}{v_1}{v_n} \arrr w$
by [Closure], then $n < \arity(\identifier{f})$ and $w =
\apps{\identifier{f}}{v_1}{v_n}$.  Given $\bindex \succ \aindex$,
let $o := \psi(w,\aindex)$; then clearly $o \ssupseteq
\psi(w,\bindex)$.  We must see that $\eprog \vvdashcall
\apps{\identifier{f}}{e_1}{e_n} \too o$; by [Value], this is the
case if for all $(e_{n+1},u)$ in the set underlying $o$ we can
derive $\eprog \vvdashcall \apps{\identifier{f}}{e_1}{e_{n+1}} \too
u'$ for some $u' \ssupseteq u$.
So let $(e_{n+1},u)$ be in this underlying set.

By definition of $\psi$, we can find $\cindex \succ \bindex' > 
\aindex$ and $v_{n+1},w'$ such that the subtree with label $\bindex'$
has a root $\eprog \vdashcall \apps{\identifier{f}}{v_1}{v_{n+1}}
\arrr w'$ and $e_{n+1} \ssupseteq \psi(v_{n+1},\bindex')$ and $u =
\psi(w',\cindex)$.  Since $\bindex' > \aindex$, also $e_i \ssupseteq
\psi(v_i,\bindex')$ for $1 \leq i \leq n$; thus, the induction
hypothesis provides $u' \ssupseteq \psi(w',\cindex) = u$ with
$\eprog \vdashcall \apps{\identifier{f}}{e_1}{e_{n+1}} \too u'$ as
required.
\item If $\eprog \vdashcall \apps{\identifier{f}}{v_1}{v_n} \arrr w$
by [Call], then $n = \arity(\identifier{f})$ and we can find a
clause, say $\rho\colon \apps{\identifier{f}}{\ell_1}{\ell_n} = s$
and an environment $\gamma$ such that
\begin{enumerate}
\item $\rho$ is the first clause in $\eprog$ whose right-hand side
  is instantiated by $\apps{\identifier{f}}{v_1}{v_n}$;
\item each $v_i = \ell_i\gamma$;
\item\label{completeness:clause} $\eprog,\gamma \vdash s \arrr w$.
\end{enumerate}
By Lemma~\ref{lem:downhelp}, using $\down{v}{V}$ if{f} $V \ssupseteq
\psi(v,\aindex)$, also $\rho$ is the first clause which matches
$\apps{\identifier{f}}{e_1}{e_n}$, and for the matching
ext-environment $\eta$, each $\eta(x) \ssupseteq \psi(\gamma(x),
\aindex) \ssupseteq \psi(\gamma(x),\aindex \cdot 1)$.  Thus using
the induction hypothesis for observation~\ref{completeness:clause},
we find $o \ssupseteq \psi(w,\bindex)$ for all $\bindex \succ
\aindex \cdot 1$.  As this includes every label $\bindex$ with
  $\bindex \succ \aindex$, we are done.
\end{itemize}
Now for the second claim, assume that $\eprog,\gamma \vdash t \arrr w$
(with label $\aindex$) and that $\eta(x) \ssupseteq \psi(\gamma(x),
\aindex)$ for all $x \in \Var(t)$; let $\bindex \succ \aindex$ which
(**) implies $\bindex \succ \aindex \cdot i$ for any string $i$ as
well.  Consider the form of $t$ (taking into account that, following
the transformation of $\prog$ to $\eprog$, we do not need to consider
applications whose head is an $\ifte{}{}{}$ or $\choice$ statement).
\begin{itemize}
\item $t = (t_1,t_2)$; then we can write $w = (w_1,w_2)$ and the
  trees with labels $\aindex \cdot 1$ and $\aindex \cdot 2$ have roots
  $\eprog,\gamma \vdash t_1 \arrr w_1$ and $\eprog,\gamma \vdash t_2
  \arrr w_2$ respectively.  Using observation (**), the induction
  hypothesis provides $o_1,o_2$ such that each $\eprog,\eta \vvdash
  s_i \too o_i \ssupseteq \psi(w_i,\bindex)$; we are done choosing
  $o := (o_1,o_2)$.
\item $t = \apps{\identifier{c}}{t_1}{t_m}$ with $\identifier{c} \in
  \Constructors$; then by Lemma~\ref{lem:safety}, each $s_i\gamma =
  b_i \in \B$; this implies that all $\gamma(x) \in \B$, so $\eta(x) =
  \gamma(x)$, and $\eprog,\eta \vdash t \too o := t\eta$ by
  [Constructor].
\item $t = \apps{\choice}{t_1}{t_n}$; then the immediate subtree is
  $\eprog,\gamma \vdash t_i \arrr w$ for some $i$. By observation
  (**), the induction hypothesis provides a suitable $o$, which
  suffices by rule [Choice] from $\vvdash$.
\item $t = \ifte{t_1}{t_2}{t_3}$; then, as the immediate subtree can
  only be obtained by [If-True] or [If-False], we have either
  $\eprog,\gamma \vdash s_1 \arrr \strue$ and $\eprog,\gamma \vdash
  s_2 \arrr w$, or $\eprog,\gamma \vdash s_1 \arrr \sfalse$ and
  $\eprog,\gamma \vdash s_3 \arrr w$.  Using the induction hypothesis
  for $\bindex = 1$, we have $\eprog,\eta \too \strue$ in the first
  case and $\eprog,\eta \too \sfalse$ in the second.  Using (**) and
  the induction hypothesis as before, we obtain a suitable $o$ using
  the inference rule [Cond-True] or [Cond-False] of $\vvdash$.
\item $t \in \V$, so the tree is obtained by [Instance]; choosing
  $o := \eta(t) \ssupseteq \psi(\gamma(t),\aindex) \ssupseteq
  \psi(\gamma(t),\bindex)$ by (**), we have $o \ssupseteq \psi(w,
  \bindex)$ by transitivity of $\ssupseteq$, and $\eprog,\eta \vvdash
  t \too o$ by [Variable] (as $o \in \{o\} = o()$).
\item $t = \apps{x}{t_1}{t_n}$ with $n > 0$; then there are
  $w_0,\dots,w_n$ such that the root is obtained using:
  \begin{itemize}
  \item $\eprog,\gamma \vdash x \arrr \gamma(x) =: w_0$ by [Instance]
    with label $\aindex \cdot 1^n$;
  \item $n$ subtrees of the form $\eprog,\gamma \vdash t_i \arrr v_i$
    with label $\aindex \cdot 1^{n-i} \cdot 2$ for $1 \leq i \leq n$;
  \item $n$ subtrees of the form $\eprog \vdashcall \app{w_{i-1}}{v_i}
    \arrr w_i$ with label $\aindex \cdot 1^{n-i} \cdot 3$ for $1 \leq
    i \leq n$;
  \item $n$ uses of [Appl], each with conclusion $w_i$ and label
    $\aindex \cdot 1^{n-i}$ for $1 \leq i \leq n$.
  \end{itemize}
  Note that here $w_n = w$.  For $1 \leq i \leq n$ we define $e_i$ and
  $o_{i-1}$ as follows:
  \begin{itemize}
  \item observing that $\aindex \cdot 1^{n-i} \cdot 3 \succ \aindex
    \cdot 1^{n-i} \cdot 2$, the induction hypothesis provides $e_i$
    such that $\eprog,\eta \vdash t_i \too e_i \ssupseteq \psi(v_i,
    \aindex \cdot 1^{n-i} \cdot 3)$
  \item $o_0 := \psi(\gamma(x),\aindex) \ssupseteq \psi(w_i,\aindex
    \cdot 1^{n-1} \cdot 2)$;
  \item for $1 < i \leq n$, let $o_{i-1} := \psi(w_{i-1},\aindex
    \cdot 1^{n-i} \cdot 2)$.
  \end{itemize}
  We also define $o_n := \psi(w_n,\bindex)$.  Then by definition of
  $\psi$, because $\aindex \cdot 1^{n-i} \cdot 3 > \aindex \cdot
  1^{n-i} \cdot 2$ and the former is the label of $\eprog \vdashcall
  \app{w_{i-1}}{v_i} \arrr w_i$, there is an element
  $(e_i,\psi(w_i,\cindex))$ in the set underlying $o_{i-1}$ for any
  $\cindex \succ \aindex \cdot 1^{n-i} \cdot 3$.  In particular, this
  means $(e_i,o_i)$ is in this set, whether $i < n$ or $i = n$.
  Thus, by a quick induction on $i$ we have $o_i \in \eta(x)(e_1,
  \dots,e_i)$, so $\eprog,\eta \vvdash s \too o_n = \psi(w,
  \bindex)$ by [Variable].
\item $t = \apps{\identifier{f}}{t_1}{t_n}$ with $n \leq
  \arity(\identifier{f})$; then there are subtrees $\eprog,\gamma
  \vdash t_i \arrr v_i$ labelled $\aindex \cdot 1^{n-i} \cdot 2$ and
  $\eprog \vdash \apps{\identifier{f}}{v_1}{v_n} \arrr w$ labelled
  $\aindex \cdot 3$.  By the induction hypothesis, there are $e_1,
  \dots,e_n$ such that $\eprog,\eta \vvdash t_i \too e_i \ssupseteq
  \psi(v_i,\aindex \cdot 3)$ for $1 \leq i \leq n$.  Therefore, by
  the $\vdashcall$ part of the induction hypothesis and (**), there
  is $o \ssupseteq \psi(w,\bindex)$ such that $\eprog,\eta \vvdash
  \apps{\identifier{f}}{e_1}{e_n} \too o$.  But then
  $\eprog,\eta \vvdash t \too o$ by [Func].
\item $t = \apps{(\apps{\identifier{f}}{s_1}{s_k})}{t_1}{t_0}$ with
  $k = \arity(\identifier{f})$ and $n > 0$; then there are
  subtrees:
  \begin{itemize}
  \item $\eprog,\gamma \vdash \apps{\identifier{f}}{s_1}{s_k} \arrr
    w_0$ by [Function] or [Appl], with label $\aindex \cdot 1^n$;
  \item $\eprog,\gamma \vdash t_i \arrr v_i$ with label $\aindex \cdot
    1^{n-i} \cdot 2$ for $1 \leq i \leq n$;
  \item $\eprog \vdashcall \app{w_{i-1}}{v_i} \arrr w_i$ with label
    $\aindex \cdot 1^{n-i} \cdot 3$ for $1 \leq i \leq n$.
  \end{itemize}
  For some $w_0,\dots,w_n$ with $w_n = w$.  In the same way as the
  previous case, there exists $o_0 \ssupseteq \psi(w_0,\aindex \cdot
  1^{n-1} \cdot 2)$ such that $\eprog,\eta \vvdash \apps{\identifier{
  f}}{s_1}{s_k} \too o_0$ (as $\aindex \cdot 1^{n-1} \cdot 2$).  The
  remainder of this case follows the case with $t = \apps{x}{t_1}{t_n
  }$.
  \qed
\end{itemize}
\end{proof}

Now we may forget $\psi$ and $\ssupseteq$ altogether: from a
derivation for $\progresult$ we have obtained a derivation
$\eprog \vvdash \apps{\identifier{start}}{d_1}{d_M} \too b$, which
almost exactly corresponds to a derivation of
$\vdash \apps{\identifier{start}}{d_1}{d_M} \leadsto b$ using
Algorithm~\ref{alg:general}.  It only remains to formalise this
correspondence:

\begin{lemma}\label{lem:usetree}
If $\eprog \vvdashcall \apps{\symb{start}}{d_1}{d_M} \too b$, then
Algorithm~\ref{alg:general} returns a set containing $b$.
If $\eprog \vvdashcall \apps{\symb{start}}{d_1}{d_M} \too b$ has a
derivation tree which only uses deterministic extensional values,
then so does Algorithm~\ref{alg:base}.
\end{lemma}

\begin{proof}
This is entirely straightforward.
Starting with $\B = \B_{d_1,\dots,d_M}^\prog$, we show:
\begin{enumerate}
\item If $\eprog \vvdashcall \apps{\identifier{f}}{e_1}{e_n} \too o$,
  then $\vdash \apps{\identifier{f}}{e_1}{e_n} \leadsto o$ is
  eventually confirmed.
\item If $\eprog,\eta \vvdash s \too o$, then $\eta \vdash s \leadsto
  o$ is eventually confirmed.
\end{enumerate}
Both statements hold regardless of which algorithm is used, provided
that all extensional values in the derivation tree are among those
considered by the algorithm.
We prove the statements together by induction on the derivation tree.
For the first, there are two inference rules that might have been
used:
\begin{description}
\item[Value] $o = O_\atype$ and for all $(e_{n+1},u) \in O$ there
  exists $u' \sqsupseteq u$ such that $\eprog \vvdashcall
  \apps{\identifier{f}}{e_1}{e_{n+1}} \too u'$ is an immediate
  subtree.  By the induction hypothesis, each such statement
  $\apps{\identifier{f}}{e_1}{e_{n+1}} \leadsto u'$ is confirmed, so
  the current statement is confirmed by step~\ref{alg:iterate:value}.
\item[Call] Immediate by the induction hypothesis and
  step~\ref{alg:iterate:lhs}.
\end{description}
For the second, suppose $\eprog,\eta \vvdash s \too o$, and consider
the inference rule used to derive this.
\begin{description}
\item[Constructor] Immediate by step~\ref{alg:prepare:base}.
\item[Pair] Immediate by the induction hypothesis and
  step~\ref{alg:iterate:pair}.
\item[Choice] Immediate by the induction hypothesis and
  step~\ref{alg:iterate:choice}.
\item[Cond-True] Immediate by the induction hypothesis and
  step~\ref{alg:iterate:ifte}.
\item[Variable] If $n = 0$, then $\eta(x) \sqsupseteq o$, so the
  statement is confirmed in step~\ref{alg:prepare:base}.
  Otherwise, by the induction hypothesis $\eta \vdash s_i \leadsto
  e_i$ is confirmed for each $i$ and $o' \sqsupseteq o$ for some
  $o' \in \eta(x)(e_1,\dots,e_n)$; the statement is confirmed in
  step~\ref{alg:iterate:rhs:var}.
\item[Func] Immediate by the induction hypothesis and
  step~\ref{alg:iterate:rhs:call}.
\item[Applied] Immediate by the induction hypothesis and
  step~\ref{alg:iterate:rhs:extra}.
  \qed
\end{description}
\end{proof}

At this point, we have all the components for
Lemma~\ref{lem:generalcomplete}.

\oldcounter{\ncompletenesslem}
\begin{lemma}
If $\progresult$, then Algorithm~\ref{alg:general} returns a set
$A \cup \{b\}$.
\end{lemma}
\startappendixcounters

\begin{proof}
Immediate by a combination of Lemmas~\ref{lem:maketree}
and~\ref{lem:usetree}.
\qed
\end{proof}

\subsection{Completeness of Algorithm~\ref{alg:base}}\label{subsec:detcomplete}

Now we turn to the deterministic case.  By Lemma~\ref{lem:usetree},
it suffices if we can find a derivation of $\eprog \vvdashcall
\apps{\symb{start}}{d_1}{d_M} \too b$ which uses only deterministic
extensional values (elements of some $\pinterpret{\atype}$).
While the tree that we built in Lemma~\ref{lem:maketree} does not
have this property, we will use it to build a tree which does.

To start, we will see that the conclusions in any derivation tree
are consistent, where \emph{consistency} of two
extensional values is defined as follows:
\begin{itemize}
\item $\consistent{d}{b}$ if{f} $d = b$ for $d,b \in \Data$;
\item $\consistent{(e_1,u_1)}{(e_2,u_2)}$ if{f} both
  $\consistent{e_1}{e_2}$ and $\consistent{u_1}{u_2}$;
\item $\consistent{A_\atype}{B_\atype}$ if{f} for all $(e_1,u_1) \in
  A$ and $(e_2,u_2) \in B$: if $\consistent{e_1}{e_2}$ then
  $\consistent{u_1}{u_2}$.
\end{itemize}

Consistency is preserved under taking ``smaller'' extensional values:

\begin{lemma}\label{lem:consistencypreserve}
If $e_1' \sqsupseteq e_1$, $e_2' \sqsupseteq e_2$ and
$\consistent{e_1'}{e_2'}$, then also $\consistent{e_1}{e_2}$.
\end{lemma}

\begin{proof}
By induction on the form of $e_1$.  If $e_1 \in \B$, then $e_1' = e_1
= e_2 = e_2'$.  If $e_1$ is a pair, then so is $e_1'$ and we use the
induction hypothesis.  Finally, suppose $e_1 = B^1_\atype, e_2 =
B^2_\atype, e_1' = A^1_\atype$ and $e_2' = A^2_\atype$.  Then for all
$(u_1,o_1) \in B^1$ and $(u_2,o_2) \in B^2$, there are $o_1'
\sqsupseteq o_1$ and $o_2' \sqsupseteq o_2$ such that $(u_1,o_1')
\in A^1$ and $(u_2,o_2') \in A^2$.  Now suppose $\consistent{u_1}{
u_2}$.  By consistency of $e_1$ and $e_2$, we then have
$\consistent{o_1'}{o_2'}$, so by the induction
hypothesis, also $\consistent{o_1}{o_2}$.  This gives
$\consistent{B_\atype^1}{B_\atype^2}$, so
$\consistent{e_1}{e_2}$.
\qed
\end{proof}

We can use this to see that conclusions using $\too$ are consistent;
that is, if all extensional values on the left of $\too$ are
consistent, then so are those on the right:

\begin{lemma}\label{lem:consistency}
Let $T_1,T_2$ be derivation trees for $\vvdash$, and let
$\treeroot(T_1),\treeroot(T_2)$ denote their roots.
Suppose given $o,o'$ such that one of the following holds:
\begin{enumerate}
\item There are $\identifier{f},e_1,\dots,e_n,e_1',\dots,e_n'$ such
  that:
  \begin{itemize}
  \item $\treeroot(T_1) = \apps{\identifier{f}}{e_1}{e_n} \too o$;
  \item $\treeroot(T_2) = \apps{\identifier{f}}{e_1'}{e_n'} \too o'$;
  \item $\consistent{e_1}{e_1'}$,\dots,$\consistent{e_n}{e_n'}$.
  \end{itemize}
\item There are $\eta,\eta'$ on the same domain and $s$ such that:
  \begin{itemize}
  \item $\treeroot(T_1) = \eprog,\eta \vvdash s \too o$;
  \item $\treeroot(T_2) = \eprog,\eta' \vvdash s \too o'$;
  \item $\consistent{\eta(x)}{\eta'(x)}$ for all $x$ occurring in
    $s$.
  \end{itemize}
  Moreover, $s$ has no sub-expressions of the form
  $\apps{(\ifte{b}{s_1}{s_2})}{t_1}{t_n}$ with $n > 0$.
\end{enumerate}
If $\choice$ does not occur in $s$ or any clause of $\eprog$, then
$\consistent{o}{o'}$.
\end{lemma}

\begin{proof}
Both statements are proved together by induction on the form of $T_1$.
For the first, consider $n$.  Since $\treeroot(T_1) $ could be
derived, necessarily $n \leq \arity(\identifier{f})$.  There are two
cases:
\begin{itemize}
\item $n < \arity(\identifier{f})$; both trees were derived by
  [Value].  Thus, we can write $o = A_\atype$ and $o' = A'_\atype$
  and have:
  \begin{itemize}
  \item for all $(e_{n+1},u_1) \in A$ there is some $u_2 \sqsupseteq
    u_1$ such that $T_1$ has an immediate subtree $\eprog \vvdash
    \app{\apps{\identifier{f}}{e_1}{e_n}}{e_{n+1}} \too u_2$;
  \item for all $(e_{n+1}',u_1') \in A'$ there is some $u_2'
    \sqsupseteq u_1'$ such that $T_2$ has an immediate subtree
    $\eprog \vvdash \app{\apps{\identifier{f}}{e_1'}{e_n'}}{e_{n+1}'}
    \too u_2'$.
  \end{itemize}
  Now let $(e_{n+1},u_1) \in A$ and $(e_{n+1},u_1') \in B$ be such
  that $\consistent{e_{n+1}}{e_{n+1}'}$.  Considering the two
  relevant subtrees, the induction hypothesis gives that
  $\consistent{u_2}{u_2'}$.
  By Lemma~\ref{lem:consistencypreserve} we then obtain the required
  property that $\consistent{u_1}{u_1'}$.
\item $n = \arity(\identifier{f})$; both trees were derived by
  [Call]. Given that extensional values of the form $A_\atype$ can
  only instantiate variables (not pairs or patterns with a constructor
  at the head), a reasoning much like the one in
  Lemma~\ref{lem:downhelp} gives us that both conclusions are
  obtained by the same clause $\apps{\identifier{f}}{\ell_1}{\ell_k}
  = s$, the first with ext-environment $\eta$ and the second with
  $\eta'$ such that each $\consistent{\eta(x)}{\eta'(x)}$.  Then
  the immediate subtrees have roots $\eprog,\eta \vvdash s \too o$
  for $T_1$ and $\eprog,\eta' \vvdash s \too o'$ for $T_2$, and we
  are done by the induction hypothesis.
\end{itemize}
For the second statement, let $T_1$ have a root $\eta \vvdash s \too
o$ and $T_2$ a root $\eta' \vvdash \too o'$, and assume that $s$ does
not contain any $\choice$ operators or if-statements at the head of
an application.  In addition, let $\consistent{\eta(x)}{\eta'(x)}$ for
all (relevant) $x$.  Then $s$ may have one of six forms:
\begin{itemize}
\item $s = \apps{\identifier{c}}{s_1}{s_m}$: then $o = s\eta$ and
  $o' = s\eta'$; as, in this case, necessarily all variables have a
  type of order $0$, $o = o'$ which guarantees consistency.
\item $s = (s_1,s_2)$: then $o = (o_1,o_2)$ and $o' = (o_1',o_2')$,
  and by the induction hypothesis both $\consistent{o_1}{o_1'}$ and
  $\consistent{o_2}{o_2'}$; thus indeed $\consistent{o}{o'}$.
\item $s = \ifte{s_1}{s_2}{s_3}$: since not
  $\consistent{\strue}{\sfalse}$, either both conclusions are derived
  by [Cond-True] or by [Cond-False]; consistency of $o$ and $o'$
  follows immediately by the induction hypothesis on the second
  subtree.
\item $s = \apps{x}{s_1}{s_n}$; the induction hypothesis provides
  $e_1,\dots,e_n$ and $e_1',\dots,e_n'$ such that each $\consistent{
  e_i}{e_i'}$ and there are $u \sqsupseteq o,u' \sqsupseteq o'$
  such that $u \in \eta(x)(e_1,\dots,e_n)$ and $u' \in \eta'(x)
  (e_1',\dots,e_n')$.  By Lemma~\ref{lem:consistencypreserve}, it
  suffices if $u$ and $u'$ are consistent.  We prove this by
  induction on $n$:
  \begin{itemize}
  \item if $n = 0$ then $u = \eta(x)$ and $u' = \eta'(x)$ and
    consistency is assumed;
  \item if $n > 0$ then there are $A_\atype \in \eta(x)(e_1,\dots,
    e_{n-1})$ and $B_\atype \in \eta'(x)(e_1',\dots,\linebreak
    e_{n-1}')$ such
    that $(e_n,u) \in A$ and $(e_n',u') \in B$.  By the induction
    hypothesis, $\consistent{A_\atype}{B_\atype}$.  Since also
    $\consistent{e_n}{e_n'}$, this implies $\consistent{u}{u'}$.
  \end{itemize}
\item $s = \apps{\identifier{f}}{s_1}{s_n}$ with $n \leq
  \arity(\identifier{f})$; then both conclusions follow by
  [Func].  The immediate subtrees provide $e_1,\dots,e_n$ and
  $e_1',\dots,e_n'$ such that, by the induction hypothesis, each
  $\consistent{e_i}{e_i'}$, as well as a conclusion $\eprog \vvdashcall
  \apps{\identifier{f}}{e_1}{e_n} \too o$ in $T_1$ and $\eprog
  \vvdashcall \apps{\identifier{f}}{e_1'}{e_n'} \too o'$ in $T_2$; we
  can use the first part of the induction hypothesis to conclude
  $\consistent{o}{o'}$.
\item $s = \apps{\identifier{f}}{s_1}{s_n}$ with $n > \arity(
  \identifier{f})$; then both conclusions follow by [Applied].  There
  are $e_1,\dots,e_n,e_1',\dots,e_n'$ such that, by the induction
  hypothesis, each $\consistent{e_i}{e_i'}$.  Moreover, there are
  $u,u'$ such that $T_1$ has a subtree with root $\eprog \vvdashcall
  \apps{\identifier{f}}{e_1}{e_k} \too u$ and $T_2$ has a subtree with
  root $\eprog \vvdashcall \apps{\identifier{f}}{e_1'}{e_k'} \too u'$,
  where $k = \arity(\identifier{f})$; by the induction hypothesis,
  clearly $\consistent{u}{u'}$, and since there are $o_2,o_2'$ such
  that $u(e_{k+1},\dots,e_n) \ni o_2 \sqsupseteq o$ and $u'(e_{k+1}',
  \dots,e_n') \ni o_2' \sqsupseteq o'$, the induction argument in
  the variable case provides $\consistent{o_2}{o_2'}$, so
  $\consistent{o}{o'}$ by Lemma~\ref{lem:consistencypreserve}.
  \qed
\end{itemize}
\end{proof}

This result implies that all (non-deterministic) extensional values in
the derivation tree are \emph{internally consistent}:
$\consistent{o}{o}$.  

Now, recall that our mission is to transform a derivation tree for
$\eprog \vvdash \apps{\identifier{start}}{d_1}{d_M} \too b$ into one
which uses only deterministic extensional values.  A key step to this
will be to define a deterministic extensional value $o' \sqsupseteq o$
for every non-deterministic extensional value $o$ occurring in the
tree.  Using consistency, we can do that; $o'$ is chosen to be $\sqcup
\{o\}$ defined below:

\begin{definition}
Given a \emph{non-empty}, \emph{consistent} set $X$---that is,
$\emptyset \neq X \subseteq \interpret{\atype}$ with $\consistent{e}{
u}$ for all $e,u \in X$---let $\sqcup X \in \pinterpret{\atype}$ be
defined as follows:
\begin{itemize}
\item if $\atype \in \Sorts$, then by consistency $X$ can only have
  one element; we let $\sqcup \{ d \} = d$;
\item if $\atype = \atype_1 \times \atype_2$, then $\sqcup X =
  (\sqcup \{ e \mid (e,u) \in X \}, \sqcup \{ u \mid (e,u) \in X \})$
  \\
  \emph{(this is well-defined because $\consistent{(e_1,u_1)}{(e_2,
  u_2)}$ implies both $\consistent{e_1}{e_2}$ and $\consistent{u_1}{
  u_2}$, so indeed the two sub-sets are consistent)}
\item if $\atype = \atype_1 \arrtype \btype$, then $\sqcup X =
  \{ (e,\sqcup Y_e) \mid e \in \pinterpret{\atype} \wedge Y_e =
  \bigcup_{A_\atype \in X} \{ o \mid (u,o) \in A \wedge e \sqsupseteq
  \sqcup \{u\}\} \wedge Y_e \neq \emptyset \}_{\atype_1 \arrtype
  \atype_2}$ \\
  \emph{(this is well-defined because for every $e$ there is only
  one $Y_e$, and $Y_e$ is indeed consistent: if $o_1,o_2 \in Y$, then
  there are $A^{(1)}_\atype,A^{(2)}_\atype \in X$ and there exist
  $u_1,u_2$ such that $(u_1,p_1) \in A^{(1)}, (u_2,o_2) \in A^{(2)}$
  and both $e \sqsupseteq u_1$ and $e \sqsupseteq u_2$; by
  Lemma~\ref{lem:consistencypreserve}---using that $\consistent{e}{e}$
  because $e \in \pinterpret{\atype}$---the latter implies that
  $\consistent{u_1}{u_2}$, so by consistency of $A_\atype^{(1)}$ and
  $A_\atype^{(2)}$ indeed $\consistent{o_1}{o_2}$)}
\end{itemize}
\end{definition}

Now, the transformation of a $\choice$-free---and therefore
consistent---derivation tree for $\eprog \vvdashcall
\apps{\identifier{start}}{d_1}{d_M} \too b$ into one which uses only
deterministic extensional values is detailed in
Lemma~\ref{lem:dettree}.  The transformation is mostly done by a
fairly straightforward induction on the depth of the tree (or more
precisely, on the maximum depth of a set of trees), but to
guarantee correctness we will be obliged to assert a number of
properties of $\sqcup$.  This is done in
Lemmas~\ref{lem:pintergreater}--\ref{lem:apply}.

\medskip
To start, we derive a kind of monotonicity for $\sqcup$ with respect
to $\sqsupseteq$:

\begin{lemma}\label{lem:sqsuppreserve}\label{lem:sqcupsup}
Let $X,Y \subseteq \interpret{\atype}$ be non-empty consistent sets,
and suppose that for every $e \in Y$ there is some $e' \in X$ such that
$e' \sqsupseteq e$.  Then $\sqcup X \sqsupseteq \sqcup Y$.

In particular, if $X \supseteq Y$ then $\sqcup X \sqsupseteq \sqcup
Y$.
\end{lemma}

\begin{proof}
The second statement follows immediately from the first, since
$\sqsupseteq$ is reflexive.  For the first statement, we use
induction on the form of $\atype$.

If $\atype \in \Sorts$ there is little to prove: $X$ and $Y$
contain the same single element.

If $\atype = \atype_1 \times \atype_2$, then $\sqcup X = (\sqcup
\{ u \mid (u,o) \in X \}, \sqcup \{ o \mid (u,o) \in X \})$ and
$\sqcup Y = (\sqcup \{ u \mid (u,o) \in Y \}, \sqcup \{ o \mid (u,o)
\in Y \})$.  Since, for every $u$ in $\{ u \mid (u,o) \in Y \}$
there is some $(u',o') \in X$ with $u' \sqsupseteq u$ (by definition
of $\sqsupseteq$ for pairs), the contain\-ment property also holds for
the first sub-set; it is as easily obtained for the second.  Thus we
complete by the induction hypothesis and the definition of
$\sqsupseteq$.

Otherwise $\atype = \atype_1 \arrtype \atype_2$; denote $\sqcup X =
A_\atype$ and $\sqcup Y = B_\atype$.  Now, all elements of $B$ can be
written as $(u,\sqcup Y_u)$ where $Y_u = \bigcup_{D_\atype \in Y}
\{ o \mid (u',o) \in D \wedge u \sqsupseteq \sqcup \{u'\} \}$, and all
elements of $A$ as $(u,\sqcup X_u)$, where $X_u = \bigcup_{C_\atype
\in X} \{ o \mid (u',o) \in C \wedge u \sqsupseteq \sqcup \{u'\} \}$.
Let $(u,\sqcup Y_u) \in B$; we claim that (1) $X_u$ is non-empty, (2)
$(u,\sqcup X_u) \in A$ and (3) $\sqcup X_u \sqsupseteq \sqcup Y_u$,
which suffices to conclude $\sqcup X \sqsupseteq \sqcup Y$.
\begin{enumerate}
\item $(u,\sqcup Y_u) \in B$ gives that $Y_u$ is non-empty, so it has
  at least one element $o$ with
  $(u',o) \in D$ for some $D_\atype \in Y$; by assumption, there is
  $C_\atype \in X$ with $C_\atype \sqsupseteq D_\atype$, which
  implies that $(u',o') \in C_\atype$ for some $o' \sqsupseteq o$; as
  $u \sqsupseteq u'$ we have $o' \in X_u$;
\item follows from (1);
\item for all $o \in Y_u$, there are $u'$ with $u \sqsupseteq u'$ and
  $D_\atype \in Y$ such that $(u',o) \in D$, and by assumption
  $C_\atype \in X$ and $(u',o') \in C$ with $o' \sqsupseteq o$; as
  $u \sqsupseteq u'$, we have $o' \in X_u$.  The induction hypothesis
  therefore gives $\sqcup X_u \sqsupseteq \sqcup Y_u$.
\qed
\end{enumerate}
\end{proof}

We can think of $\sqcup$ as defining a kind of \emph{supremum}:
$\sqcup X$ is the supremum of the set $\{ \sqcup \{e\} \mid e \in X
\}$ with regards to the ordering relation $\sqsupseteq$.  The first
part of this is given by Lemma~\ref{lem:sqcupsup}: $\sqcup X$ is
indeed greater than $\sqcup \{e\}$ for all $e \in X$ because $X
\supseteq \{e\}$.
The second part, that $\sqcup X$ is the \emph{smallest} deterministic
extensional value with this property, holds by
Lemma~\ref{lem:pintergreater}:

\begin{lemma}\label{lem:pintergreater}
Let $X = X^{(1)} \cup \dots \cup X^{(n)} \subseteq \interpret{\atype}$
be a consistent set with $n > 0$ and all $X^{(i)}$ non-empty, and let
$e \in \pinterpret{\atype}$ be such that $e \sqsupseteq \sqcup X^{(i)}$
for all $1 \leq i \leq n$.  Then $e \sqsupseteq \sqcup X$.
\end{lemma}

\begin{proof}
By induction on the form of $\atype$.

If $\atype \in \Sorts$, then each $\sqcup X^{(i)} = e$; thus, $X^{(1)
} = \dots = X^{(n)} = X = \{e\}$ and $\sqcup X = e$ as well.

If $\atype = \atype_1 \times \atype_2$, then $e = (e_1,e_2)$ and
$\sqcup X = (\sqcup Y_1,\sqcup Y_2)$, where $Y_j = \{ u_j \mid (u_1,
u_2) \in X \}$ for $j \in \{1,2\}$.
Let $Y_j^{(i)} = \{ u_j \mid (u_1,u_2) \in X^{(i)} \}$.
Then clearly each $Y_j = Y_j^{(1)} \cup \dots \cup Y_j^{(n)}$, and
$e \sqsupseteq \sqcup X^{(i)}$ implies that each $e_j \sqsupseteq
\sqcup Y_j^{(i)}$.  The induction hypothesis gives $e_j \sqsupseteq
\sqcup Y_j$ for both $j$.

If $\atype = \atype_1 \arrtype \atype_2$, then write $e = A_\atype$.
Now,
\begin{itemize}
\item for $u \in \pinterpret{\atype_1}$,
  denote $Y_u^{(i)} = \bigcup_{B_\atype \in X^{(i)}} \{ o \mid
  (u',o) \in B \wedge u \sqsupseteq \sqcup \{u'\} \}$;
\item for $(u,\sqcup Y_u) \in \sqcup X$, we can write $Y = Y_u^{(1)}
  \cup \dots \cup Y_u^{(N)}$;
\item for $(u,\sqcup Y_u) \in \sqcup X$, some $Y_u^{(i)}$ must be
  non-empty;
\item as $(u,\sqcup Y_u^{(i)}) \in \sqcup X^{(i)}$,
  there exists $(u,o') \in A$ with $o' \sqsupseteq \sqcup Y_u^{(i)}$;
\item as there is only one $o'$ with $(u,o') \in A$, we obtain
  $o' \sqsupseteq \sqcup Y_u^{(j)}$ for all non-empty $Y_u^{(j)}$;
\item by the induction hypothesis, $o' \sqsupseteq \sqcup
  (Y_u^{(1)} \cup \dots \cup Y_u^{(N)}) = \sqcup Y_u$.
\end{itemize}
Thus, $A_\atype \sqsupseteq \sqcup X$ as required.
\qed
\end{proof}

The next helper result concerns application of deterministic
extensional values as (partial) functions.  Very roughly, we see that
if $e$ is at least the ``supremum'' of $\{e^{(1)},\dots,e^{(m)}\}$,
then the result $c \in e(u_1,\dots,u_n)$ of applying $e$ to some
extensional values $u_1,\dots,u_n$ is at least as large as each
$\sqcup\ e^{(j)}(u_1,\dots,u_n)$.

The lemma is a bit broader than this initial sketch, however, as it
also allows for the $e^{(j)}$ to be applied on smaller $u_i^{(j)}$;
i.e., we actually show that $c \sqsupseteq \sqcup\ e^{(j)}(u_1^{(j)},
\dots,u_n^{(j)})$ if each $u_i \sqsupseteq \sqcup \{ u_i^{(j)} \mid 1
\leq j \leq m \}$.  Formally:\pagebreak

\begin{lemma}\label{lem:apply}
Let $n \geq 0$ and suppose that:
\begin{itemize}
\item $\pinterpret{\atype_1 \arrtype \dots \arrtype \atype_n \arrtype
  \btype} \ni e \sqsupseteq \sqcup \{ e^{(1)},\dots,e^{(m)} \}$;
\item $\pinterpret{\atype_i} \ni u_i = \sqcup \{ u_i^{(1)},\dots,u_i^{(m)}
  \}$ for $1 \leq i \leq n$;
\item $e^{(j)}(u_1^{(j)},\dots,u_n^{(j)}) \ni c^{(j)} \sqsupseteq
  o^{(j)}$ for $1 \leq j \leq m$;
\item $o = \sqcup \{ o^{(1)},\dots,o^{(m)} \}$.
\end{itemize}
Then there exists $c \in \pinterpret{\btype}$ such that $e(u_1,\dots,u_n)
\ni c \sqsupseteq o$.
\end{lemma}

\begin{proof}
By induction on $n$.  First suppose that $n = 0$, so each $e^{(j)} =
c^{(j)} \sqsupseteq o^{(j)}$.  Then $e \sqsupseteq \sqcup \{ e^{(1)},
\dots, e^{(m)} \} \sqsupseteq \sqcup \{ o^{(1)},\dots,o^{(m)} \} = o$
by Lemma~\ref{lem:sqsuppreserve}, so $e() \ni e \sqsupseteq o$ by
transitivity of $\sqsupseteq$.

Now let $n > 0$.  For $1 \leq j \leq m$ the third observation gives
$A^{(j)}$ such that $e^{(j)}(u_1^{(j)},\dots,u_{n-1}^{(j)}) \ni
A^{(j)}_{\atype_n \arrtype \btype}$ and $(u_n^{(j)},c^{(j)}) \in
A^{(j)}$.
Then by the induction hypothesis, there exists $A_{\atype \arrtype
\btype} \in e(u_1,\dots,u_{n-1})$ such that $A_{\atype_n \arrtype
\btype} \sqsupseteq \sqcup \{ A^{(1)},\dots,A^{(m)} \}$.
That is, omitting the subscript $n$:

\begin{itemize}
\item $\pinterpret{\atype \arrtype \btype} \ni
  A_{\atype \arrtype \btype} \sqsupseteq \sqcup
  \{ A^{(1)}_{\atype \arrtype \btype},\dots,A^{(m)}_{\atype \arrtype
  \btype} \}$;
\item $\pinterpret{\atype} \ni u = \sqcup \{ u^{(1)}, \dots, u^{(m)} \}$;
\item for $1 \leq j \leq m$: $(u^{(j)},c^{(j)}) \in A^{(j)}$ for some
  $c^{(j)} \sqsupseteq o^{(j)}$;
\item $o = \sqcup \{ o^{(1)},\dots,o^{(m)} \}$.
\end{itemize}
Moreover, for every $c$ such that $(u,c) \in A$ also $c \in e(u_1,
\dots,u_n)$; thus, we are done if we can identify such $c \sqsupseteq
o$.

Let $B_u^{(j)} := \{ o' \mid (u',o') \in A^{(j)} \wedge u
\sqsupseteq \sqcup \{ u' \} \}$
and let $B_u := B_u^{(1)} \cup \dots \cup B_u^{(m)}$.  Then we have:
\begin{itemize}
\item $c^{(j)} \in B_u$ for $1 \leq j \leq m$:
  since $(u^{(j)},c^{(j)}) \in A^{(j)}$, and $u = \sqcup \{
  u^{(1)},\dots,u^{(m)} \} \sqsupseteq \sqcup \{ u^{(j)} \}$ by
  Lemma~\ref{lem:sqcupsup}, we have $c^{(j)} \in B_u^{(j)} \subseteq
  B_u$;
\item since therefore $B_u \neq \emptyset$, the pair
  $(u,\sqcup B_u)$ occurs in the set underlying $\sqcup \{ A^{(1)},
  \dots,A^{(m)} \}$;
\item since $A_{\atype \arrtype \btype} \sqsupseteq \sqcup \{
  A^{(1)}_{\atype \arrtype \btype},\dots,A^{(m)}_{\atype \arrtype
  \btype}\}$, there exists $c \sqsupseteq \sqcup B_u$ such that
  $(u,c) \in A$;
\item since $A_{\atype \arrtype \btype} \in \pinterpret{\atype_n
  \arrtype \btype}$, there is only one choice for $c$;
\item $c \sqsupseteq \sqcup B_u \sqsupseteq \sqcup \{c^{(j)}\}
  \sqsupseteq \sqcup \{o^{(j)}\}$ for all $1 \leq j \leq m$ by
  Lemmas~\ref{lem:sqcupsup} and~\ref{lem:sqsuppreserve};
\item therefore $c \sqsupseteq o$ by Lemma~\ref{lem:pintergreater}.
  \qed
\end{itemize}
\end{proof}

At last, all preparations done.  We now turn to the promised proof
that in a deterministic setting, it suffices to consider deterministic
extensional values.

\begin{lemma}\label{lem:dettree}
If $\prog$ is deterministic and $\eprog \vdashcall
\apps{\symb{start}}{d_1}{d_M} \too b$, then this can be derived using
a \emph{functional tree}: a derivation tree where all extensional
values are in some $\pinterpret{\atype}$.
\end{lemma}

\begin{proof}
Let $\eprog$ be deterministic (so also $\eprog$ is).
For the sake of a stronger induction hypothesis, it turns out to be
useful to use induction on \emph{sets} of derivation trees, rather
than a single tree.  Specifically, we prove the following statements:
\begin{enumerate}
\item Suppose $T_1,\dots,T_N$ are derivation trees, and there are
  fixed $\identifier{f},n$ such that each tree $T_j$ has a root
  $\eprog \vvdash \apps{\identifier{f}}{e^{(j)}_1}{e^{(j)}_n} \too
  o^{(j)}$, where $\consistent{e^{(j)}_i}{e^{(k)}_i}$ for all
  $1 \leq j,k \leq N$ and $1 \leq i \leq n$.
  Let $e_1,\dots,e_n$ be deterministic extensional values such
  that $e_i \sqsupseteq \sqcup \{ e^{(j)}_i \mid 1 \leq j \leq N \}$
  for $1 \leq i \leq n$.
  We can derive $\eprog \vvdashcall \apps{\identifier{f}}{e_1}{e_n}
  \too o := \sqcup \{ o^{(j)} \mid 1 \leq j \leq N \}$
  by a functional tree.
\item Suppose $T_1,\dots,T_N$ are derivation trees, and there is some
  fixed $s$ such that each tree $T_j$ has a root $\eprog,\eta^{(j)}
  \vvdash s \too o^{(j)}$, where $\consistent{\eta^{(j)}(x)}{\eta^{
  (k)}(x)}$ for all $1 \leq j,k \leq N$ and variables $x$ in the
  shared domain.  Let $\eta$ be an ext-environment on the same domain
  mapping to functional extensional values such that $\eta(x)
  \sqsupseteq \sqcup \{ \eta^{(j)}(x) \mid 1 \leq j \leq N \}$ for
  all $x$.
  Writing $o := \sqcup \{ o^{(j)} \mid 1 \leq j \leq N \}$, we can
  derive $\eprog,\eta \vvdash s \too o$ by a functional tree.
  (We assume that no sub-expression of $s$ has an if-then-else
  at the head of an application.)
\end{enumerate}
The first of these claims proves the lemma for $N = 1$: clearly data
expressions are self-consistent, and the only $o \sqsupseteq b =
\sqcup \{b\}$ is $b$ itself, so the claim says that the root can be
derived using a functional tree.

We prove the claims together by a shared induction on the maximum
depth of any $T_j$.  We start with the first claim.  There are two
cases:
\begin{itemize}
\item $n = \arity(\identifier{f})$: then for each $T_j$ there is a
  clause $\rho_j\colon \apps{\identifier{f}}{\ell_1}{\ell_n} = s$
  which imposes $\eta^{(j)}$ such that the immediate subtree of $T_j$
  is $\eprog,\eta^{(j)} \vvdash o^{(j)}$.

  Now, let $\ell : \atype$ be a linear pattern, $\eta$ an
  ext-environment and $e,u \in \interpret{\atype}$ be such that
  $\consistent{e}{u}$ and $\ell\eta = e$.  By a simple induction on
  the form of $\ell$ we find an ext-environment $\eta'$ on domain
  $\Var(\ell)$ such that $\ell\eta' = u$ and $\consistent{\eta(x)}{
  \eta'(x)}$. 
  
  Thus, the first matching clause $\rho_j$ is necessarily the same for
  all $T_j$, and we have $\consistent{\eta^{(j)}(x)}{\eta^{(k)}(x)}$
  for all $j,k,x$.
  For all $1 \leq i \leq n$ and $1 \leq j \leq N$, we have
  $e_i^{(j)} = \ell_i\eta^{(j)}$.  Another simple induction on
  $\ell_i$ proves that
  we can find $\eta$ with each $\eta(x) \sqsupseteq \sqcup \{ \eta^{(
  j)}(x) \mid 1 \leq j \leq N \}$ such that $e_i = \ell_i\eta$.
  
  The induction hypothesis gives $\eprog,\eta \vvdash \too o$, so
  $\apps{\identifier{f}}{e_1}{e_n} \too o$ by [Call].

\item $n < \arity(\identifier{f})$: each of the trees $T_j$ is derived
  by [Value].  Write $o = O_\atype$ and $o^{(j)} = O^{(j)}_\atype$ for
  $1 \leq j \leq N$.  We are done by [Value] if $\eprog \vvdashcall
  \app{\apps{\identifier{f}}{e_1}{e_n}}{e_{n+1}} \too o'$ for all
  $(e_{n+1},o') \in O$.

  Since $o = \sqcup \{ o^{(1)},\dots,o^{(N)} \}$, we can write
  $o' = \sqcup Y_{e_{n+1}}$ and identify a non-empty set
  $\mathit{Pairs}_{e_{n+1}} = \{ (e,u) \in O^{(1)} \cup \dots \cup
  O^{(N)} \mid e_{n+1} \sqsupseteq \sqcup \{ e \} \}$ such that
  $Y_{e_{n+1}} = \{ u \mid (e,u) \in \mathit{Pairs}_{e_{n+1}} \}$.

  For each element $(e,u)$ of $\mathit{Pairs}_C$, some $T_j$ has a
  subtree with root $\eprog \vvdashcall
  \app{\apps{\identifier{f}}{e_1^{(j)}}{e_n^{(j)}}}{e} \too u$.  Let
  $\mathit{Trees}_{e_{n+1}}$ be the corresponding set of trees, and
  note that all trees in $\mathit{Trees}_{e_{n+1}}$ have a strictly
  smaller depth than the $T_j$ they originate from, so certainly
  smaller than the maximum depth.

  Now, for $1 \leq i \leq n+1$, let $\mathit{Args}_i := \{$ argument
  $i$ of the root of $T \mid T \in \mathit{Trees}_{e_{n+1}} \}$.
  We observe that:
  \begin{itemize}
  \item for $1 \leq i \leq n$: $e_i \sqsupseteq \sqcup
    \mathit{Args}_i$: we have $\mathit{Args}_i \subseteq \{ e_i^{(1)},
    \dots, e_i^{(N)} \}$, so by Lemma~\ref{lem:sqcupsup},
    $e_i \sqsupseteq \sqcup \{ e_i^{(j)} \mid 1 \leq j \leq N \}
    \sqsupseteq \sqcup \mathit{Args}_i$, which suffices by
    transitivity (Lemma~\ref{lem:sqsuptrans});
  \item $e_{n+1} \sqsupseteq \sqcup \mathit{Args}_{j+1}$: $e_{n+1}
    \sqsupseteq \{e\}$ for all $e \in \mathit{Args}_{j+1}$, so this
    is given by Lemma~\ref{lem:pintergreater}.
  \item $o' = \sqcup Y_{e_{n+1}} = \sqcup \{$right-hand sides of the
    roots of $\mathit{Trees}_{e_{n+1}} \}$.
  \end{itemize}
  Therefore $\eprog \vvdashcall \app{\apps{\identifier{f}}{e_1}{e_n
  }}{e_{n+1}} \too o'$ by the induction hypothesis as required.
\end{itemize}
For the second case, consider the form of $s$.
\begin{itemize}
\item $s = \apps{\identifier{c}}{s_1}{s_m}$ with $\identifier{c} \in
  \Constructors$: then each $o^{(i)} = s\eta^{(i)} \in \B$, so
  $o = o^{(1)} = \dots = o^{(N)}$ and---since the variables in $s$
  all have order $0$---we have $\eta(x) = \eta^{(1)}(x) = \dots =
  \eta^{(N)}(x)$ for all relevant $x$.  Thus also $o = s\eta$ and
  we complete by [Constructor].
\item $s = (s_1,s_2)$; each tree $T_j$ has two immediate subtrees:
  one with root $\eprog,\eta^{(j)} \vvdash s_1 \too o^{(j)}_1$ and
  one with root $\eprog,\eta^{(j)} \vvdash s_2 \too o^{(j)}_2$, where
  $o^{(j)} = (o^{(j)}_1,o^{(j)}_2)$.

  We can write $o = (o_1,o_2)$ where $o_1 = \sqcup \{ o^{(j)}_1
  \mid 1 \leq j \leq N \}$ and $o_2 = \sqcup \{ o^{(j)}_2 \mid 1 \leq
  j \leq N \}$, and as the induction hypothesis for both subtrees
  gives $\eprog,\eta \vvdash s_1 \too o_1$ and $\eprog,\eta \vvdash
  s_2 \too o_2$ respectively, we conclude $\eprog,\eta \vvdash s \too
  o$ by [Pair].
\item $s = \ifte{s_1}{s_2}{s_3}$: for each tree $T_j$, the first
  subtree has the form $\eprog,\eta^{(j)} \vvdash s_1 \too \strue$
  or $\eprog,\eta^{(j)} \vvdash s_1 \too \sfalse$; by consistency of
  derivation trees (Lemma~\ref{lem:consistency}), either $\strue$ or
  $\sfalse$ is chosen for \emph{all} these subtrees.  We assume the
  former; the latter case is symmetric.

  By the induction hypothesis for this first subtree, $\eprog,\eta
  \vvdash s_1 \too \strue = \sqcup \{ \strue,\dots,\strue \}$ as
  well.

  The second immediate subtree of all trees $T_j$ has a root of the
  form $\eprog,\eta^{(j)} \vvdash s_2 \too o^{(j)}$.  By the
  induction hypothesis for this second subtree, $\eprog,\eta \vvdash
  s_2 \too o$.
  Thus we conclude $\eprog,\eta \vvdash s \too o$ by [Cond-True].
\item $s = \apps{x}{s_1}{s_n}$ with $x \in \V$: each of the trees
  $T_j$ has $n$ subtrees of the form $\eprog,\eta^{(j)} \vvdash s_1
  \too e_i^{(j)}$ (for $1 \leq i \leq n$); by the induction
  hypothesis, we have $\eprog,\eta \vvdash s_i \too e_i$, where
  $e_i = \sqcup \{ e_i^{(j)} \mid 1 \leq j \leq N \}$.  But then:
  \begin{itemize}
  \item $\pinterpret{\atype_1 \arrtype \dots \arrtype \atype_n \arrtype
    \btype} \ni \eta(x) \sqsupseteq \sqcup \{ \eta^{(1)}(x),\dots,
    \eta^{(N)}(x) \}$;
  \item $\pinterpret{\atype_i} \ni e_i = \sqcup \{ e_i^{(1)},\dots,
    e_i^{(N)} \}$ for $1 \leq i \leq n$;
  \item there are $u^{(j)}$ such that $\eta^{(j)}(e_1^{(j)},\dots,
    e_n^{(j)}) \ni u^{(j)} \sqsupseteq o^{(j)}$ for $1 \leq j \leq
    N$;
  \item $o = \sqcup \{ o^{(1)},\dots,o^{(N)} \}$.
  \end{itemize}
  By Lemma~\ref{lem:apply}, there exists $u \in o(e_1,\dots,e_n)$
  such that $u \sqsupseteq o$.  We conclude $\eprog,\eta \vvdash s
  \too o$ by [Variable].
\item $s = \apps{\identifier{f}}{s_1}{s_n}$ with $n \leq \arity(
  \identifier{f})$: then necessarily each $\eprog,\eta^{(j)} \vvdash
  s \too o^{(j)}$ follows by [Func].  Thus, for $1 \leq j \leq N$
  there are $e_1^{(j)},\dots,e_n^{(j)}$ such that:
  \begin{itemize}
  \item $\eprog,\eta^{(j)} \vvdash s_i \too e_i^{(j)}$ for $1 \leq i
    \leq n$ and
  \item $\eprog \vvdashcall \apps{\identifier{f}}{e_1^{(j)}}{e_n^{(j)
    }} \too o^{(j)}$.
  \end{itemize}
  Now, clearly each set $\{ e_i^{(j)} \mid 1 \leq j \leq N \}$ is
  consistent by the simple fact that there are derivation trees for
  them: this is the result of Lemma~\ref{lem:consistency}.
  Defining $e_i := \sqcup \{ e_i^{(1)},\dots,e_i^{(N)} \}$ for $1 \leq
  j \leq N$, the induction hypothesis gives that $\eprog,\eta \vvdash
  s_i \too e_i$, and that $\eprog \vvdashcall \apps{\identifier{f}}{
  e_1}{e_n} \too o$, all by functional trees.  We complete with
  [Func].
\item $s = \apps{\identifier{f}}{s_1}{s_n}$ with $n > k :=
  \arity(\identifier{f})$: then there are
  $e_1^{(j)},\dots,e_n^{(j)},u^{(j)},c^{(j)}$ such that for all $1
  \leq j \leq N$:
  \begin{itemize}
  \item tree $T_j$ has subtrees $\eprog,\eta^{(j)} \vvdash s_i \too
    e_i^{(j)}$ for $1 \leq i \leq n$;
  \item tree $T_j$ has a subtree $\vvdashcall \apps{\identifier{f}}{
    e_1^{(j)}}{e_k^{(j)}} \too u^{(j)}$;
  \item $u^{(j)}(e_{k+1}^{(j)},\dots,e_n^{(j)}) \ni c^{(j)}
    \sqsupseteq o^{(i)}$.
  \end{itemize}
  Therefore, by the induction hypothesis and Lemma~\ref{lem:apply},
  we can identify $e_1,\dots,e_n,u,c$ such that:
  \begin{itemize}
  \item $e_i = \sqcup \{e_i^{(1)},\dots,e_i^{(N)}\}$ and
    $\eprog,\eta \vvdash s_i \too e_i$ for $1 \leq i \leq n$;
  \item $\eprog \vvdashcall \apps{\identifier{f}}{e_1}{e_k} \too u =
    \sqcup \{ u^{(1)},\dots,u^{(N)} \}$;
  \item $u(e_{k+1},\dots,e_n) \ni c \sqsupseteq o$.
  \end{itemize}
  Therefore $\eprog,\eta \vvdash s \too o$ by [Apply].
  \qed
\end{itemize}
\end{proof}

With this, the one remaining lemma---completeness of
Algorithm~\ref{alg:base}---is trivial.

\oldcounter{\dcompletenesslem}
\begin{lemma}
If $\progresult$ and $\prog$ is deterministic, then
Algorithm~\ref{alg:base} returns a set $A \cup \{b\}$.
\end{lemma}
\startappendixcounters

\begin{proof}
Suppose $\progresult$ for a deterministic program $\prog$.
By Lemma~\ref{lem:maketree}, we can derive $\vvdashcall
\apps{\symb{start}}{d_1}{d_M} \too b$.  By Lemma~\ref{lem:dettree},
this can be derived by a tree which only uses deterministic
extensional values.  By Lemma~\ref{lem:usetree},
Algorithm~\ref{alg:base} therefore returns a set containing $b$.
\qed
\end{proof}

\section{Arrow depth and unitary variables (Section~\ref{sec:nopartialvar})}
\label{app:arrowdepth}

In Sections~\ref{sec:deterministic} and~\ref{sec:elementary}, we have
demonstrated two things:
\begin{itemize}
\item that cons-free deterministic programs of data order $K$ characterise
  $\exptime{K}$
\item that cons-free non-deterministic programs of data order $K > 0$
  characterise $\elementary$
\end{itemize}

However, most of the proof effort has gone towards the complexity and
correctness of the simulation algorithm---arguably the least interesting side,
since the characterisation result for deterministic programs is a natural
extension of an existing result of~\cite{jon:01}, while the surprising result
for non-deterministic programs is that we get \emph{at least} $\elementary$,
not that we cannot go beyond.

The efforts pay off, however, when we consider what is needed to recover the
original hierarchy.  The proofs require very little adaptation to obtain
Theorems~\ref{thm:arrowdepth} and~\ref{thm:unitary}.  We start with
Theorem~\ref{thm:arrowdepth}, which we split up in its two parts.

\begin{lemma}\label{lem:arrowdepth:simulate}
Every decision problem in $\exptime{K}$ is accepted by a deterministic
cons-free program with data arrow depth $K$.
\end{lemma}

Here, a program has \emph{data arrow depth $K$} if all variables are typed with
a type of arrow depth $K$.

\begin{proof}
Both Lemma~\ref{lem:module:pol} and~\ref{lem:module:exp} also apply if ``data
order $K$'' is replaced by ``data arrow depth $K$''.  With this observation, we
may copy the proof of Lemma~\ref{lem:deterministic:simulate}.
\qed
\end{proof}

\begin{lemma}\label{lem:arrowdepth:algorithm}
Every decision problem accepted by a deterministic cons-free program
$\prog$ with data arrow depth $K$ is in $\exptime{K}$.
\end{lemma}

\begin{proof}
Defining a type to be ``proper'' if its arrow depth is smaller than
$K$, types of order $0$ are proper and $\atype \times \btype$ is
proper if{f} both $\atype$ and $\btype$ are.  Thereofre, all the
proofs in Appendix~\ref{app:properness} extend to arrow depth, and we
immediately obtain a variation of Lemma~\ref{lem:proper} where data
order is replaced by data arrow depth.  Note that
Lemma~\ref{lem:proper} considers the transformation from $\prog$ to
$\eprog$ in both algorithms.

Now, in Algorithm~\ref{alg:general}, alter step~\ref{alg:prepare:eprog} by
using the transformation which considers arrow depth rather than data order,
and in step~\ref{alg:prepare:statements}, only include statements
$\apps{\identifier{f}}{e_1}{e_n} \leadsto o$ if $\mathit{depth}(\atype_{n+1}
\arrtype \dots \arrtype \atype_m \arrtype \asortorpair') \leq K$ (rather than
considering $\typeorder{}$).  This does not affect correctness of the algorithm,
as is easily checked by going over the proofs of
Appendix~\ref{app:correctness}: the only place in the algorithm where it may
be important whether any statements $\apps{\identifier{f}}{e_1}{e_n} \leadsto
o$ were removed is step~\ref{alg:iterate:rhs:func}, but here only calls with an
output arrow depth $\leq K$ may be used (due to the preparation
step~\ref{alg:prepare:eprog} and the altered Lemma~\ref{lem:proper}).

Moreover, by the combination of Lemma~\ref{lem:complexitycore} (which also
applies to the thus modified algorithm) and Lemma~\ref{lem:interpretcard}, this
altered algorithm finds the possible results of $\prog$ on given input in
$\timecomp{a \cdot \exp_2^K(n^b)}$ for some $a,b$.  Therefore, any
decision problem accepted by $\prog$ is in $\exptime{K}$.
\qed
\end{proof}

We thus conclude:

\setcounter{theorem}{\arrowdepththm}
\begin{theorem}
The class of non-deterministic cons-free programs where all
variables are typed with a type of arrow depth $K$ characterises
$\exptime{K}$.
\end{theorem}

\begin{proof}
By the combination of Lemmas~\ref{lem:arrowdepth:simulate}
and~\ref{lem:arrowdepth:algorithm}.
\qed
\end{proof}

We turn to Theorem~\ref{thm:unitary}, which considers programs with unitary
variables.  Again, one direction---the minimum power of such programs---is
quite simple:

\begin{lemma}\label{lem:unitary:simulate}
Every decision problem in $\exptime{K}$ is accepted by a deterministic
cons-free program with data arrow depth $K$ and unitary variables.
\end{lemma}

\begin{proof}
All variables employed in both Lemma~\ref{lem:module:pol} and
Lemma~\ref{lem:module:exp} have a type that is either a sort, or has the
form $\atype \arrtype \bool$; thus, the simulation program is unitary, and we
may copy the proof of Lemma~\ref{lem:deterministic:simulate}.
\qed
\end{proof}

The second part of Theorem~\ref{thm:unitary} can once more be derived
using a variation of Algorithm~\ref{alg:general}.  However, here we
must be a little careful: where both data order and arrow depth are
\emph{recursive} properties, the property that a type is ``unitary''
(i.e., of the form $\asortorpair$ or $\atype \arrtype \asortorpair$
with $\typeorder{\asortorpair} = 0$) is not recursive.  Thus, a
unitary type of a fixed data order may still have an arbitrarily high
arrow depth.
We circumvent this problem by altering unused subtypes.

\begin{lemma}\label{lem:unitary:algorithm}
Every decision problem accepted by a deterministic cons-free program
$\prog$ with data order $K$ and unitary variables is in $\exptime{K}$.
\end{lemma}

\begin{proof}
Let a type $\atype$ be \emph{proper} if $\typeorder{\atype} \leq K$
and (a) $\typeorder{\atype} = 0$ or (b) $\atype$ has the form $\btype
\arrtype \asortorpair$ with $\typeorder{\asortorpair} = 0$ or (c)
$\atype$ has the form $\atype_1 \times \atype_2$ with both $\atype_1$
and $\atype_2$ proper.  This notion of properness has the properties
described in Definition~\ref{def:propertype}, so the proofs in
Appendix~\ref{app:properness} extend; $\prog$ can be transformed into
a program $\eprog$ with data order $K$ and unitary variables such
that for all clauses $\apps{\identifier{f}}{\ell_1}{\ell_k} = s$:
\pagebreak
all sub-expressions $t$ of $s$ have a unitary type with order $\leq
K$.

Now let $\mathit{fixtype}$ be defined as follows:
\begin{itemize}
\item $\mathit{fixtype}(\asort) = \asort$ for $\asort \in \Sorts$
\item $\mathit{fixtype}(\atype \times \btype) =
  \mathit{fixtype}(\atype) \times \mathit{fixtype}(\btype)$
\item $\mathit{fixtype}(\atype_1 \arrtype \dots \arrtype \atype_n
  \arrtype \asortorpair) = \mathit{fixtype}(\atype_1) \arrtype
  \asortorpair$ if $\typeorder{\asortorpair} = 0$ and $n > 0$.
\end{itemize}

Then clearly $\mathit{depth}(\mathit{fixtype}(\atype)) \leq K$
whenever $\typeorder{\atype} \leq K$.  Given type assignments $\F$
(for defined symbols and data constructors) and $\Gamma$ for
variables, let $\F' := \{ \identifier{f} : \mathit{fixtype}(\atype_1)
\arrtype \dots \arrtype \mathit{fixtype}(\atype_m) \arrtype
\mathit{fixtype}(\asortorpair) \mid \identifier{f} : \atype_1
\arrtype \dots \arrtype \atype_m \arrtype \asortorpair \in \F \}$ and
$\Gamma' := \{ x : \mathit{fixtype}(\atype_1) \arrtype \dots \arrtype
\mathit{fixtype}(\atype_m) \arrtype \mathit{fixtype}(\asortorpair)
\mid x : \atype_1 \arrtype \dots \arrtype \atype_m \arrtype
\asortorpair \in \Gamma \}$.
Now suppose that all clauses in $\eprog$ are well-typed under $\F'$
and the corresponding type environment $\Gamma'$.
Since typing does not affect the semantics of
Figure~\ref{fig:evaluation}, $\progresult$ if and only if
$\progeval{\eprog}{d_1,\dots,d_M} \mapsto b$ still holds.  Using
Lemma~\ref{lem:arrowdepth:algorithm}---and the observation that the
translation from $\prog$ to $\eprog$ takes constant time as it does
not consider the input $d_1,\dots,d_M$---any decision problem
accepted by $\prog$ is therefore in $\exptime{K}$.

It remains to be seen that every clause $\apps{\identifier{f}}{
\ell_1}{\ell_k} = s$ which is well-typed under $\F$ with type
environment $\Gamma$ is also well-typed using $\F'$ and $\Gamma'$
instead.
To see this, we prove the following by induction on the size of $s$:

\emph{Suppose variables have a proper type, and let $s : \atype$
using $\F,\Gamma$.
If for all $t \subtermeq s$, the type of $t$ is proper w.r.t.\ 
$\F,\Gamma$ and $t$ does not have the form $\apps{(\ifte{b}{s_1}{s_3}
)\linebreak
}{t_1}{t_n}$ or $\apps{(\apps{\choice}{s_1}{s_m})}{t_1}{t_n}$ with $n
> 0$, then $s : \mathit{fixtype}(\atype)$ using $\F',\Gamma'$.}

\begin{itemize}
\item If $s = \apps{\choice}{s_1}{s_m}$, then each $s_i : \atype$
  using $\F,\Gamma$, so by the induction hypothesis each $s_i :
  \mathit{fixtype}(\atype)$ using $\F',\Gamma'$; this gives
  $s : \mathit{fixtype}(\atype)$ following the typing rules for
  $\choice$.
\item If $s = \ifte{b}{s_1}{s_2}$, then by the induction hypothesis
  (and using that $\mathit{fixtype}(\bool) = \bool$),
  $b : \bool$ and both $s_1 : \mathit{fixtype}(\atype)$ and
  $s_2 : \mathit{fixtype}(\atype)$.
\item If $s = \apps{\identifier{c}}{s_1}{s_m}$, then $\atype \in
  \Sorts$ and we can write $\identifier{c} : \asortorpair_1
  \arrtype \dots \arrtype \asortorpair_m \arrtype \atype \in \F
  \cap \F'$ where each $\asortorpair_i$ has type order $0$.  By the
  induction hypothesis, each $s_i : \mathit{fixtype}(\asortorpair_i)
  = \asortorpair_i$ using $\F',\Gamma'$.
\item If $s = \apps{a}{s_1}{s_n}$ with $a \in \V \cup \Defineds$, then
  $a$ is typed with $\btype_1 \arrtype \dots \arrtype \btype_n \arrtype
  \atype$ in $\F \cup \Gamma$.  Since $s$ has a proper type, we know
  that $\atype$ has the form $\asortorpair$ or $\ctype \arrtype
  \asortorpair$ or $\ctype_1 \times \ctype_2$; therefore
  $a : \mathit{fixtype}(\btype_1) \arrtype \dots \arrtype
  \mathit{fixtype}(\btype_n) \arrtype \mathit{fixtype}(\atype) \in
  \F' \cup \Gamma'$.
  Since each $s_i : \mathit{fixtype}(\btype_i)$ by the induction
  hypothesis, we obtain $s : \mathit{fixtype}(\atype)$.
\end{itemize}
Applying the result also to the $\apps{\identifier{f}}{\ell_1}{
\ell_k}$, the entire clause is well-typed.
\qed
\end{proof}

\end{document}